%% file: perc_EC.tex
\documentclass[leqno,11pt]{article}

\usepackage{amsmath,amsthm,amssymb,mathrsfs}
\usepackage{times}
\usepackage{fancybox,fancyhdr,graphics,epsfig}
\usepackage[usenames,dvipsnames]{color}
\usepackage{bbm,caption,subcaption}
\usepackage{verbatim}
\usepackage{tikz-cd}
\usepackage{float}

\def\c{\mathrm{c}}
\input{defs}
\theoremstyle{empty}

\title{Homological Percolation and the Euler Characteristic}
\author{Omer Bobrowski	\thanks{Technion -- Israel Institute of Technology, \emph{omer@ee.technion.ac.il}}\and Primoz Skraba\thanks{Queen Mary University of London, \emph{p.skraba@qmul.ac.uk}}}

\date{\today}
\begin{document}
\maketitle

\begin{abstract}
In this paper we study the connection between the phenomenon of homological percolation  (the formation of ``giant" cycles in persistent homology), and the zeros of the expected Euler characteristic curve. We perform an experimental study that covers four different models: site-percolation on the cubical and permutahedral lattices, the Poisson-Boolean model, and Gaussian random fields. All the models are generated on the flat torus $\T^d$, for $d=2,3,4$. The simulation results strongly indicate that the zeros of the expected Euler characteristic curve approximate the critical values for homological-percolation. Our results also provide some insight about the approximation error. Further study of this connection could have powerful implications both in the study of percolation theory, and in the field of Topological Data Analysis.
\end{abstract}

\section{Introduction}
This paper aims to make a connection between two seemingly unrelated mathematical topics, in the context of spatial stochastic processes. The first is a large-scale phenomenon we refer to as \emph{homological percolation}, where \emph{giant cycles} are formed.
The second is
an integer-valued topological invariant, known as the \emph{Euler characteristic} (EC).

To describe these two topics and the connection between them, we will use the language of \emph{persistent homology} \cite{edelsbrunner_persistent_2008,edelsbrunner_topological_2002,edelsbrunner_persistent_2014,zomorodian_computing_2005}, which is main workhorse in the field of Topological Data Analysis (TDA) \cite{carlsson_topology_2009,ghrist_barcodes:_2008}.
Persistent homology is an algebraic-topological functional that is applied to \emph{filtrations} (nested sequences) of topological spaces.  It is essentially a tool that tracks the formation and destruction of topological features such as connected components (0-cycles), holes (1-cycles), cavities (2-cycles), and their higher dimensional analogues (non-trivial $k$-cycles).
Persistent homology has been demonstrated to be a powerful tool in the analysis of various types of data (e.g.~neuroscience \cite{singh_topological_2008}, cosmology \cite{adler_modeling_2017}, and complex networks \cite{horak_persistent_2009}).

Suppose that $X$ is a ``nice" topological space, and that we have a filtration $\set{X_t}_{t\in \R}$ of spaces such that $X_s \subset X_t\subset X$ for all $s < t$.  We can classify the cycles captured by persistent homology into two groups:
the group of ``giant" cycles will be those that are also nontrivial cycles (holes) in $X$, while all other cycles (that are trivial in $X$) will be considered ``small". In a data-analytic language we can think of the giant cycles as the ``topological signal" hidden in the filtration, as they capture information about the underlying space. The small cycles are considered as nuisance ``noise" one might wish to filter, in order to reveal the signal. See Figure \ref{fig:ph} below. One of the main challenges in TDA is to identify for a given persistent-homology, which feature belongs to which group.

The phenomenon we refer to as \emph{homological percolation} describes the appearance (or birth) of the giant $k$-cycles. From the theoretical-probabilistic perspective, this study is at a very early stage. However, relying on  classical results in percolation theory together with recent simulations (including in this paper), it is conjectured that homological percolation occurs as a sharp phase transition. In other words, given a random filtration $\set{X_t}$ the probability for creating the giant cycles switches from zero to one as a result of an infinitesimal increase in the filtration parameter $t$. Moreover, for a fixed $k$ the thresholds for all giant $k$-cycles coincide, and the critical values are increasing in $k$. In other words, if $t^{\text{perc}}_k$ is the critical value for the emergence of the giant $k$-cycles, then $t^{\text{perc}}_k \le t^{\text{perc}}_{k'}$ for all $k<k'$. These conjectures summarize the first half of the story.%\footnote{\textcolor{red}{Do we want to mention coupling?} [Omer] I'm not sure what you mean by that...}

The second half of the story is about a rather different mathematical object. The \emph{Euler characteristic} (EC) is an integer-valued topological invariant that can be assigned to a topological space. Quite remarkably, the EC can be defined in  several different ways, vastly different in nature (e.g.~geometric, combinatorial, topological, analytic), which are all equivalent under quite general conditions (local compactness). There are many ways to compute the EC.  For example, we can compute the EC by counting the number of cells in a cell-complex, counting the critical points of a Morse function, or integrating the Gaussian curvature of a manifold.  For our purposes, we use the ``homological" definition of the EC, i.e.
$$\chi(X) := \sum_k (-1)^k \beta_k(X),$$
where $X$ is a topological space. The integer numbers $\beta_k(X)$ are called the \emph{Betti numbers}, which count the number of $k$-dimensional non-trivial cycles in $X$. The EC is a very interesting mathematical object \cite{adler_random_2007,stoyan_stochastic_1987}, and over the years it was also found to be very useful as a  statistical tool. Two areas of applications where the EC was proven to be quite powerful are cosmology \cite{colley1996topology,kogut1996tests,worsley_boundary_1995} and brain imaging \cite{taylor_detecting_2007,worsley_testing_2001}.
In the random setting, somewhat surprisingly, much more is known about the distribution of the EC for a random space $X$ (as we see in Section \ref{sec:models}) compared to the individual Betti numbers defining it.

Given a topological space $X$ and a filtration $\set{X_t}_{t\in \R}$, one may calculate its \emph{EC curve} $\chi(t) := \chi(X_t)$. For several random filtrations, such as the ones discussed in this paper, the expected value $\mean{\chi(t)}$ has been analyzed in the past \cite{bobrowski_vanishing_2017,taylor_gaussian_2009}. Somewhat surprisingly, in all the models we discuss here, as well as many others, while the EC curves look completely different, it is always the case that the expected EC curve has exactly $(d-1)$ zeros (where $d$ is the dimension of the generating model). Denote these zeros by $t^{\text{ec}}_1,\ldots,t^{\text{ec}}_{d-1}$, and recall that $t_1^{\text{perc}},\ldots, t_{d-1}^{\text{perc}}$ are the homological-percolation thresholds.
The question we wish to pursue in this paper is the then following:
\begin{center}
	\fbox{
	\parbox{0.5\textwidth}{
	\centering
Is there a connection between $t^{\text{perc}}_k$ and $t^{\text{ec}}_k$?}}
\end{center}

A priori, there is no obvious reason why such a connection should exist. Indeed, both homological percolation and the EC curve are related to the homology of a given filtration. However, homological percolation describes  the giant cycles formed, while the EC contains information about the total number of cycles, regardless of their size. In addition, the EC curve is a quantitative descriptor while percolation is a qualitative phenomenon.
Contradicting this intuition, our main goal in this paper is to argue that such a connection exists, and is potentially universal in the sense that it occurs across vastly different stochastic models.  We note that we are not aiming to provide any analytic statements here. Instead, we want to suggest that this link exists by presenting simulation results for several random systems.

In this paper we consider four percolation models: site percolation on a cubical grid, site percolation on a permutahedral grid, continuum percolation model, and sub-level sets of Gaussian random fields. In all these models an explicit formula for the expected EC curve can be calculated. We simulate these models on the $d$-dimensional \emph{flat torus} (i.e.~a $d$-dimensional box with periodic boundary conditions), and compare the critical percolation values to the zeros of the expected EC curve.

\paragraph{Main results.} The simulations we present in Section \ref{sec:sim} highly suggest a positive answer to the question above. In all models and all dimensions tested, the simulations indicate that $t^{\text{perc}}_k \approx t^{\text{ec}}_k$, where determining the exact meaning of ``$\approx$" remains future work. Note that all the models we study depend on a parameter $n$ (either grid-size, or number of points). Defining $\Delta_k := (t^{\text{perc}}_k - t^{\text{ec}}_k)$, it would be tempting to conjecture that $\Delta_k\xrightarrow{n\to\infty} 0$. However, our simulation results indicate that while the difference converges, the limit might be nonzero. If the model is symmetric with respect to the parameter $t$ (e.g. the permutahedral complex and a zero-mean Gaussian field), and if $d$ is even, then  our simulations as well as analytical arguments show that indeed $\Delta_{d/2} = 0$.
A second-order observation we make from the simulations is that the sign of the error term $\Delta_k$ is not arbitrary. For $k < d/2$ it seems that we always have $\Delta_k < 0$, while for $k>d/2$ we have $\Delta_k > 0$.

In addition to the interesting and surprising mathematical phenomenon that we reveal here, there are also potential  applied aspects to the conjectures we make in this paper.
In most models in statistical physics, the exact percolation thresholds are not known. At best there exist some theoretical bounds, or numerical approximations.
However, since the zeros of the expected EC curves can be found in many cases, our hope is that these could be used them as an improved approximation for the real percolation thresholds. In addition to probability and statistical physics, we also believe that these results can have implications in TDA, for example by enhancing the detection of significant topological features in data.

We should note that the connection between percolation thresholds and the EC was studied previously in \cite{neher_topological_2008}, and certainly inspired the current work. There, the authors focused on classical percolation (the formation of giant connected components) which is a special case of the higher-dimensional homological notion we consider here. Thus, there is only a single threshold to consider. In addition, the models considered in \cite{neher_topological_2008} are two and three dimensional, whereas here we wish to argue that this phenomenon occurs in all dimensions, and across various types of percolation models.

\section{Preliminaries}\label{sec:homology}

In this section we wish to provide a rather non-formal introduction to the fundamental terminology we will be using in this paper.  We  include references for  more formal treatment of the topics discussed.

%%%%%
\subsection{Homology}
%%%%%

Homology is an algebraic-topological structure that describes various types of topological phenomena in topological spaces using algebraic structures.
Let $X$ be a topological space. In this paper we will use homology with field coefficients. In this case, the $k$-th homology $\Hg_k(X)$ is a vector space, such that its basis elements correspond to the following features. The basis of $\Hg_0(X)$ corresponds to the connected components of $X$ (also referred to as $0$-cycles), $\Hg_1(X)$ to ``holes'' in $X$ ($1$-cycles), $\Hg_2(X)$ to ``voids'' or ``bubbles'' in $X$ ($2$-cycles), and more generally -- $\Hg_k(X)$ represents $k$-cycles, that can be thought of as shapes similar to a $k$-sphere.
The Betti numbers are the corresponding dimensions $\beta_k(X) := \dim \Hg_k(X)$, that count the number of nontrivial $k$-cycles in $X$. In Figure \ref{fig:betti} we present a few examples for spaces along with their Betti numbers. There are many excellent introductions to homology theory, for example, we refer the reader to  \cite{hatcher_algebraic_2002}.

\begin{figure}[ht]
\centering
\includegraphics[width = 0.7\textwidth]{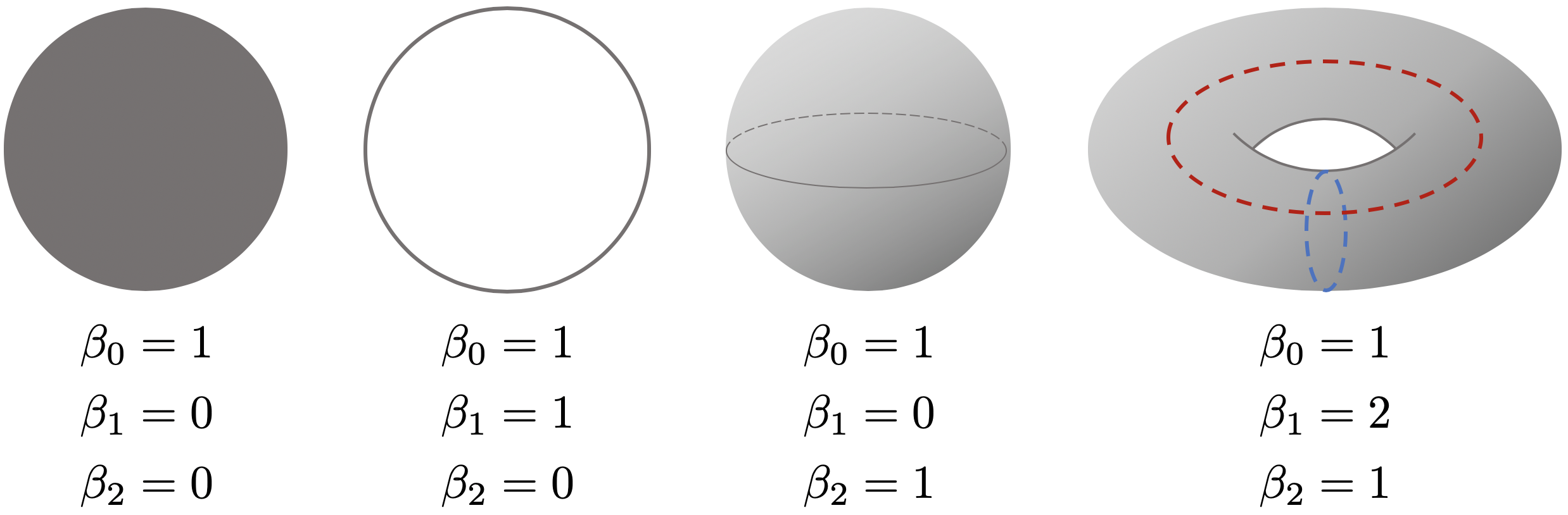}
\caption{\label{fig:betti} Example of simple topological spaces with their Betti numbers. From left to right - a disc, a circle, a 2-dimensional sphere, and a 2-dimensional torus. For the torus, we marked the two $1$-cycles in dashed lines.}
\end{figure}

In addition to describing the topology of a single space $X$, the language of homology also provides means to  match $k$-cycles between two spaces. Let $X,Y$ be topological spaces, and let $f:X\to Y$ be a continuous function. Then for every $k$ there exists a corresponding linear transformation called the \emph{induced map} $f_*:\Hg_k(X)\to \Hg_k(Y)$ that maps $k$-cycles in $X$ into $k$-cycles in $Y$.

\begin{figure}[ht]
\centering
\begin{subfigure}[t]{0.25\textwidth}
\centering
\includegraphics[height=\textwidth]{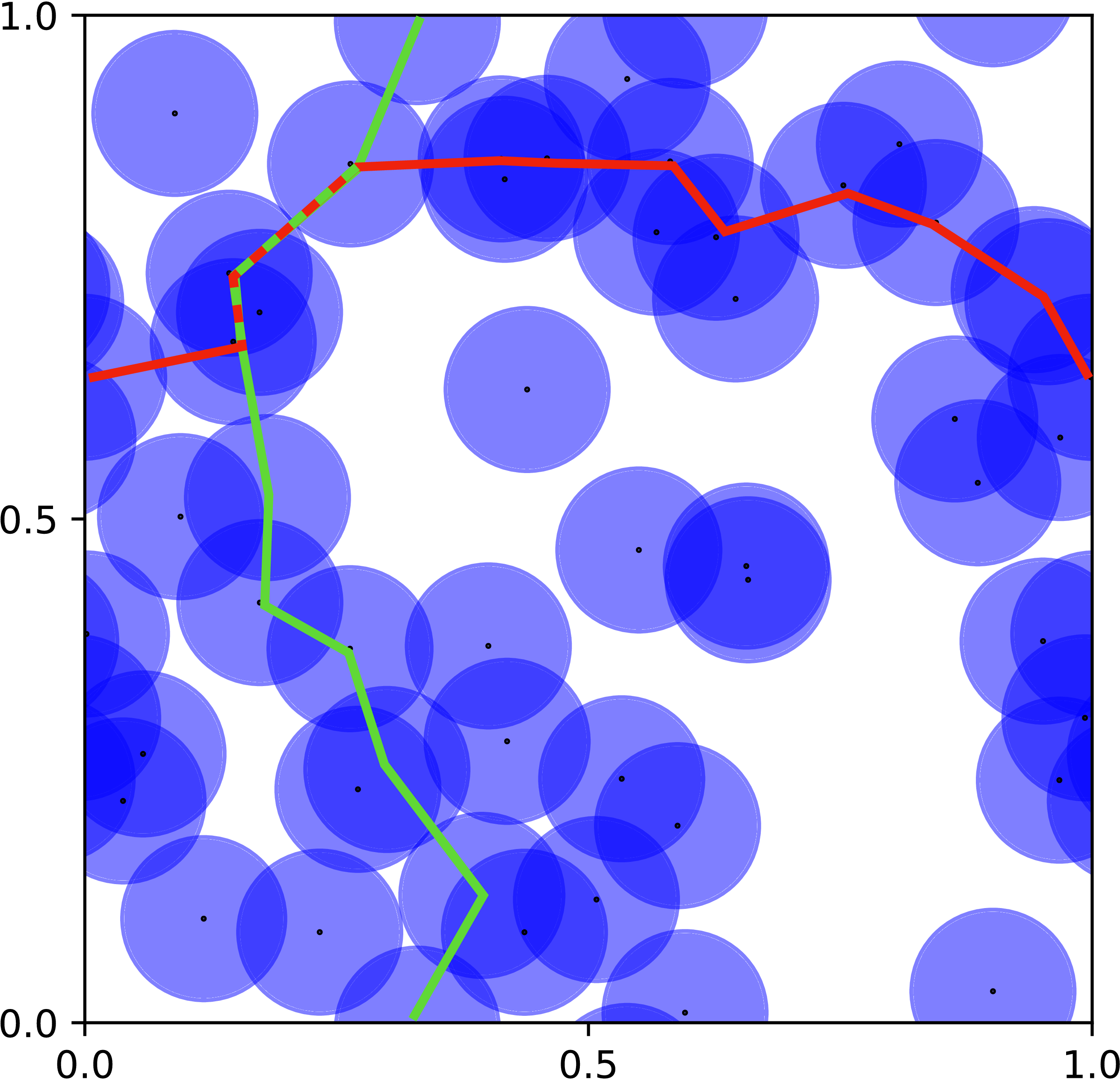}
\caption{}
\end{subfigure}
\hfill
\begin{subfigure}[t]{0.4\textwidth}
\centering
\includegraphics[height=0.6\textwidth]{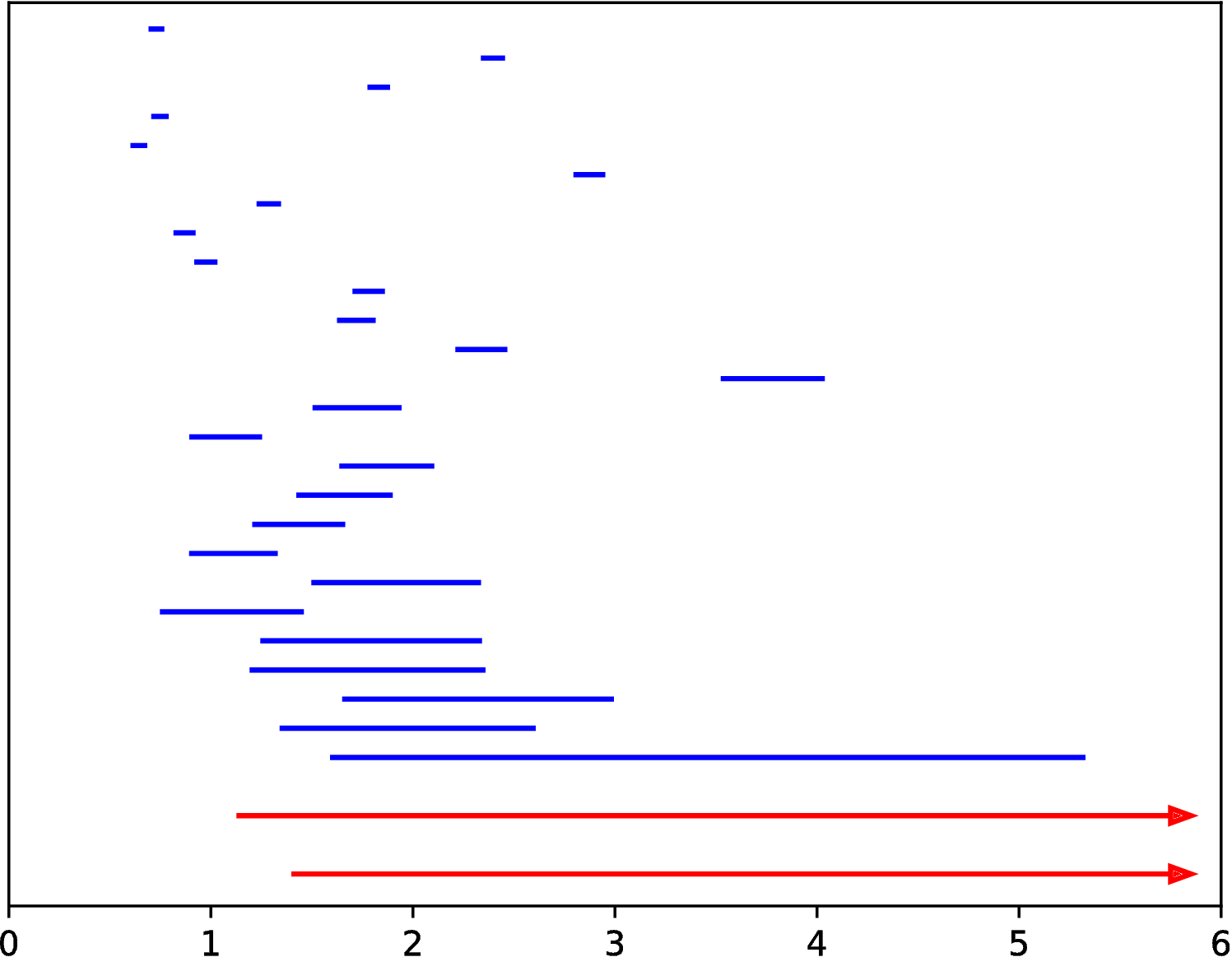}
\caption{}
\end{subfigure}
\hfill
\begin{subfigure}[t]{0.25\textwidth}
\centering
\includegraphics[height=\textwidth]{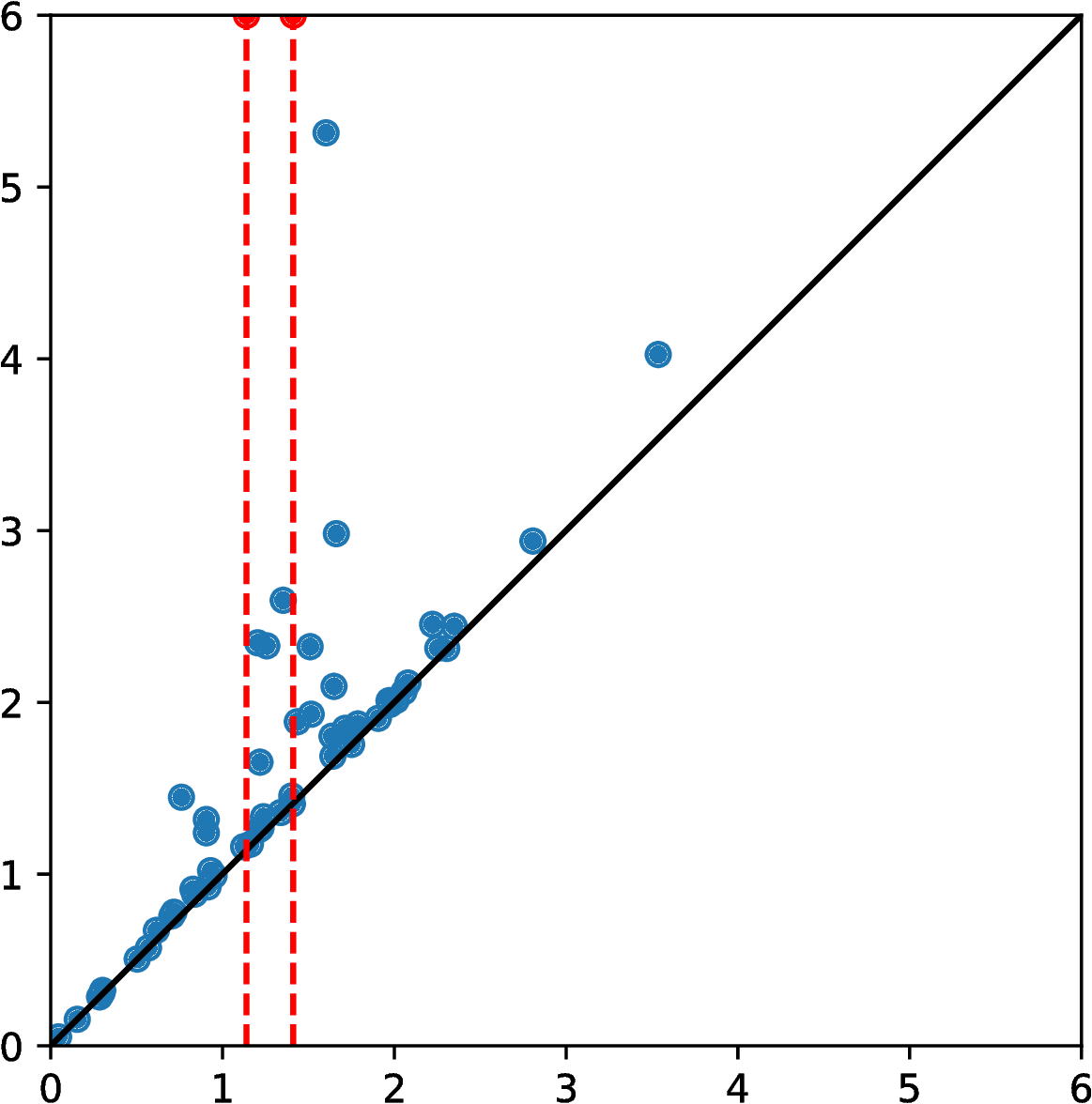}
\caption{}
\end{subfigure}
\caption{\label{fig:ph}
Persistent homology for a point sample on the 2D flat torus (a unit box with periodic boundary conditions). (a) A set of 50 points $\cX$ sampled from the flat torus. The filtration $X_t$ taken here is the union of balls of radius $t$ around the points. We  increase $t$  from $0$ to $\infty$ and calculate $\PH_1$. (b) The barcode for $\PH_1$. Each $1$-cycle  is represented by a  bar, where the endpoints are the radii in which the cycle was formed and filled in (birth, death).
 The two red bars correspond to the two cycles in the torus (``the giant cycles") while the blue ones are considered as ``noise". (c) The persistence diagram for $\PH_1$. Here the $(\birth,\death)$ pairs are plotted as points in the plane. Two points are far away from the diagonal (shown by the red dashed line), representing the true non-trivial or giant cycles of the torus. This figure was generated using the GUDHI package \cite{maria2014gudhi}.}
\end{figure}
%%%%
%%%%%
%%%%%
\subsection{Persistent Homology}
%%%%%

Persistent homology is one of the fundamental tools used in the field of \emph{Applied Topology} or \emph{Topological Data Analysis}. The motivation for developing persistent homology was that as data-analytic features, homological properties can be quite unstable in the sense that small perturbations to the data may result in a significant change of the homological structure.
The solution provided by persistent homology is that instead of extracting the homological features of a single space, we consider a sequence of spaces and extract information about homological cycles together with their evolution throughout the filtration. This can be thought of as a ``multi-scale" version of homology.

A bit more concretely, a filtration $\bX = \set{X_t}_t$ is a set of topological spaces such that for all $s<t$ we have $X_s\subset X_t$. The inclusion maps $i:X_s\hookrightarrow X_t$ induce mappings between cycles $i_*:\Hg_k(X_s)\to \Hg_k(X_t)$. These mappings allow us to track the evolution of cycles as they form and disappear throughout the filtration. Without getting into the formal mathematical definitions, we can think of the $k$-th persistent homology $\PH_k(\bX)$ as a collection (more accurately a graded-module) of $k$-persistent cycles. For each cycle $\gamma\in \PH_k(\bX)$ we can assign two values $\mathrm{bth}(\gamma)$ and $\mathrm{dth}(\gamma)$ standing for ``birth" and ``death" ($\mathrm{bth}(\gamma) \le \mathrm{dth}(\gamma)$), representing the times where the cycle $\gamma$ was formed and later filled in. As a data-analytic tool, persistent homology provides a topological signature for data that also includes some geometric information that makes it more robust to noise than the fixed-scale homology. See Figure \ref{fig:ph} for an example.
For more details as well as formal definitions see \cite{edelsbrunner_computational_2010, edelsbrunner_persistent_2014}. For an overview of TDA see \cite{carlsson_topology_2009,ghrist_barcodes:_2008,doi:10.1146/annurev-statistics-031017-100045}

%%%%%
\subsection{Giant cycles}
%%%%%
We are now ready to define homological percolation.
Let $M$ be a ``nice" compact space, and let $\set{X_t}$ be a filtration such that $X_t\subset M$ for all $t$. As mentioned in the introduction, by giant cycles we refer to those cycles that appear in the filtration for some $t$, and represent one of the nontrivial cycles in $\Hg_k(M)$. These are also referred to as \emph{essential cycles}. We will make this description a bit more formal.

For each $t$, the inclusion map $i:X_t\hookrightarrow M$ induces a map in homology $i_{*,t}:\Hg_k(X_t)\to \Hg_k(M)$. The image of $i_{*,t}$ stands for all the cycles that exist in $X_t$ and are mapped to nontrivial cycle in $M$. We will refer to these cycles as ``giant". By the term \emph{homological percolation} we refer to the study of how and when these giant cycles are formed. For example, the longest red bars in Figure \ref{fig:ph}(b) represent the two holes of the torus, and therefore we consider them as giant, while the other cycles are noise.

Suppose that the filtration $\set{X_t}$ is generated at random (we will discuss specific random models in Section \ref{sec:models}).
For each $t$, fixing $k$ we can define the following events,
\[
	E_t := \set{\im\param{i_{*,t}} \ne 0},\qquad A_t := \set{\im\param{i_{*,t}} = \Hg_k(M)}.
\]
In other words, $E_t$ is the event that there \emph{exists} a giant $k$-cycle in $X_t$, while $A_t$ is the event that \emph{all} possible giant $k$-cycles exist in $X_t$. It is conjectured that similarly to other percolation models, the appearance of the giant $k$-cycles follows a sharp phase transition. This means that there exists a value $t_k^{\mathrm{perc}}>0$ such that
\eqb\label{eqn:crit_perc}
	\prob{E_t} = \prob{A_t} = \begin{cases}1 & t > t_k^{\mathrm{perc}},\\
		0 & t < t_k^{\mathrm{perc}}, \end{cases}
\eqe
where in most cases we study, the stochastic model has an intrinsic parameter $n$ and the equalities above hold in the limit when $n\to\infty$.
This conjecture is supported  by simulations (as the ones presented in this paper), as well as a theoretical work in progress \cite{bobrowski2019} for the boolean model discussed below. It is further conjectured that the thresholds are ordered, so that $t_1^{\mathrm{perc}} < t_2^{\mathrm{perc}} < \cdots < t_{d-1}^{\mathrm{perc}}$ where $d$ is the maximal degree possible (dictated by the dimension of $M$).

%%%%%
\subsection{The Euler Characteristic}
%%%%%
The Euler characteristic (EC) is an integer-valued additive functional. Using the language of homology, one can define the EC of a topological space $X$ as
\[
\chi(X) := \sum_k (-1)^k \beta_k(X),
\]
where $\beta_k(X)$ are the Betti numbers discussed above.
One of the key properties of the EC is that it is a \emph{topological invariant}, namely of $X$ and $Y$ are two spaces that are ``similar'' topologically in the sense that there is a continuous transformation form one to the other (known as \emph{homotopy equivalence}), then $\chi(X)=\chi(Y)$. The EC
shows up in various areas of mathematics (combinatorics, integral geometry, topology, analysis, etc.), and can be defined in various different ways.

In most stochastic models, evaluating  quantities related to the distribution of  homology or persistent homology is between difficult to impossible.
Surprisingly, however, this is not the case for the EC. For example, the expected values of the Betti numbers are unknown in almost all stochastic models studied to-date, while in almost all the models  an explicit formula for the expected EC exists (see Section \ref{sec:models}). Therefore, for probabilistic and statistical analysis, the EC is much favorable, and indeed several interesting applications in statistics and data science were developed based on EC calculations \cite{richardson_efficient_2014,taylor_detecting_2007,worsley_boundary_1995,worsley_estimating_1995}.

Studying a filtration $\set{X_t}$ as in persistent homology, we can define the EC curve $\chi(t) := \chi(X_t)$,
that tracks the evolution of the EC in time. In the random setting we study in this paper, we will mainly focus on the expected EC curve
\[
\bar\chi(t):=\mean{\chi(t)}.
\]

%%%%%
\section{Random Percolation Models} \label{sec:models}
%%%%%
We  focus on three different types of stochastic models to establish our conjectures about the connection between the Euler characteristic and homological percolation. In this section we provide the basic definitions for these models, as well as the formulae we use to calculate the expected EC curve. One of the main reasons for choosing these models  is that in all of them, we can derive an explicit formula for the expected  EC curve $\bar\chi(t)$ (where $t=p,\lambda,\text{ or }\alpha$ depending on the model below).

For simplicity, all the models we  discuss generate random subsets of the $d$-dimensional flat torus $\T^d$. By `flat torus' we refer to the quotient of the box $[0,1]^d$ with the relation $\set{0\sim 1}$ (i.e. opposite faces are ``glued'' together). The flat torus is a good model to study topological phenomena as (a) the metric on it is locally Euclidean, (b) it is a manifold with no boundary, and (c) it has non-trivial homology in all degrees $k=0,\ldots, d$. More precisely, $\beta_k(\T^d) = \binom{d}{k}$.

\subsection{Site percolation models}\label{sec:site}
We start with  simple discrete models for random subsets of $\T^d$. As opposed to the continuous models we discuss later, discrete percolation models are well studied, and often more tractable, both from a theoretical and simulation perspective. We will examine two types of structures, discussed next.

\subsubsection{Cubical complex}

A cubical complex $Q$ is a collection of cubical faces (i.e. vertices, edges, squares, cubes, etc.), that is closed under the boundary operation. We denote by $Q^d_n$ the cubical complex obtained by taking the flat torus $\T^d = [0,1]^d\bs \{0\sim 1\}$ and splitting it into $n$ equal-size boxes, where we assume that $n=m^d$, for some $m\in \N$. Note that every $d$-dimensional box is in $Q^d_n$ together with all its $k$-dimensional faces ($k=0,\ldots,d-1$).

We will consider cubical complexes $Q$ that are subsets of $Q_n^d$. Each such complex is homeomorphic to a subset of $\T^d$ via the natural embedding. Thus, we will interchangeably refer to $Q$ as either a sub-complex of $Q_n^d$ or as a closed subset of $\T^d$. Denote by $F_k(Q)$ the number of $k$-dimensional faces of $Q$. Then, the EC of $Q$ can be calculated by (cf. \cite{hatcher_algebraic_2002})
\eqb\label{eqn:ec_faces}
	\chi(Q) = \sum_{k=0}^d (-1)^k F_k(Q).
\eqe

The \emph{random} cubical complex we study here, denoted $Q(n,p)$ is generated by taking $Q_n^d$ and declaring each $d$-dimensional face (or a site) as either \emph{open} with probability $p$, or \emph{closed} with probability $1-p$, independently between the faces. We then define $Q(n,p)$ as the union of all open boxes (together with their lower dimensional faces). See example in Figure \ref{fig:uniformgiantcycles}. Calculating the expected EC using \eqref{eqn:ec_faces} (see Appendix \ref{sec:ec_calc}) yields,
\eqb\label{eqn:ec_cubical}
	\bar\chi_Q(p) := \mean{\chi(Q(n,p))}= n \sum_{k=0}^d (-1)^k\binom{d}{k} (1-(1-p)^{2^{d-k}}).
\eqe

\begin{rem}
Notice that the definition of $Q(n,p)$ does not imply any connection between $Q(n,p_1)$ and $Q(n,p_2)$ for $p_1\ne p_2$.
However, in order to discuss percolation phenomena as well as persistent homology for the cubical complex, we want the sequence $\set{Q(n,p)}_{p=0}^1$ to be a filtration, so that $Q(n,p_1) \subset Q(n,p_2)$ for all $p_1< p_2$. The simplest way to establish that is the following. Let $U_1,\ldots, U_n$ be $n$ iid random variables, so that $U_i\sim U[0,1]$. Then, for any fixed $p$, we say that site $i$ is open in $Q(n,p)$ if $U_i \le p$. This way for all $p\in[0,1]$ the complex $Q(n,p)$ has the distribution we desire, and indeed $Q(n,p_1)\subset Q(n,p_2)$ for all $p_1 < p_2$.
\end{rem}

\begin{figure}[H]
\centering
\begin{subfigure}[t]{0.3\textwidth}
\centering
\includegraphics[height=\textwidth]{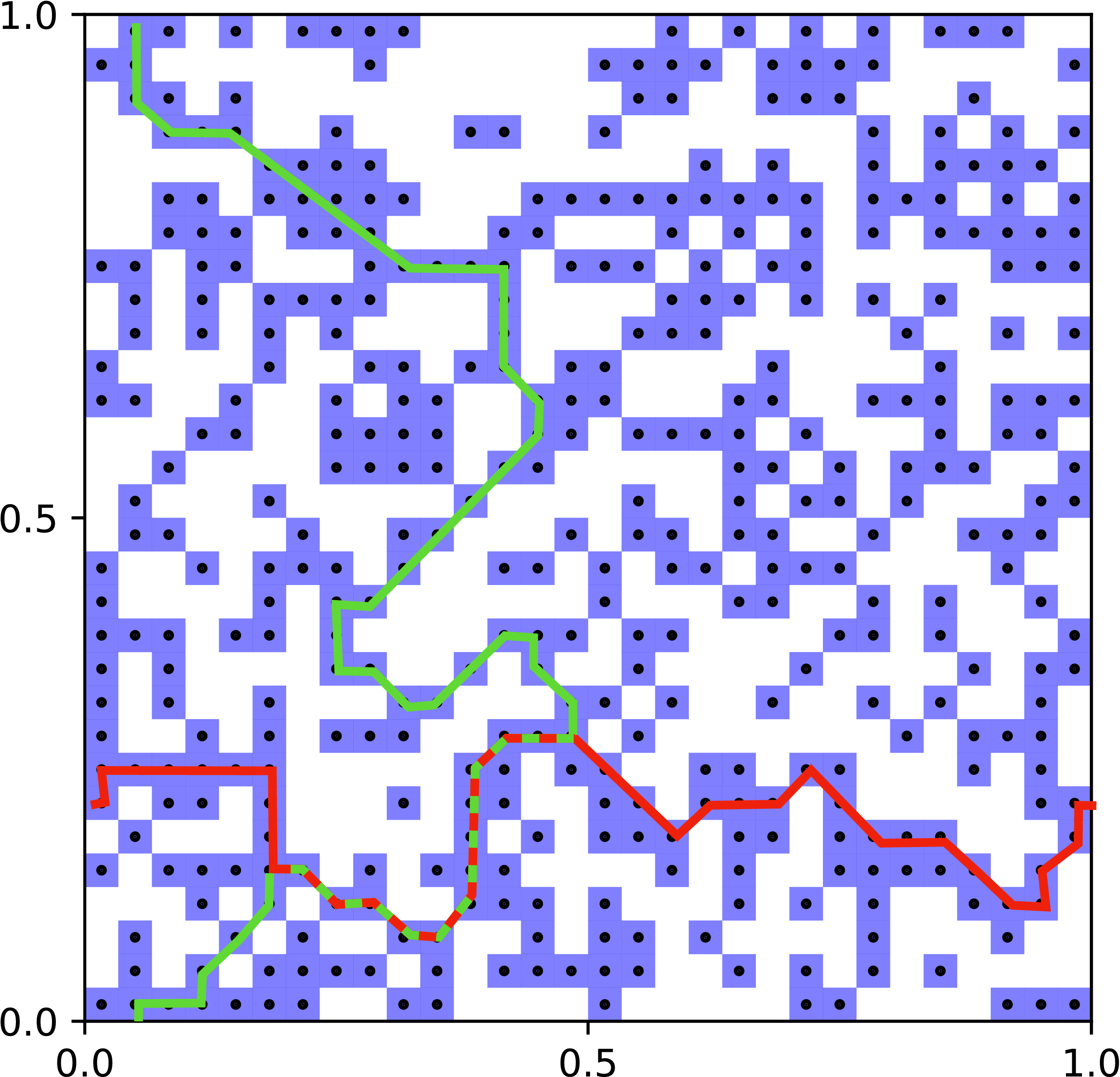}
\caption{}
\end{subfigure}
\hfill
\begin{subfigure}[t]{0.3\textwidth}
\centering
\includegraphics[height=\textwidth]{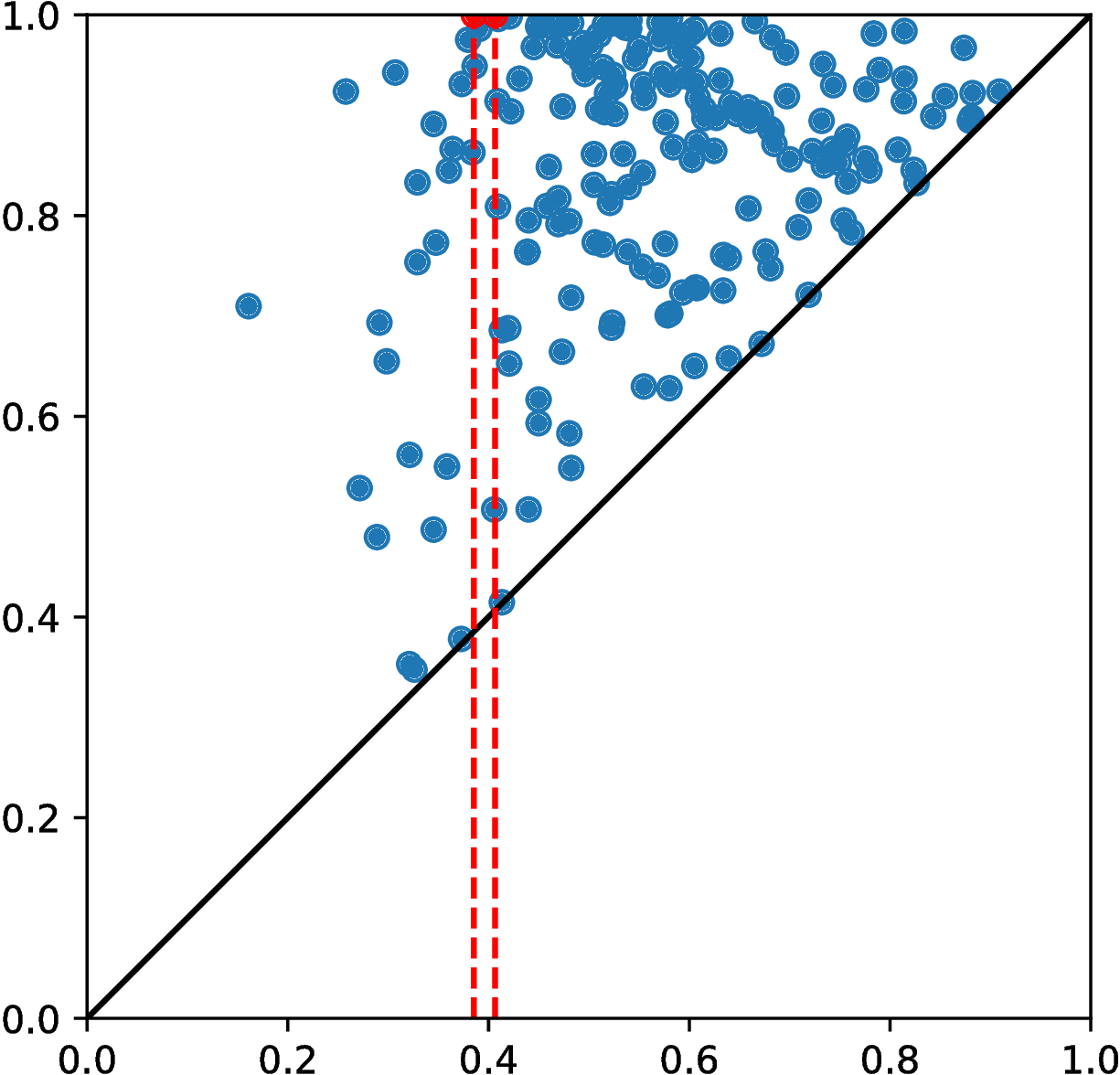}
\caption{}
\end{subfigure}
\hfill
\begin{subfigure}[t]{0.3\textwidth}
\includegraphics[height=\textwidth]{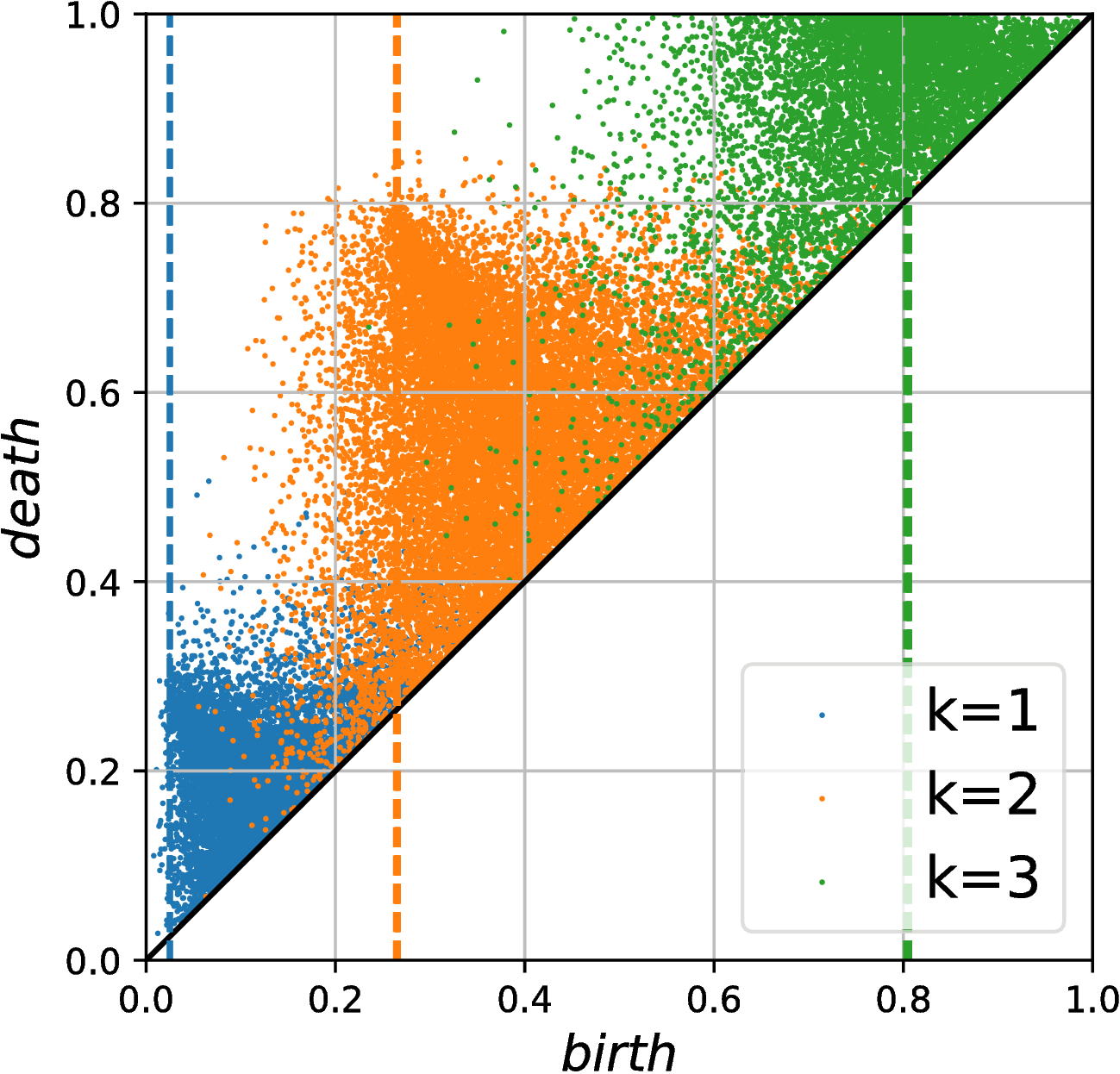}
\caption{}
\end{subfigure}
\caption{\label{fig:uniformgiantcycles}
   (a) The emergence of the giant cycles for $Q(n,p)$ with $d=2$, $n=2500$, and $p=0.43$. The red lines mark the two giant 1-cycles (recall the period boundary gluing).  (b) The corresponding persistence diagram for $\PH_1$. The two red and green curves mark the times when each of the giant 1-cycles appeared. (d) The persistence diagram for the case $d=4$, $n=65536$. In this case we have homology in degrees $k=1,2,3$. The vertical lines mark the birth times of the giant $k$-cycles.}
\end{figure}

\subsubsection{Permutahedral complex}\label{sec:perm}
We introduce an alternative discrete model to address some inherent shortcomings of the cubical model. In two dimensions, a hexagonal tiling is often used in percolation theory instead of the $\mathbb{Z}^2$-grid.
Here we will use a higher-dimensional notion, we refer to as a \emph{permutahedral} tessellation, where the basic building block is a \emph{permutahedron} -- the generalization of a hexagon to arbitrary dimensions. Notice that taking all $3!$ permutations of the coordinates $(1,2,3)$ yields 6 vertices in $\R^3$ that form a hexagon. Similarly, a $d$-dimensional permutahedron is the polytope obtained by taking the convex hull of all $(d+1)!$ permutations of $(1,\ldots,d+1)$ in $\R^{d+1}$.

Next we construct the \emph{permutahedral lattice}. The definitions here follow in~\cite{conway2013sphere} (Section 6.6) and \cite{baek2009some}. Define
\[
	\widehat \R^d := \set{(x_0,x_1\ldots, x_d) \in \R^{d+1} \given \sum_{i=0}^d x_i = 0},
\]
i.e.~$\widehat \R^d$ is a $d$-dimensional plane in $\R^{d+1}$. Define the $A_d$-lattice as $A_d := \Z^{d+1}\cap \widehat \R^d$, and its dual  $A^*_d$  as
$$A^*_d = \left\{x \in \widehat\R^{d} \given \forall y\in A_d:  x\cdot y \in \mathbb{Z} \right\}.$$
Taking $\pi : \widehat \R^d \to \R^d$ to be the natural isometry, then the Voronoi cells of $\pi(A_d^*)$ form a permutahedral tessellation of $\R^d$, and the set of centers of these cells $\pi(A_d^*)$ is called a \emph{permutahedral lattice}.

A \emph{site} in this model is the closure of a Voronoi cell of the a point in $\pi(A^*_d)$. Each site has the structure of a $d$-dimensional  {permutahedron}, hence the name of the model. Note that the points of $\pi(A_d)$ form the vertices of this polytope (see also \cite{ziegler2012lectures}).

To go from a tessellation of $\R^d$ to a tessellation of the the torus $\T^d$, we use an analogous method as in the cubical case, gluing ``opposite" faces together. In practice, this is done by identifying points of $\pi(A^*_d)$ (similarly to the periodic Delaunay complex~\cite{bogdanov:hal-01224549}). As in the cubical case, this limits our choice of grid size ($n$), depending on the dimension. We denote the corresponding tessellation of $\T^d$ as $P_n^d$.
%To perform the gluing we consider the canonical basis and find the corresponding point to glue as shown in Figure~\ref{}.
\begin{figure}[ht]
\centering
\begin{subfigure}{0.3\textwidth}
\includegraphics[width=\textwidth]{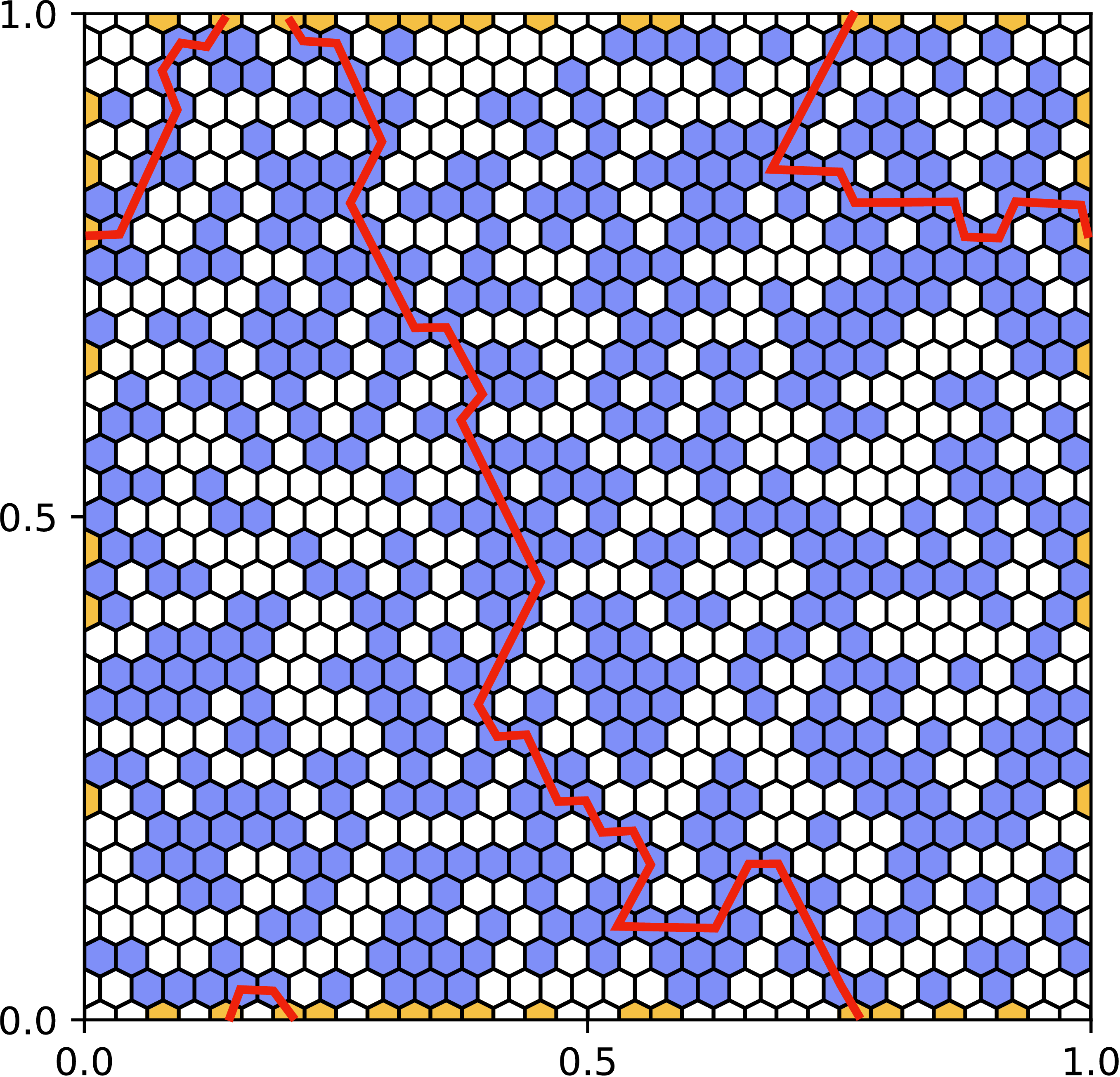}
\caption{}
\end{subfigure}
\hfill
\begin{subfigure}{0.3\textwidth}
\includegraphics[width=\textwidth]{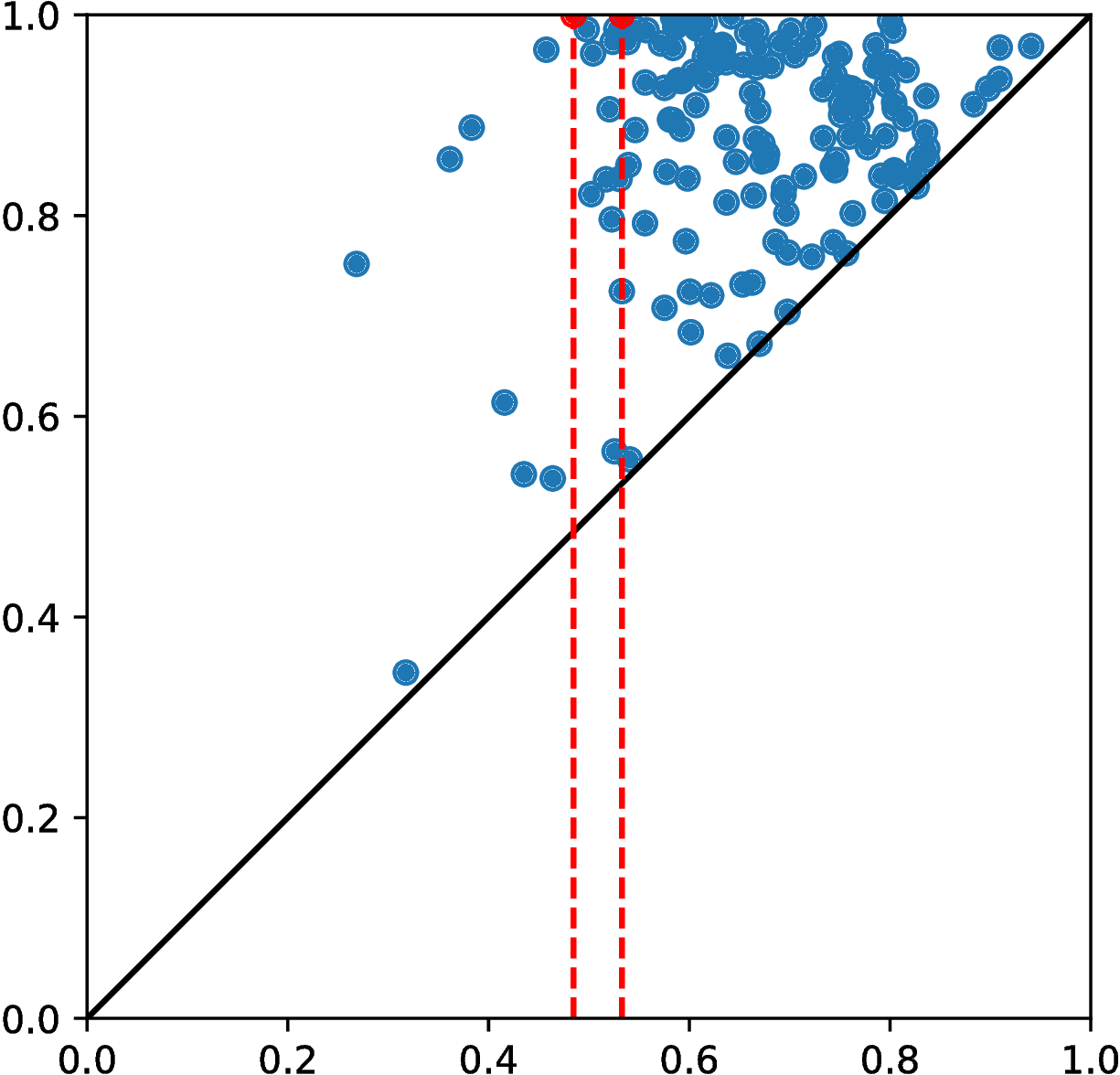}
\caption{}
\end{subfigure}
\hfill
\begin{subfigure}{0.3\textwidth}
\includegraphics[width=\textwidth]{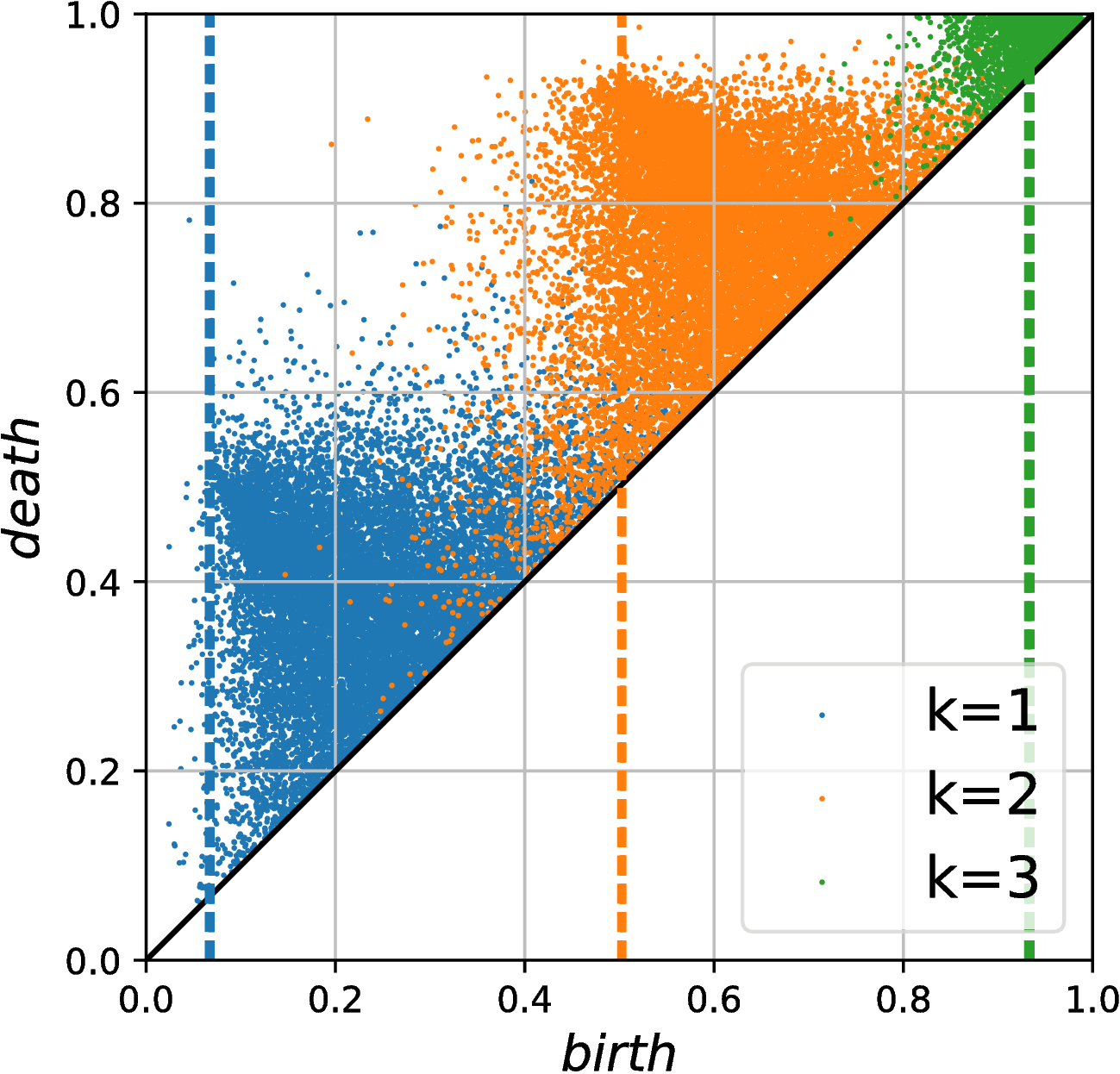}
\caption{}
\end{subfigure}
\caption{\label{fig:hex} Site percolation on a permutahedral grid. (a) In $d=2$ we have a hexagonal grid. Since we are on the flat torus, sites on opposite sides are glued together (marked by yellow). The red path marks a giant $1$-cycle at $p=\frac{1}{2}$. Note that only one giant cycle has appeared and that it is the sum of the two ``obvious" giant cycles. (b) $\PH_1$. (c) Persistence diagrams for the case $d=4$.}
\end{figure}

Similarly to the random cubical model, we define a random permutahedral complex.  Let $P(n,p)$ be a random subset of $P_n^d$ where each site is open with probability $p$, and closed with probability $1-p$. As in the cubical case, we can calculate the EC by counting faces in different dimensions. This leads to (see Appendix \ref{sec:ec_calc}),
\eqb\label{eqn:ec_perm}
\bar\chi_P(p) := \mean{\chi(P(n,p))} = n\sum\limits_{k=0}^d (-1)^{d-k}\left(1- (1-p)^{k+1}\right) \sum\limits_{j=0}^{k+1}(-1)^{k+1-j} \binom{k}{j}j^{d+1}.
\eqe

The main reason for our interest in the permutahedral complex, is that in contrast to the cubical model, it exhibits a powerful duality property.
Let $P\subset P_n^d$ be a sub-complex and recall that we define a giant as an element in the image $i_*:\Hg_k(P)\to \Hg_k(\T^d)$. Define,
\[
	\cB_k(P) := \dim(i_*(\Hg_k(P))),
\]
i.e.~$\cB_k(P)$ is the number of giant $k$-cycles in $P$.
Next, let $ P^\c = \mathrm{cl}(P_n^d \bs P)$ - the closure of the complement. The following lemma (in fact a stronger version of it) is proved in Appendix \ref{sec:duality}.

\begin{lem}\label{lem:dual_betti}
For $0\le k \le d$,
\[
\cB_k(P) + \cB_{d-k}(P^\c) = \beta_k(\T^d).
\]
\end{lem}

The lemma implies that whenever a giant $k$-cycle emerges in $P$, a giant $(d-k)$-giant cycle disappears in $P^\c$ and vice verse. This lemma can be viewed as a ``homological" version of the duality argument used for  percolation in $\Z^2$ (see, e.g.~\cite{grimmett_percolation._nodate}), that is the fact that a horizontal crossing in one grid prevents a vertical crossing in the dual grid.

Notice that by the definition of $P(n,p)$, we have that ${P^\c(n,p)}$ has the same distribution as $P(n,1-p)$. Recall from the introduction that we conjecture  the existence of an increasing sequence of sharp thresholds denoted $p_1^\mathrm{perc} < p_2^\mathrm{perc} < \cdots < p_{d-1}^\mathrm{perc}$ where the  the giant $k$-cycles appear. Therefore, together with Lemma \ref{lem:dual_betti}, we can show that if the sharp thresholds exist, then in the permutahedral complex case we have
\[
	p_k^{\mathrm{perc}} = 1 - p_{d-k}^{\mathrm{perc}}.
\]
In other words,  the appearance of the $k$-cycles in this model is in symmetry with the appearance of the $(d-k)$-cycles (as can be seen in the simulations later).  Notice that if $d$ is even we have that $p_{d/2} = 1/2$.
This is a generalization of the phenomenon known for the 2-dimensional hexagonal lattice, where the (classic) percolation threshold is exactly $1/2$ (see e.g.~\cite{grimmett_percolation._nodate}).
The symmetry between $P^\c(n,p)$ and $P(n,1-p)$ also implies a symmetry for the expected EC curve, so we have (see Appendix \ref{sec:duality}),
\[
	\bar\chi_P(p) =  (-1)^{d-1}\bar \chi_P(1-p).
\]
While the symmetry of the EC curve is not obvious in Equation \eqref{eqn:ec_perm}, it is quite apparent in the simulations we present later.

\subsection{Boolean model}\label{sec:boolean}

Let $X_1,X_2,\ldots,$ be a sequence of $\iid$ random variables uniformly distributed on $\T^d$. Let $N\sim\pois{n}$ be a Poisson random variable, independent of $\set{X_i}$. Then the process $\cP_n := \set{X_1,\ldots, X_N}$ is called a \emph{spatial Poisson process} on $\T^d$.
The simple Boolean model we consider here is merely the union of  balls
\[
B_r(\cP_n) := \bigcup_{p\in \cP_n}B_r(p),
\]
where $B_r(p)$ is a closed  ball of radius $r>0$ around $p$. In this model, it is known \cite{penrose_random_2003} that
percolation occurs when $ \omega_d n r^d = \lambda$ (or $r = (\lambda/\omega_d n)^{1/d}$) for some fixed  value $\lambda>0$, where $\omega_d$ is the volume of the unit ball in $\R^d$. Consequently, the percolation model we  study is
\eqb\label{eqn:bnr}
	B(n,\lambda) := B_{(\lambda/\omega_d n)^{1/d}}(\cP_n).
\eqe

In \cite{bobrowski_vanishing_2017} a formula for the expected EC of $B(n,\lambda)$ was proved.
The main idea was to consider the distance function to the point process $\cP_n$,  and evaluate the expected number of its critical points.
A mathematical framework known as \emph{Morse theory}, provides a formula similar to \eqref{eqn:ec_faces}, where the number of $k$-faces is replaced by the number of critical points of index $k$. The result in \cite{bobrowski_vanishing_2017} is the following formula,
\eqb\label{eqn:eec_bool}
\bar\chi_B(\lambda) := \mean{\chi(B(n,\lambda)} = n e^{-\lambda}\param{1+\sum_{k=1}^{d-1}A_{d,k} \lambda^k},
\eqe
where $A_{d,k}$ are defined in \cite{bobrowski_vanishing_2017} via some geometric integrals.
For $d=2$ we can show that this results in
\[
	\bar\chi_B(\lambda)  = ne^{-\lambda}(1-\lambda).
\]
For $d=3$, the calculations in \cite{bobrowski_topology_2014} show that %\footnote{we should verify this formula works}
\[
	\bar\chi_B(\lambda) = ne^{-\lambda}\param{1 - 3\lambda + \frac{3}{32}\pi^2\lambda^2}.
\]
For higher dimensions, the integral formulae in \cite{bobrowski_vanishing_2017} are   difficult to calculate explicitly. In our simulations for $d=4$, we use numerical methods to approximate the coefficients $A_{d,k}$.

\begin{figure}[ht]
\centering
\begin{subfigure}[t]{0.3\textwidth}
\centering
\includegraphics[height=\textwidth]{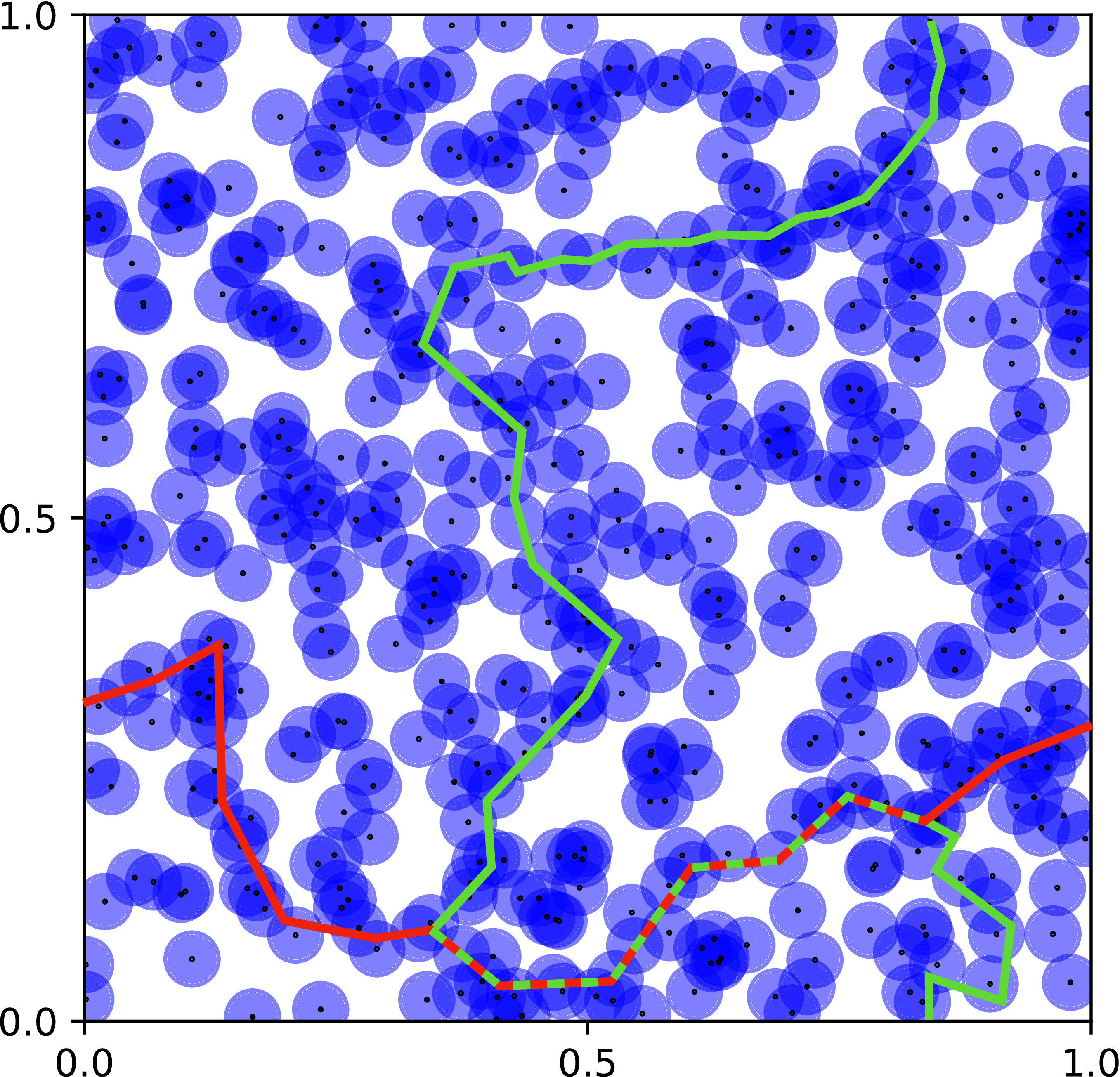}
\caption{}
\end{subfigure}
\hfill
\begin{subfigure}[t]{0.3\textwidth}
\centering
\includegraphics[height=\textwidth]{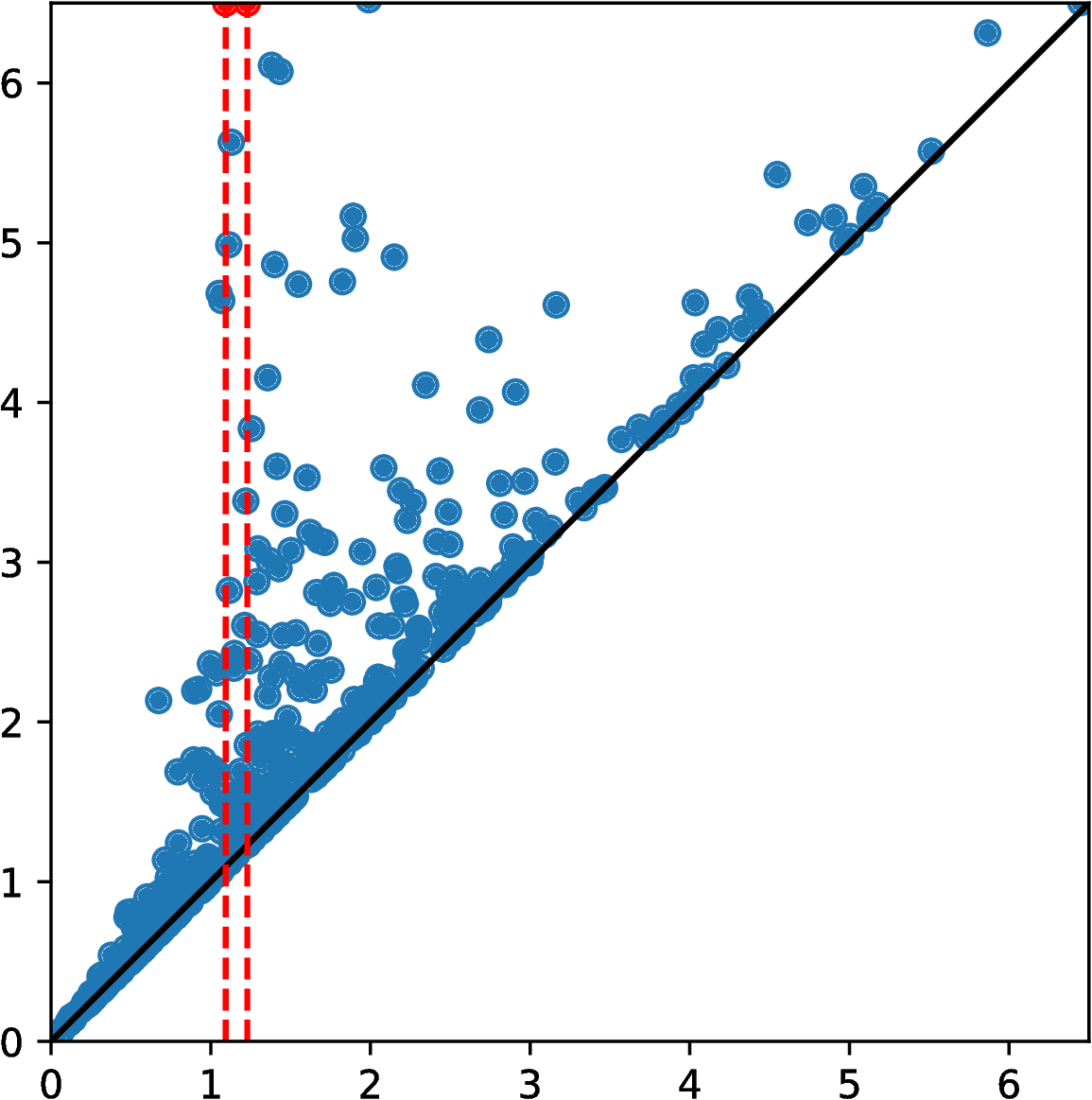}
\caption{}
\end{subfigure}
\hfill
\begin{subfigure}[t]{0.3\textwidth}
	\includegraphics[height=\textwidth]{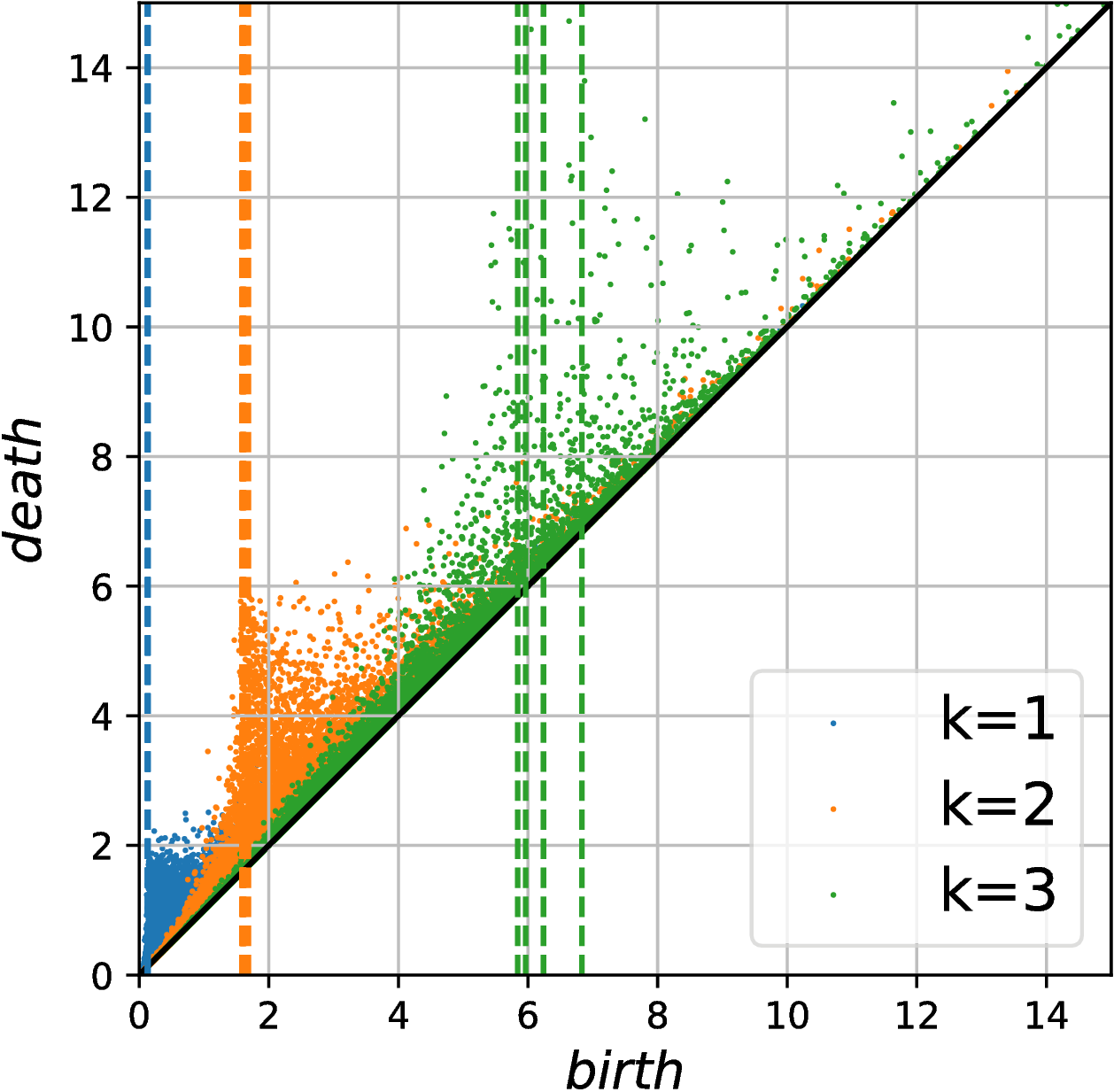}
\caption{}
\end{subfigure}
\caption{\label{fig:bool_giantcycles}
   (a) The emergence of the giant cycles for $B(n,\lambda)$ with $d=2$, $n=500$, and $\lambda=1.19$. (b) The corresponding persistence diagram for $\PH_1$. (d) The persistence diagrams  $\PH_k$ for the case $d=4$, $n=5000$, where $k=1,2,3$.}
\end{figure}
\subsection{Gaussian random fields}\label{sec:grf}

The last model we study is of a completely different nature than the previous ones. A real-valued Gaussian field on the torus, is a random function $f:\T^d\to \R$ such that for every $k$ and every collection of points $x_1,\ldots,x_k \in \T^d$, the random variables $f(x_1),\ldots, f(x_k)$ have a multivariate Gaussian (aka normal) distribution. It is well-known that the entire distribution of the random field $f$ is determined by its expectation function $\mu:\T^d\to\R$ and covariance function $C:\T^d\times\T^d\to \R$, defined as
\[
\splitb
	\mu(x) &:= \mean{f(x)},\\
	 C(x,y) &:= \cov{f(x)}{f(y)} = \mean{(f(x)-\mu(x))(f(y)-\mu(y))}
\splite
\]
for all $x,y\in \T^d$. In this paper we will consider $f$ with $\mu\equiv 0$ and with a covariance function
\eqb\label{eqn:grf_cov}
C(x,y) = \exp\left(-\frac{||x-y||^2}{\sigma^2}\right).
\eqe

For a given Gaussian field $f$, we will study the percolation phenomena as well as the EC for the sub-level sets, defined as%\footnote{or super?}
\[
G(\alpha) := \set{x\in \T^d : f(x)\le \alpha}.
\]
Notice that by definition $\set{G(\alpha)}_{\alpha = -\infty}^\infty$ is a filtration.
This, in particular, implies that we can define the notions of persistent homology and homological percolation for this model as well.

\begin{figure}[h]
\centering
\begin{subfigure}[t]{0.3\textwidth}
\centering
\includegraphics[height=\textwidth]{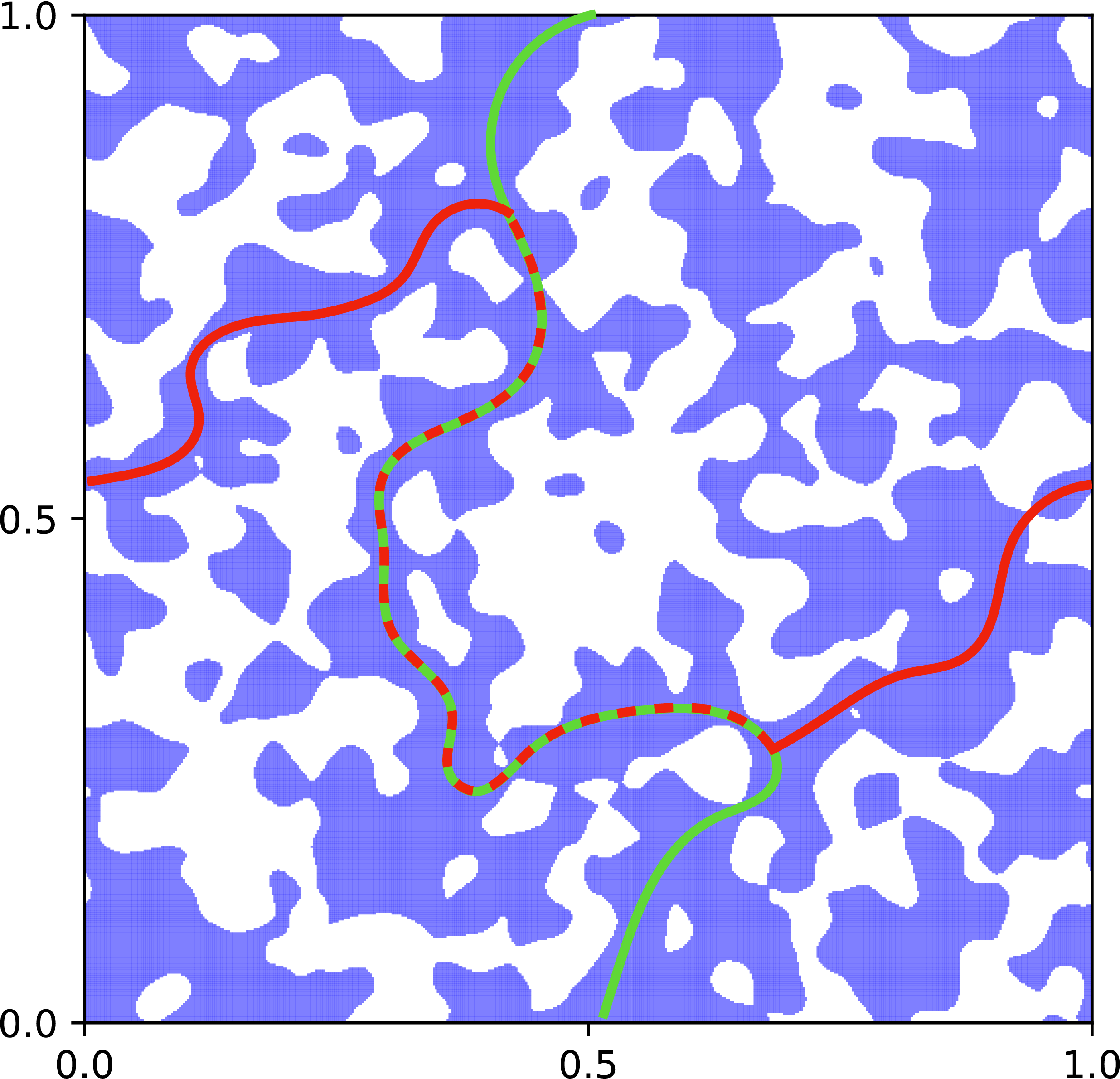}
\caption{}
\end{subfigure}
\hfill
\begin{subfigure}[t]{0.3\textwidth}
\centering
\includegraphics[height=\textwidth]{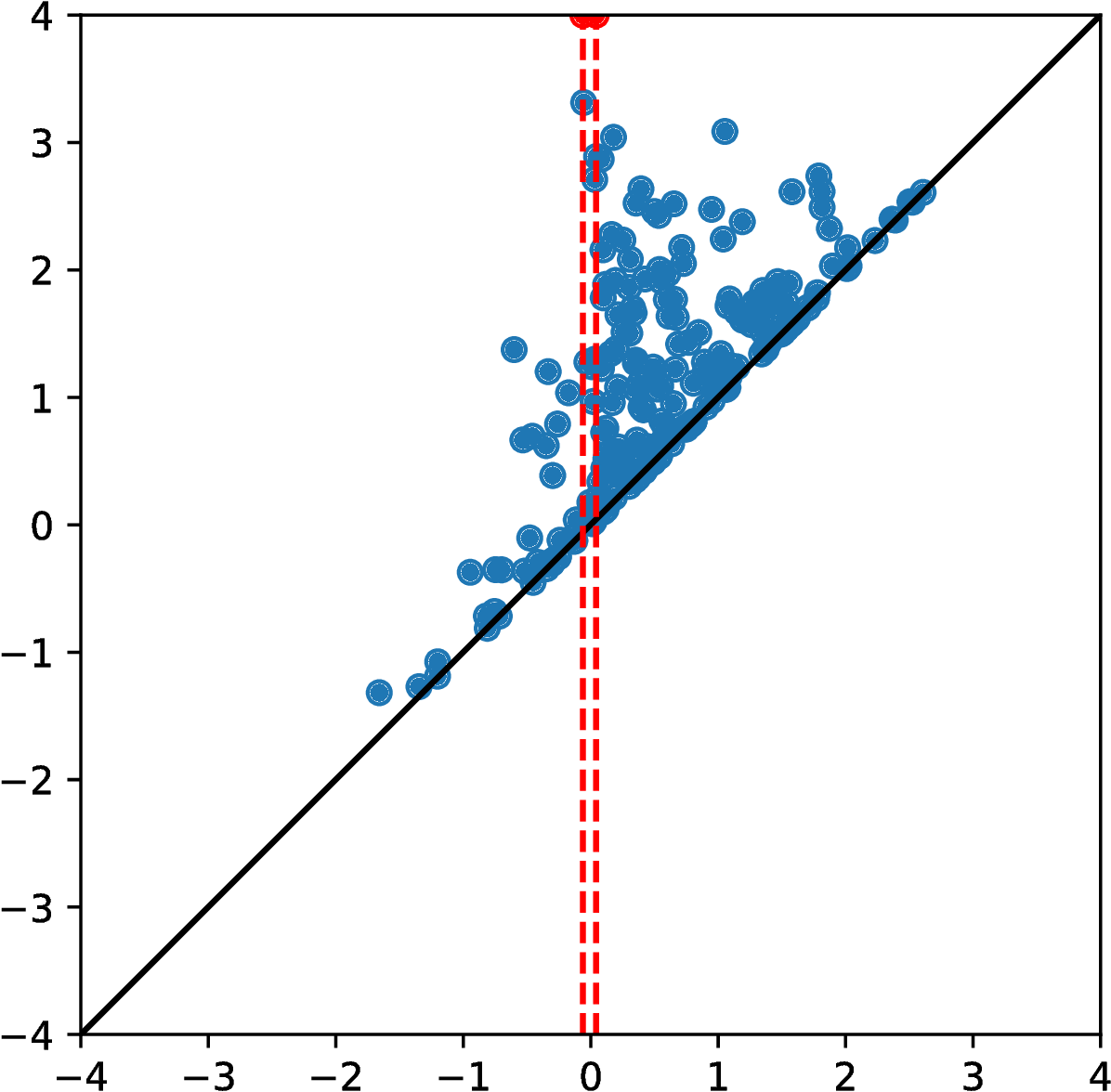}
\caption{}
\end{subfigure}
\hfill
\begin{subfigure}[t]{0.3\textwidth}
	\includegraphics[height=\textwidth]{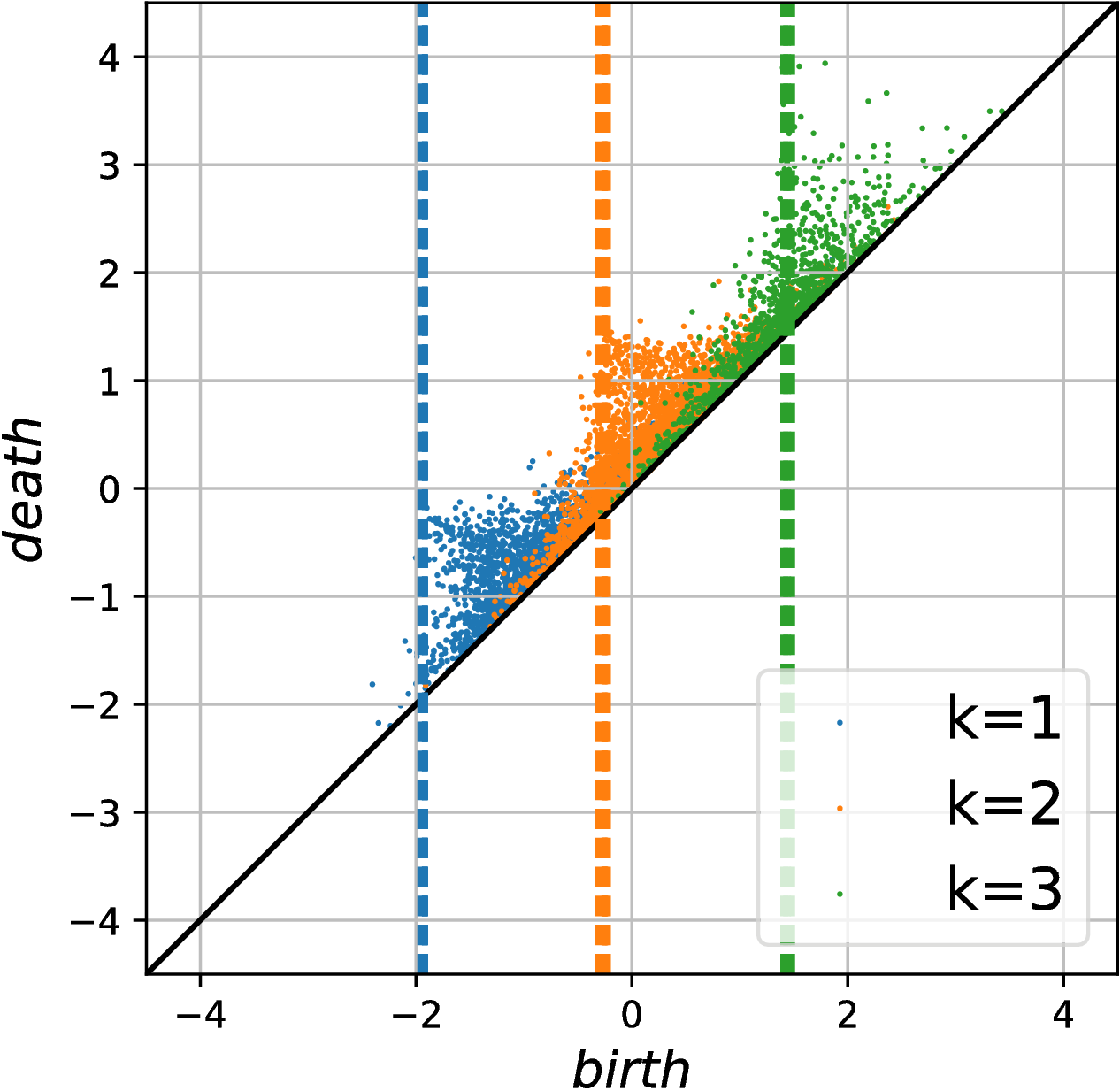}
\caption{}
\end{subfigure}
\caption{\label{fig:gaussian_giantcycles}
   (a) The emergence of the giant cycles for $G(\alpha)$ where $d=2$, and the grid size is $512\times 512$.  (b) The corresponding persistence diagram for $\PH_1$. (d) Persistence diagrams for the case $d=4$.}
\end{figure}

The evaluation of the expected EC for $G(\alpha)$ is the most complicated of all the models in this paper.
This was done in \cite{taylor_gaussian_2009} via a formula known as the \emph{Gaussian Kinematic Formula} (GKF).
The fundamental idea behind the GKF is to use Morse theory in a similar way to that of the Boolean model.
With some assumptions on the mean $\mu(\cdot)$ and the covariance function $C(\cdot,\cdot)$, one can show that the random function  $f$ is a Morse function  with probability $1$. Briefly, Morse functions are differentiable, and have at most a single critical point at each level $\alpha$. Evaluating the expected number of critical points then leads to the expected EC of  $G(\alpha)$, as in the Boolean model.

The  GKF as presented in \cite{taylor_gaussian_2009} covers general Gaussian fields defined on general Riemannian manifolds. In the special case that we examine in this paper, i.e.~a Gaussian field on $\T^d$ with the covariance function given in \eqref{eqn:grf_cov}, the GKF yields the following formula.
\eqb\label{eqn:ec_gauss}
\bar\chi_G(\alpha) := \mean{\chi(G(\alpha)} = \frac{2}{\omega_d} (2\pi)^{-\frac{d+1}{2}} \cH_{d-1}(-\alpha) e^{-\alpha^2/2},
\eqe

where $\omega_d$ is the volume of a unit ball in $\R^d$, and $\cH_n$ is the Hermite polynomial, given by
\eqb\label{eqn:hermite}
	\cH_n(x) = n! \sum_{j=0}^{\floor{n/2}} \frac{(-1)^j x^{n-2j}}{j!(n-2j)!2^j}.
\eqe

Finally, note that when simulating a Gaussian random field we have to take a discretized grid. The size of this grid will bed noted by $n$.

\section{Results}\label{sec:sim}
In this section we present simulation results for the four models described above, for dimensions $d=2,3,4$. The computations were done using the GUDHI~\cite{maria2014gudhi} library. The technical details about the simulations can be found in Appendix \ref{sec:sim_details}.

\begin{rem} Notice that $t_k^{\text{perc}}$ was defined in \eqref{eqn:crit_perc} as the (non-random) critical value for the \emph{probability} of homological percolation to switch form 0 to 1. For the models we are studying, these phase transition is always defined in the limit as $n\to \infty$. Since we do not have an access to the limit, we use the (random) birth-time of the first giant $k$-cycle, as an approximation to $t_k^{\text{perc}}$.
\end{rem}

In Figure \ref{fig:ec_curves} we show the theoretical expected Euler curves, together with the mean appearance of the giant cycles.
This figure demonstrates several noteworthy phenomena. Across all models and dimensions, the appearances of the giant cycles align with the zeros of the EC curve. In particular, for the permutahedral complex in the even dimensions (Figures~\ref{fig:perm2},\ref{fig:perm4}), the giant cycle in the middle dimension  aligns perfectly with the corresponding zero of the EC curve. This is a direct consequence of the symmetry between the complex and its complement discussed in Section \ref{sec:perm}.
Note that in the odd dimension (3) the EC curve for the permutahedral complex is still symmetric, but there is no middle dimension for the giant cycles. Finally, for dimensions 3 and 4, the giant cycles appear before the corresponding zero in the lower dimensions and after the zero in the higher dimensions. Smaller examples in dimension 5 also follow this pattern. This leads to the following conjecture.
\begin{con}\label{con:1}
For $d\geq 3$ and every $k < d/2$ we have $t_k^{\mathrm{perc}} \le t_k^{\mathrm{ec}}$, while for every $k>d/2$ we have $t_k^{\mathrm{perc}} \ge t_k^{\mathrm{ec}}$.
\end{con}
For the middle dimension ($d/2$ when $d$ is even), if the model is symmetric with respect to the parameter $t$ then the middle giant cycle should align perfectly with the corresponding zero in the EC curve, as in the permutahedral complex.
Notice that the GRF model we take is also symmetric (zero mean). We suspect that the tiny difference between $t_{d/2}^{\text{perc}}$ and $t_{d/2}^{\text{ec}}$ (Figures \ref{fig:gaussian2} and \ref{fig:gaussian4}) is due to the fact that we are using a discretized sample of a continuous field. This is supported by the statistics presented below.
For asymmetric models (e.g. the cubical and Boolean) it is not clear what should happen for $d/2$.
In Figures \ref{fig:uniform2},\ref{fig:uniform4},\ref{fig:boolean2}, and \ref{fig:boolean4}, we consistently see $t_{d/2}^{\text{ec}} < t_{d/2}^{\text{perc}}$, however, this behavior would have flipped if we were to take the complement objects.

Conjecture \ref{con:1} may have significant implications. For example, in Gaussian random fields, it is not known (for $d\ge 3$) whether the percolation thresholds for the super- and sub-level sets are separated, i.e.~whether there exists a regime where both the sub- and super-level sets have a giant component simultaneously. In \cite{bobrowski2019}, we show that $t_1^{\mathrm{perc}}$ coincides with the percolation threshold for the sub-level sets, while $t_{d-1}^{\mathrm{perc}}$ coincides with the threshold for the super-level sets.
Therefore, given that the expected EC curve is known in a closed form,  bounds on the relationship between the zeros of the expected EC curve and the giant cycles would imply separation.

Next, we investigate the error terms $\Delta_k = (t^{\text{perc}}_k - t^{\text{ec}}_k)$. In Figure \ref{fig:stats} we show some statistics for the examples in Figure \ref{fig:ec_curves}. For each of the models we repeated the simulation, and estimated the mean and variance for the first birth time of a giant $k$-cycle.
We observe that $\Delta_k$ converges rather quickly in the number of points (though as noted above not necessarily to zero), and the variance  becomes  small very quickly (note that the x-axis is logarithmic). Furthermore, with the exception of the Gaussian random field, we observe that $\Delta_k$ generally remains either always positive or always negative in line with Conjecture \ref{con:1}.  Note that the standard deviation in the Gaussian random field model remains roughly constant and centered around zero for all grid sizes. This is because the covariance function of the GRF is held constant, indicating that the correlation effects are the key driver of variability in the appearance of the giant cycles.

%%%%%
\begin{figure}[ht!]
	\centering
     \begin{subfigure}[b]{0.3\textwidth}
         \centering
				 \includegraphics[width=\textwidth]{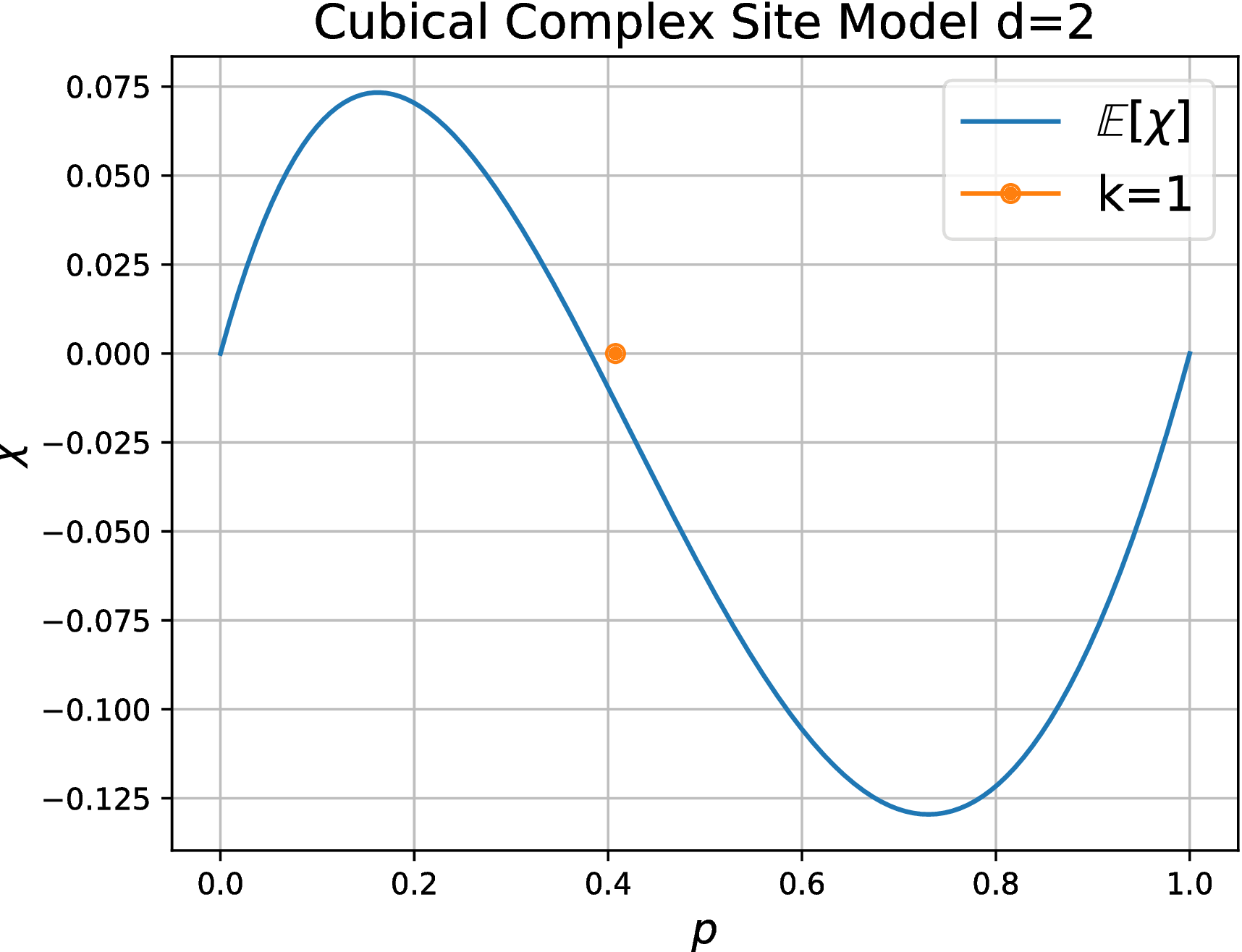}
				 \caption{}
         \label{fig:uniform2}
     \end{subfigure}
     \hfill
     \begin{subfigure}[b]{0.3\textwidth}
         \centering
				 \includegraphics[width=\textwidth]{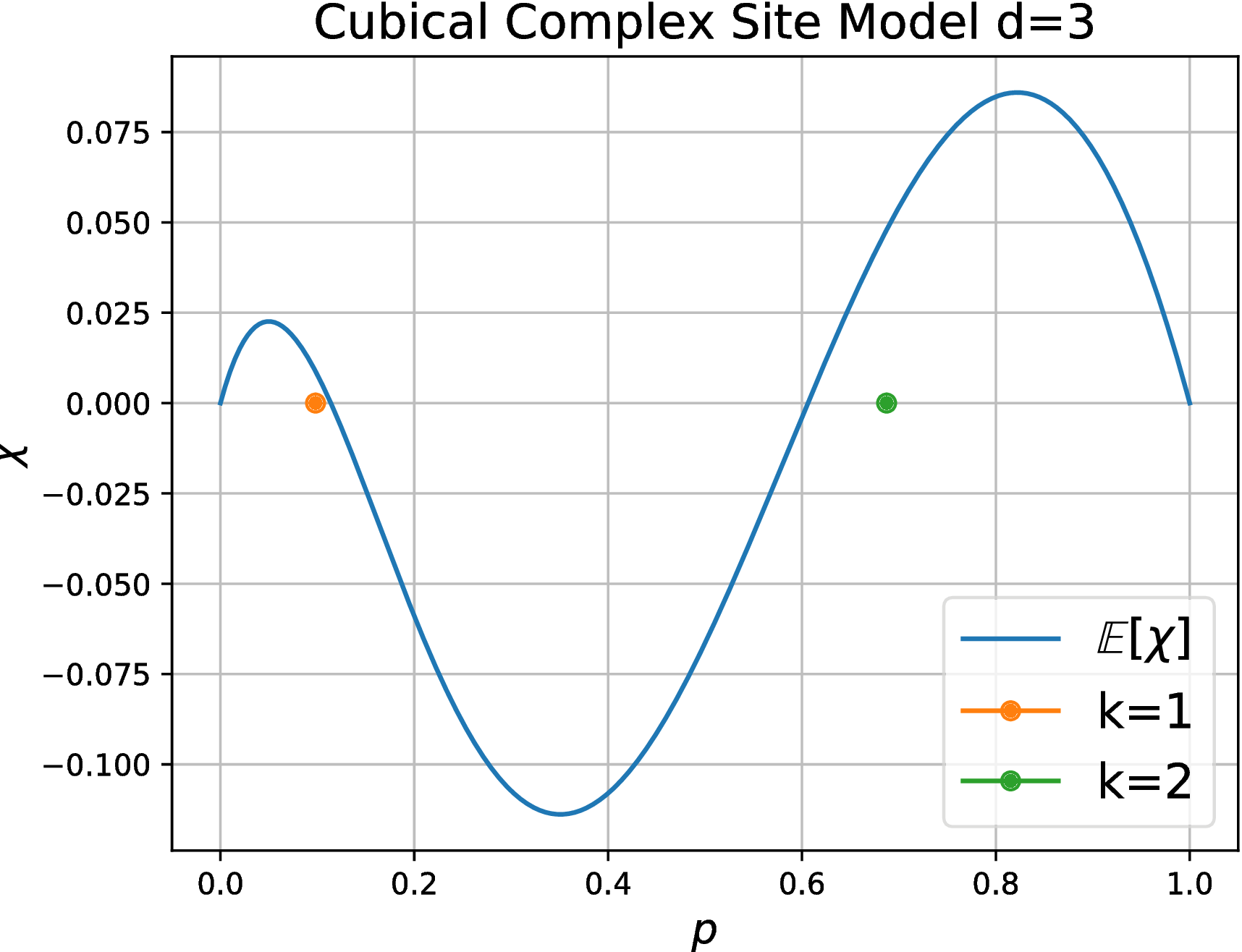}
				 \caption{}
         \label{fig:uniform3}
     \end{subfigure}
     \hfill
     \begin{subfigure}[b]{0.3\textwidth}
         \centering
				 \includegraphics[width=\textwidth]{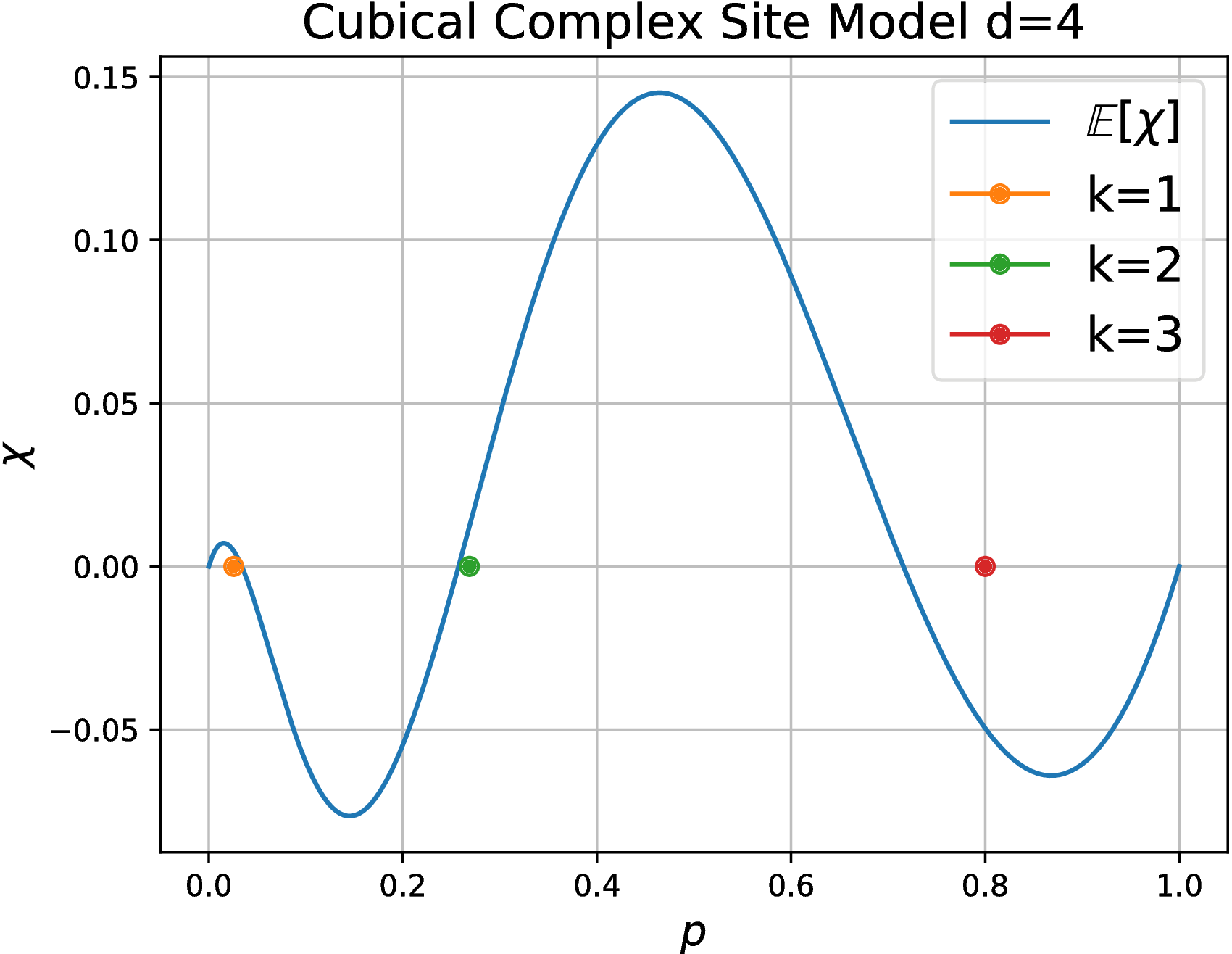}
				 \caption{}
         \label{fig:uniform4}
     \end{subfigure}
     \hfill
     \begin{subfigure}[b]{0.3\textwidth}
         \centering
				 \includegraphics[width=\textwidth]{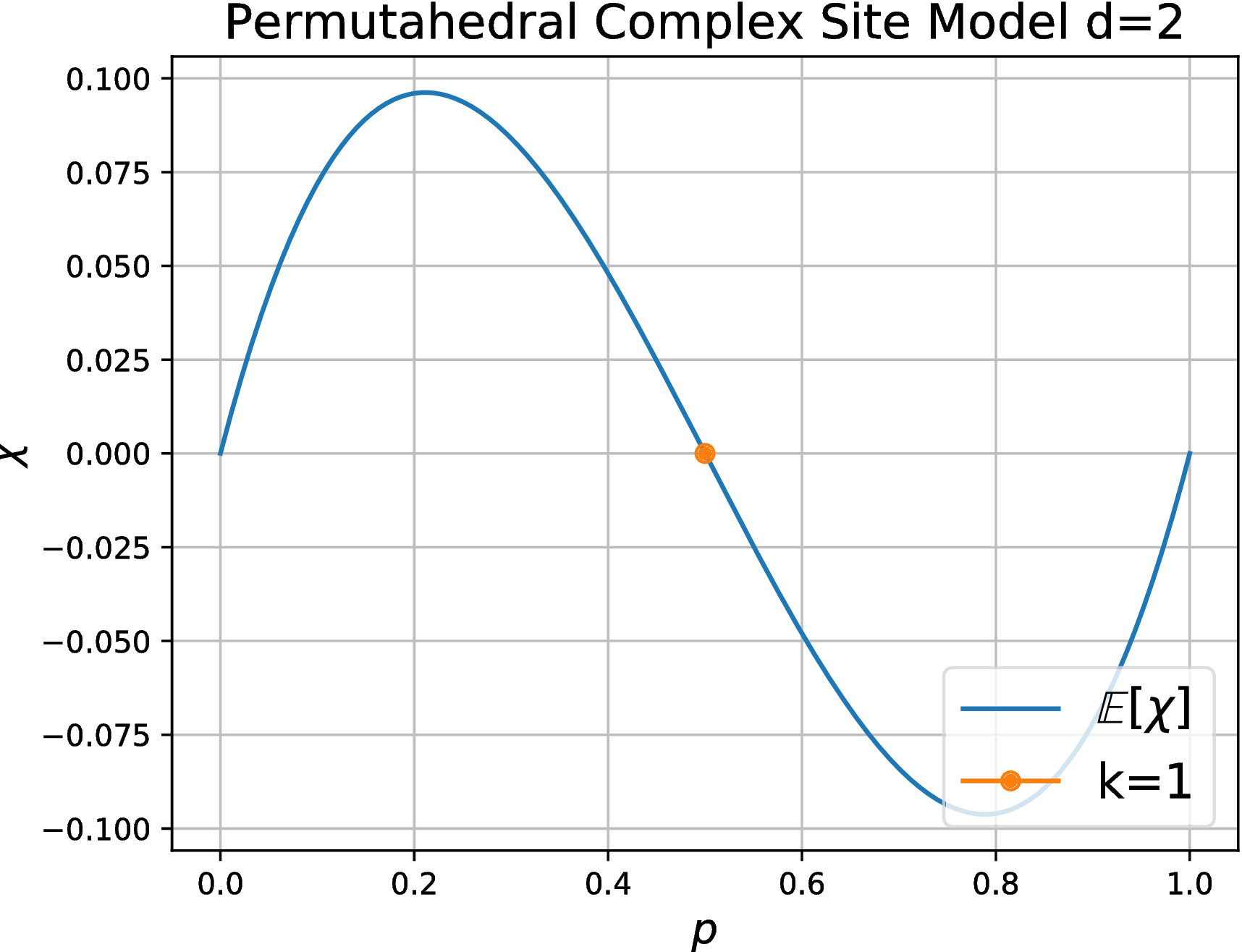}
				 \caption{}
         \label{fig:perm2}
     \end{subfigure}
     \hfill
     \begin{subfigure}[b]{0.3\textwidth}
         \centering
				 \includegraphics[width=\textwidth]{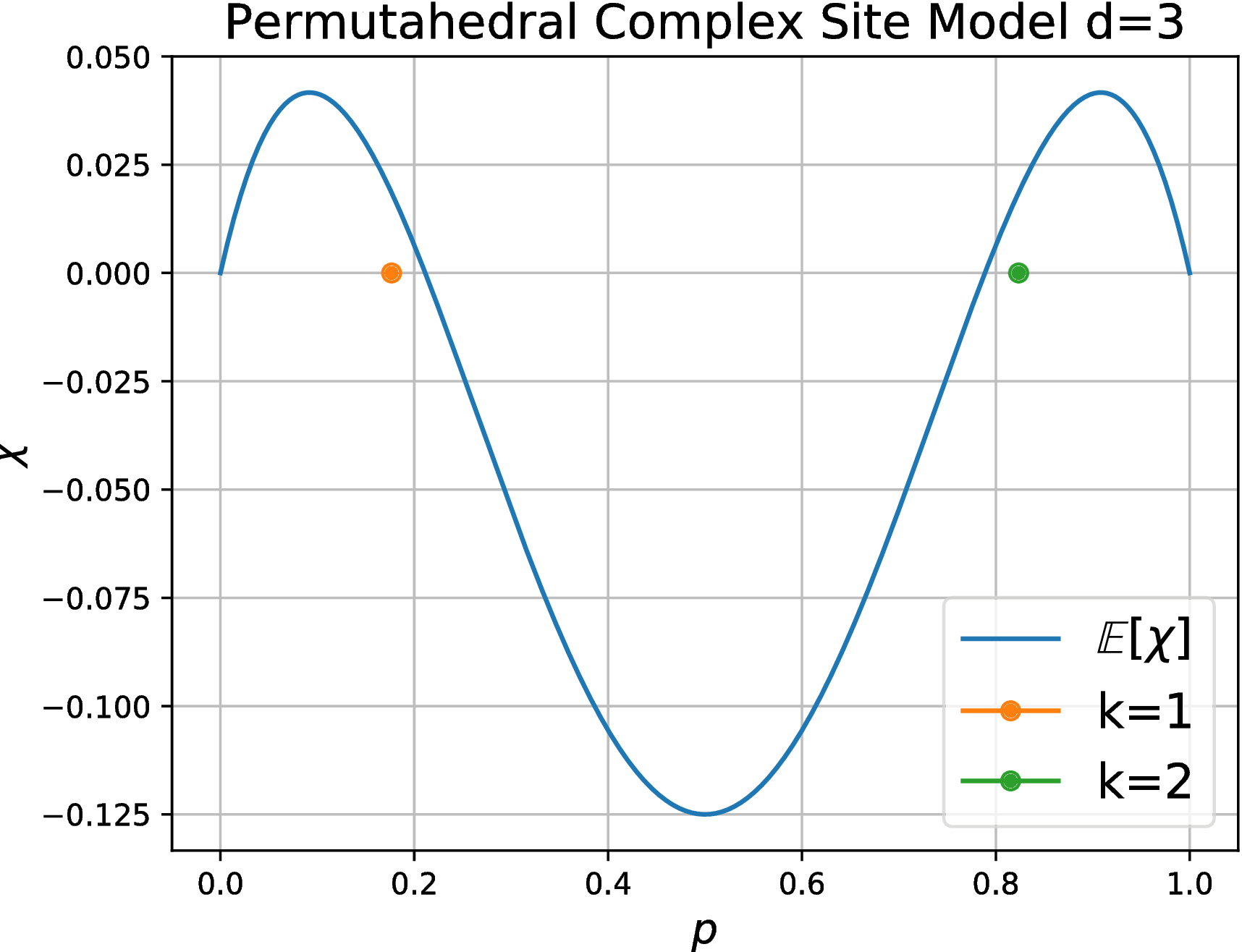}
         \caption{}
         \label{fig:perm3}
     \end{subfigure}
     \hfill
     \begin{subfigure}[b]{0.3\textwidth}
         \centering
				 \includegraphics[width=\textwidth]{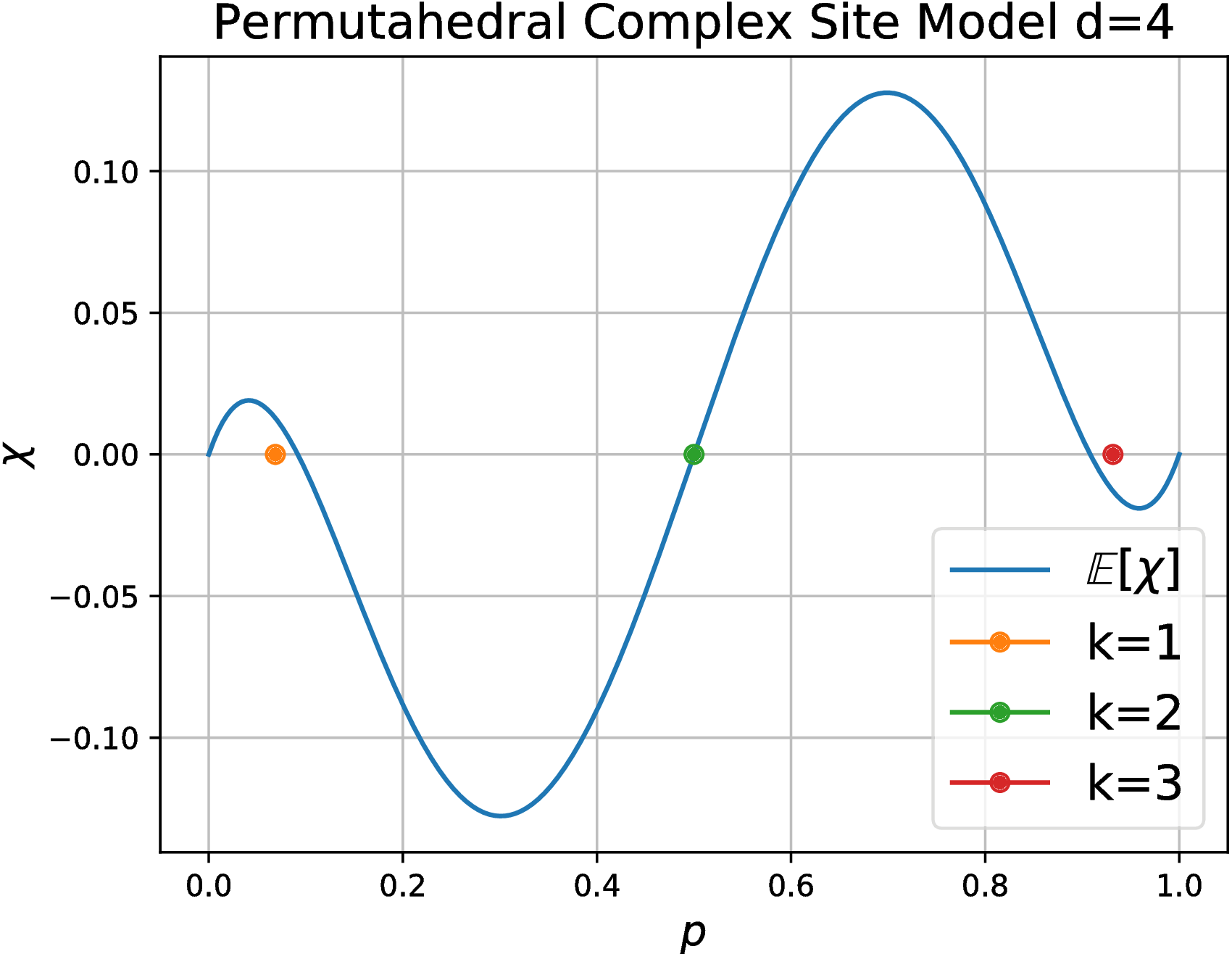}
				 \caption{}
         \label{fig:perm4}
     \end{subfigure}
     \begin{subfigure}[b]{0.3\textwidth}
         \centering
				 \includegraphics[width=\textwidth]{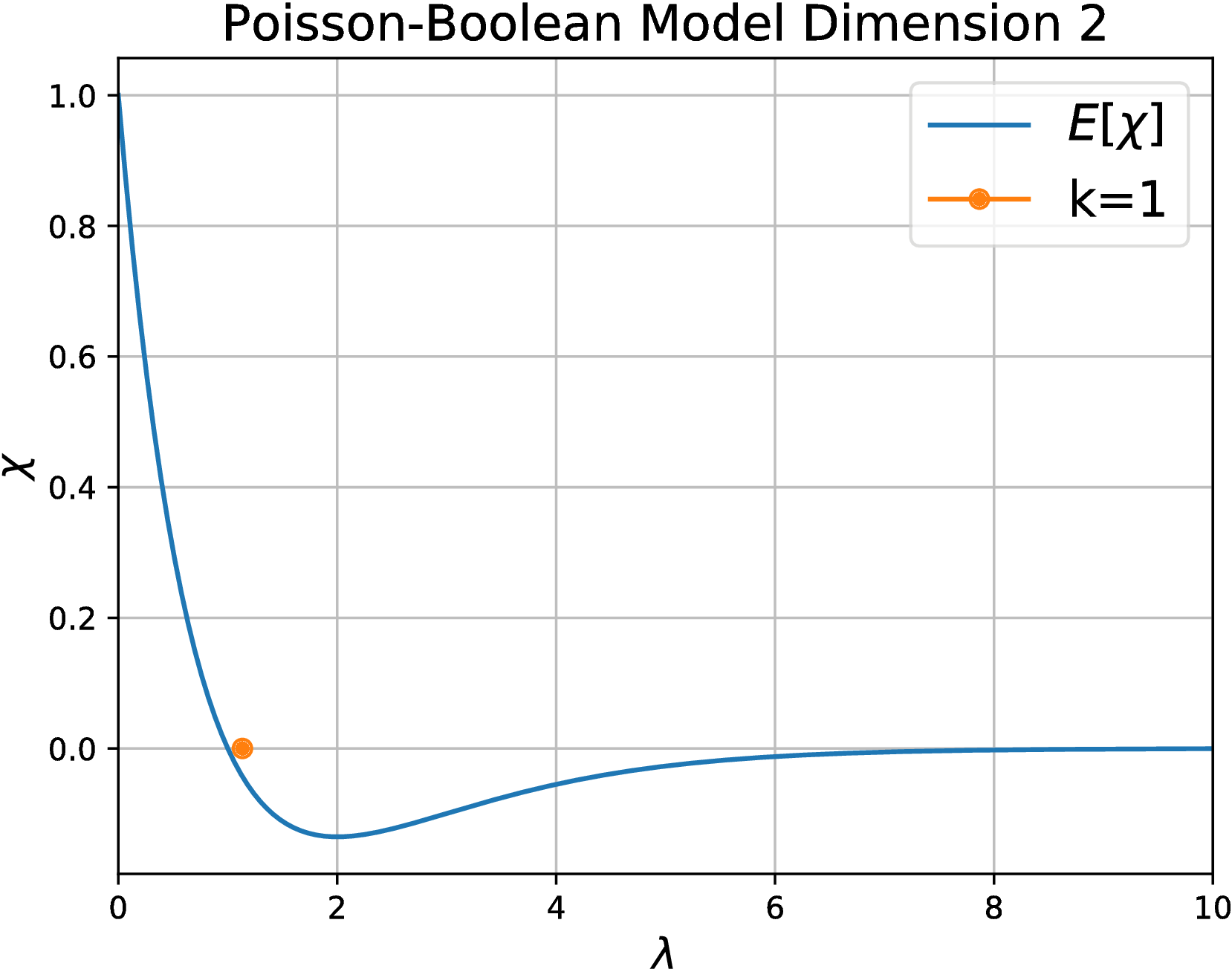}
				 \caption{}
         \label{fig:boolean2}
     \end{subfigure}
     \hfill
     \begin{subfigure}[b]{0.3\textwidth}
         \centering
				 \includegraphics[width=\textwidth]{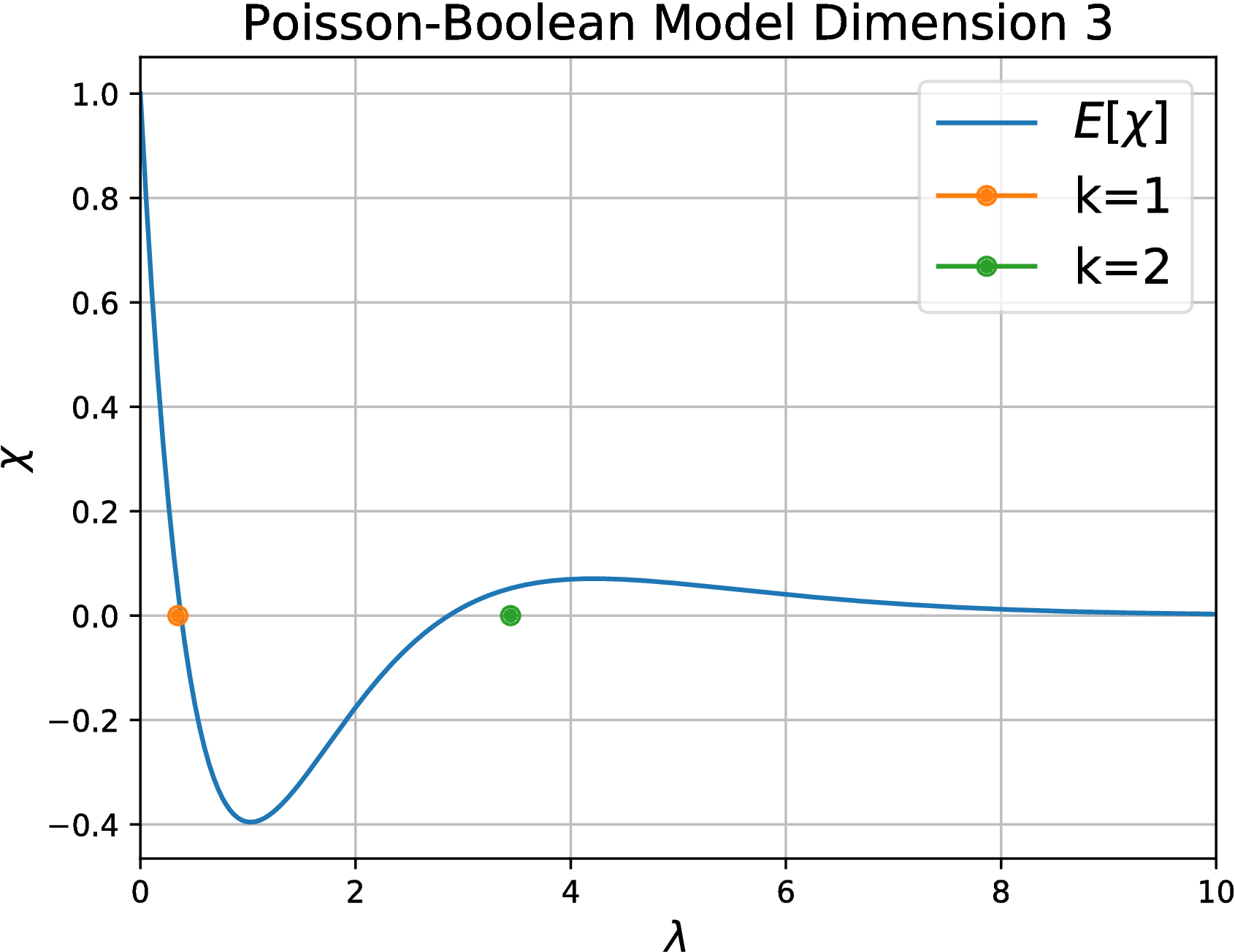}
				 \caption{}
         \label{fig:boolean3}
     \end{subfigure}
     \hfill
     \begin{subfigure}[b]{0.3\textwidth}
         \centering
				 \includegraphics[width=\textwidth]{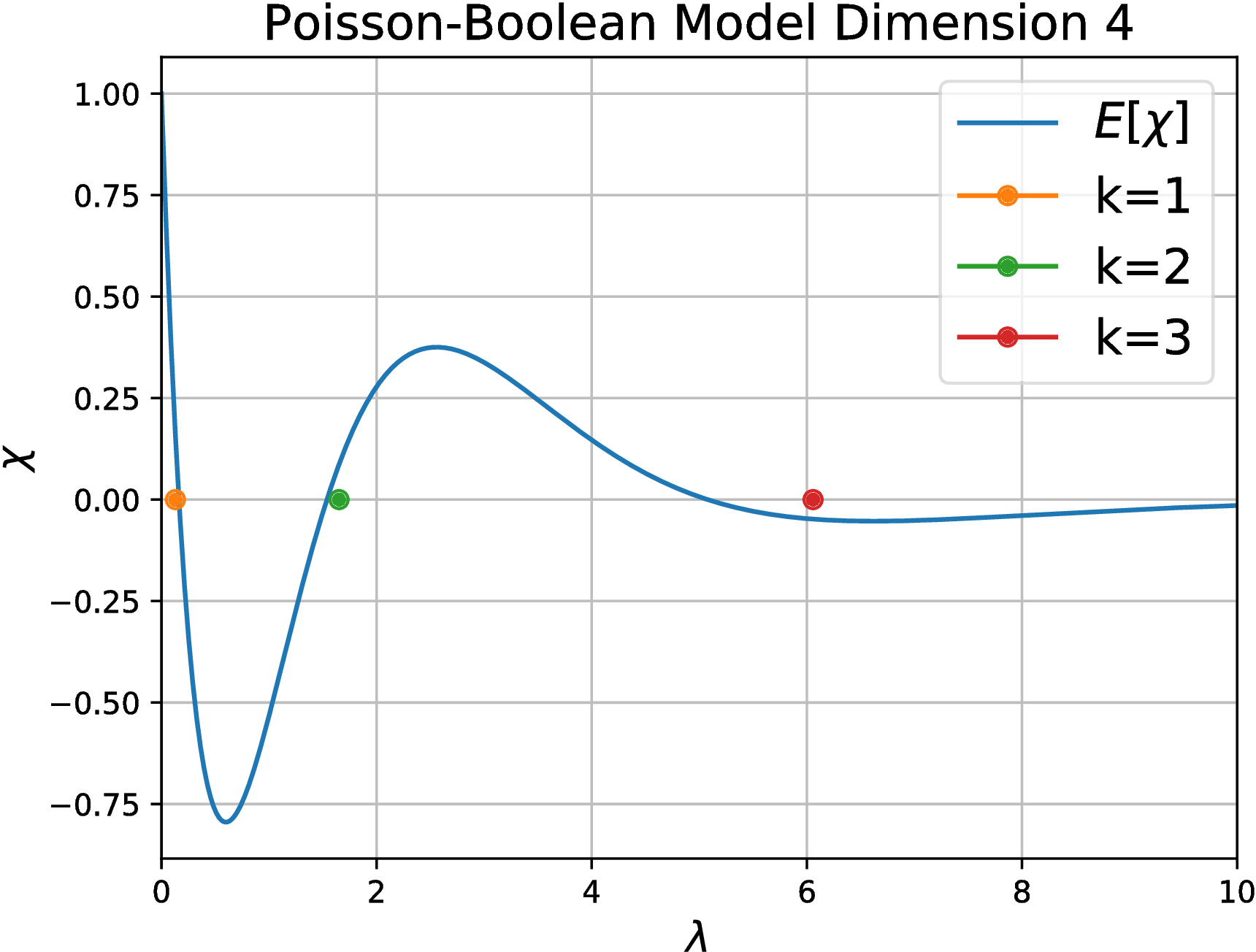}
         \caption{}
         \label{fig:boolean4}
     \end{subfigure}
     \begin{subfigure}[b]{0.3\textwidth}
         \centering
				 \includegraphics[width=\textwidth]{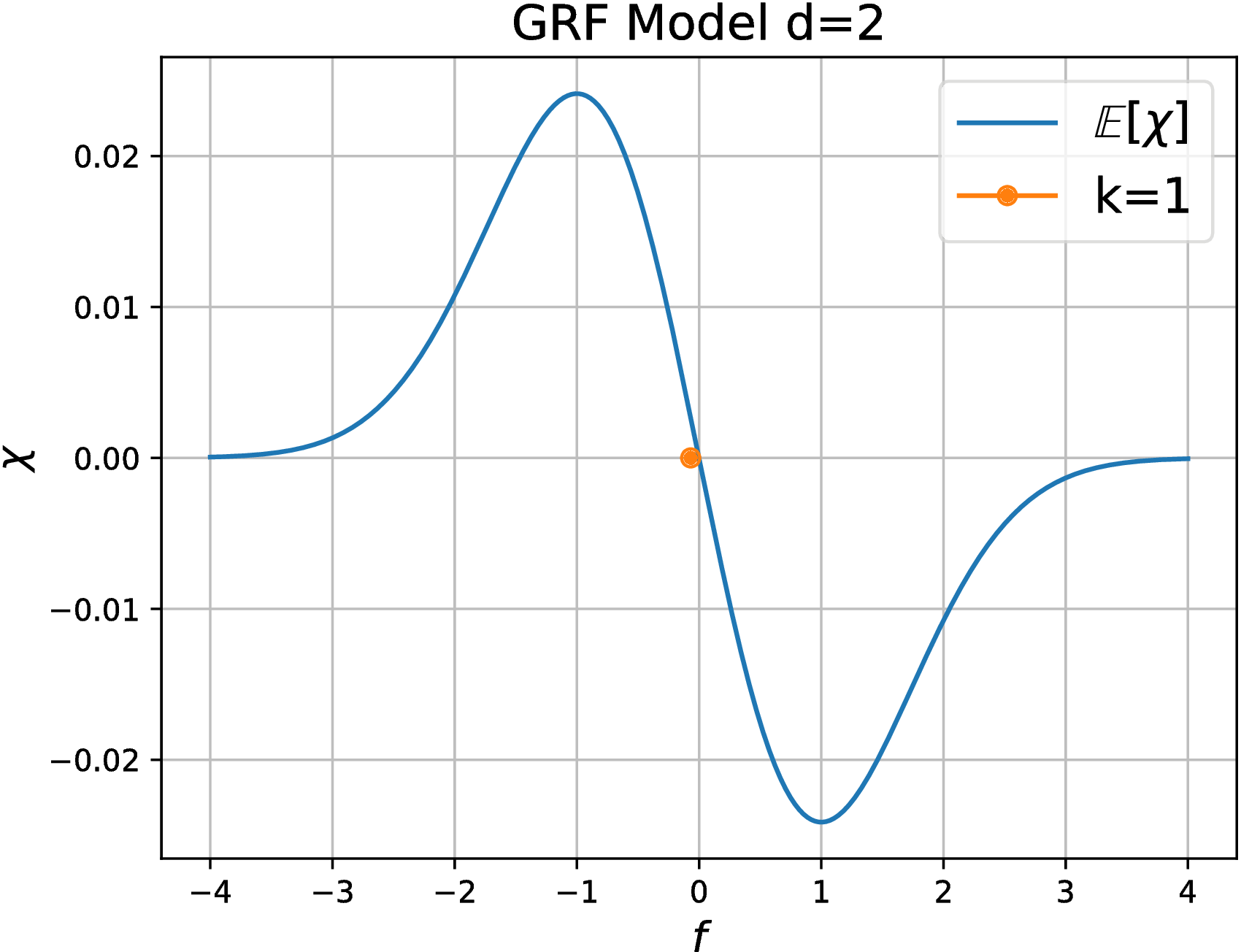}
         \caption{}
         \label{fig:gaussian2}
     \end{subfigure}
     \hfill
     \begin{subfigure}[b]{0.3\textwidth}
         \centering
				 \includegraphics[width=\textwidth]{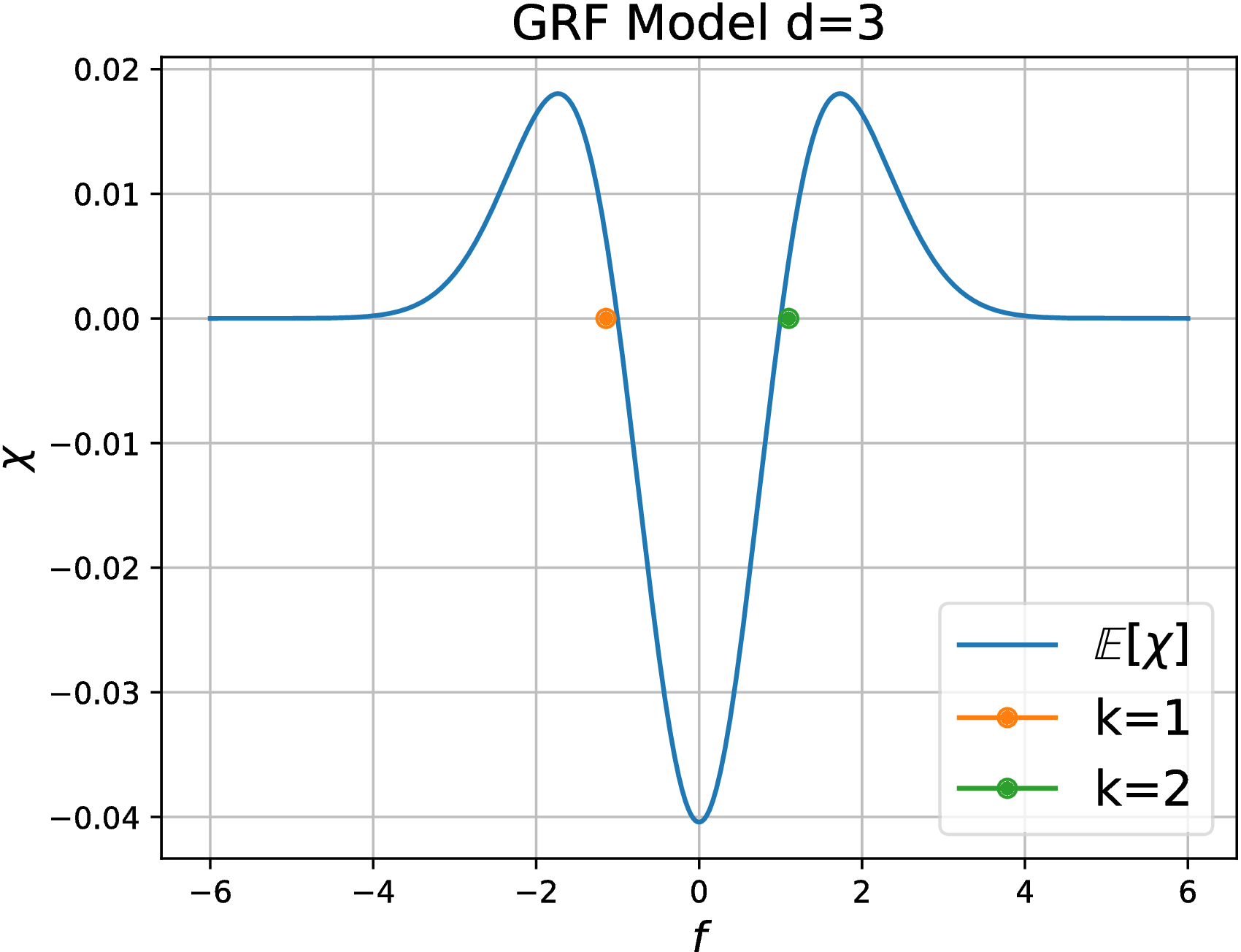}
         \caption{}
         \label{fig:gaussian3}
     \end{subfigure}
     \hfill
     \begin{subfigure}[b]{0.3\textwidth}
         \centering
				 \includegraphics[width=\textwidth]{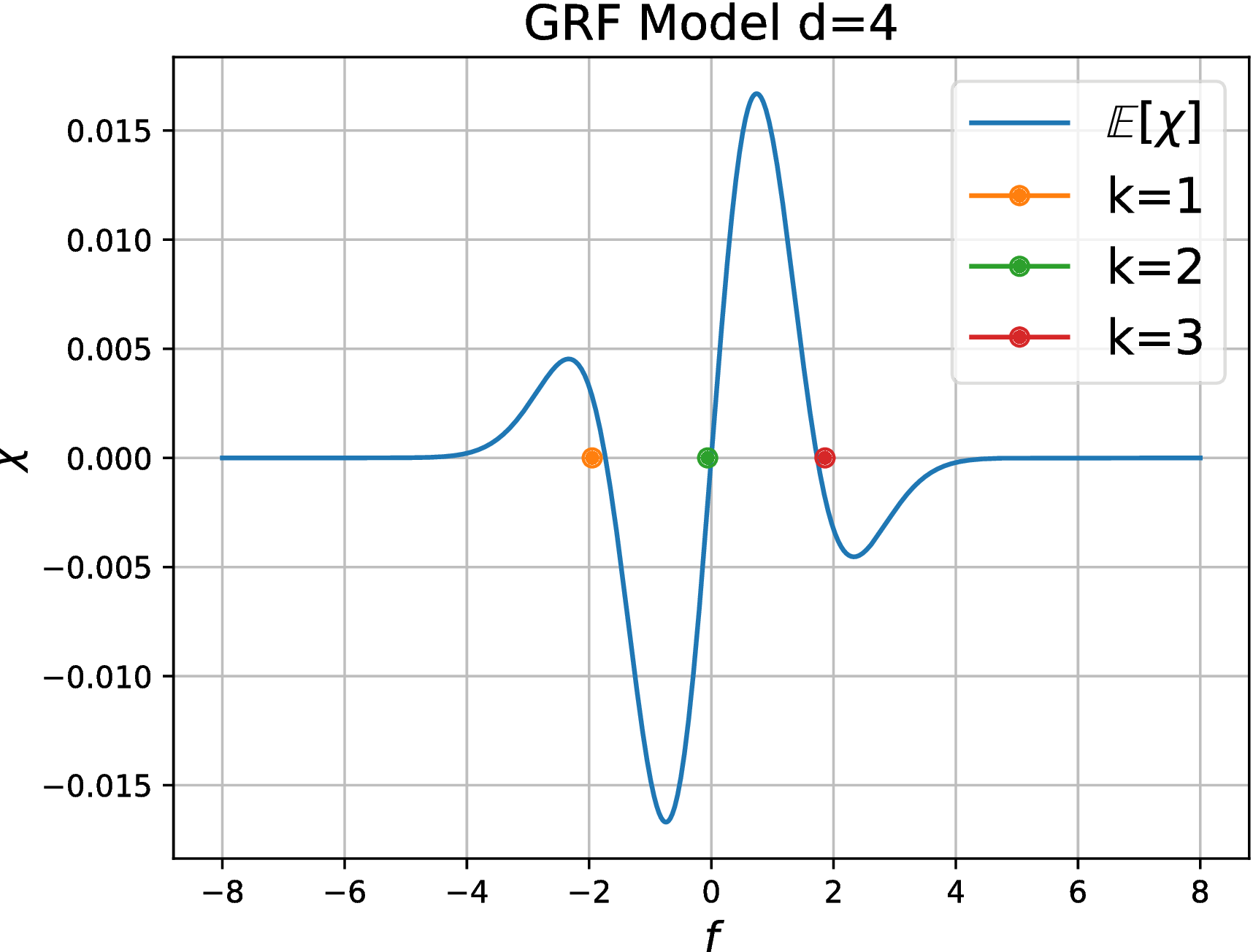}
         \caption{}
         \label{fig:gaussian4}
     \end{subfigure}
     \caption{\label{fig:ec_curves} The expected EC curve and the giant cycles.
     In each plot we draw the expected EC curve (solid line), along with the birth time of the first giant $k$-cycle for $k=1,\ldots, d-1$ (dots). (a)-(c) The random cubical complex. (d)-(f) The random permutahedral complex. (g)-(i) The Boolean model. (j)-(l) The Gaussian random field. We simulated all the models on the $d$-dimensional torus, for $d=2,3,4$ (from left to right). }
\end{figure}
%%%%%

%%%%%
\begin{figure}[H]
	\centering
     \begin{subfigure}[b]{0.3\textwidth}
         \centering
         \includegraphics[width=\textwidth]{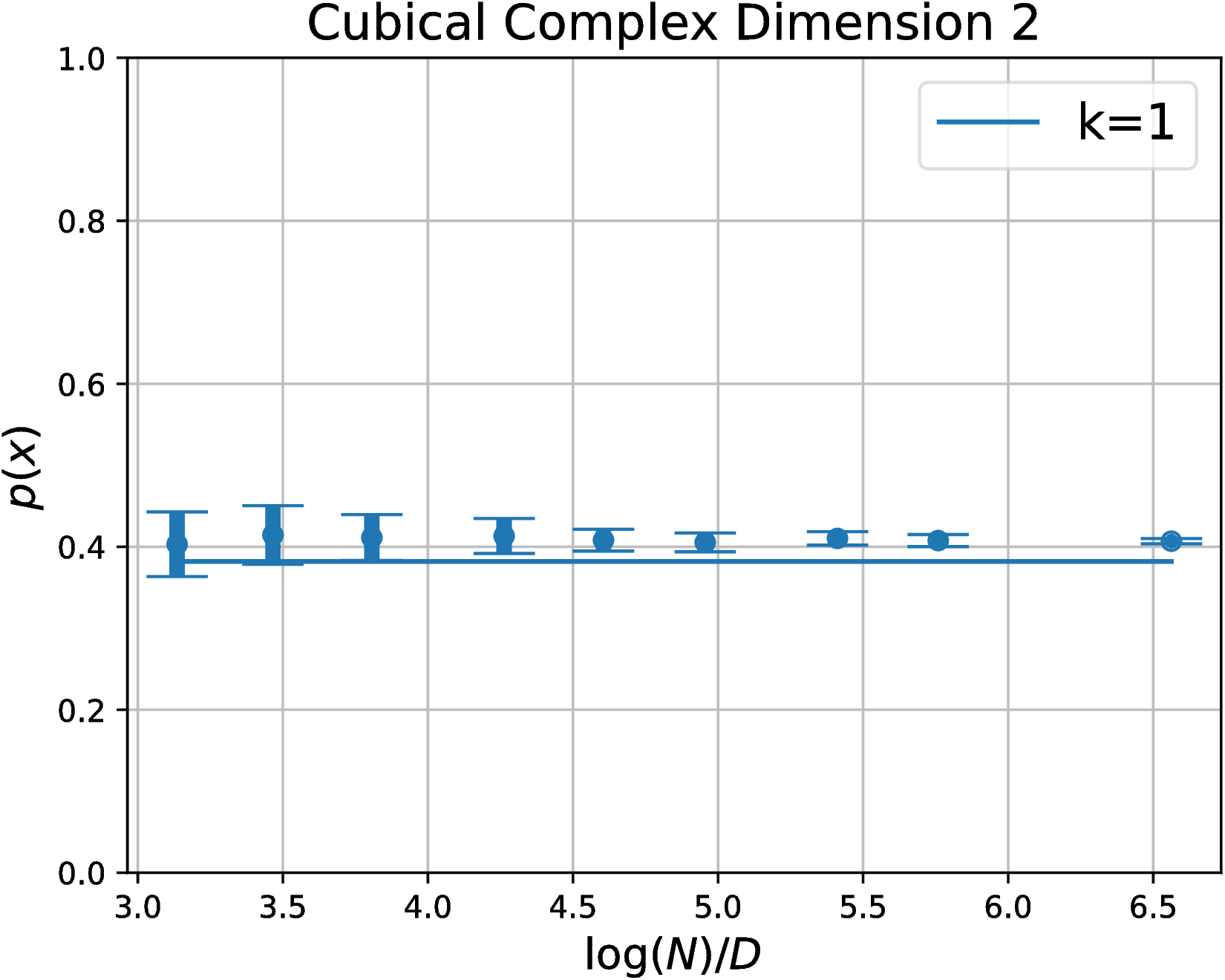}
         \caption{}
         \label{fig:uniformex2}
     \end{subfigure}
     \hfill
     \begin{subfigure}[b]{0.3\textwidth}
         \centering
				 \includegraphics[width=\textwidth]{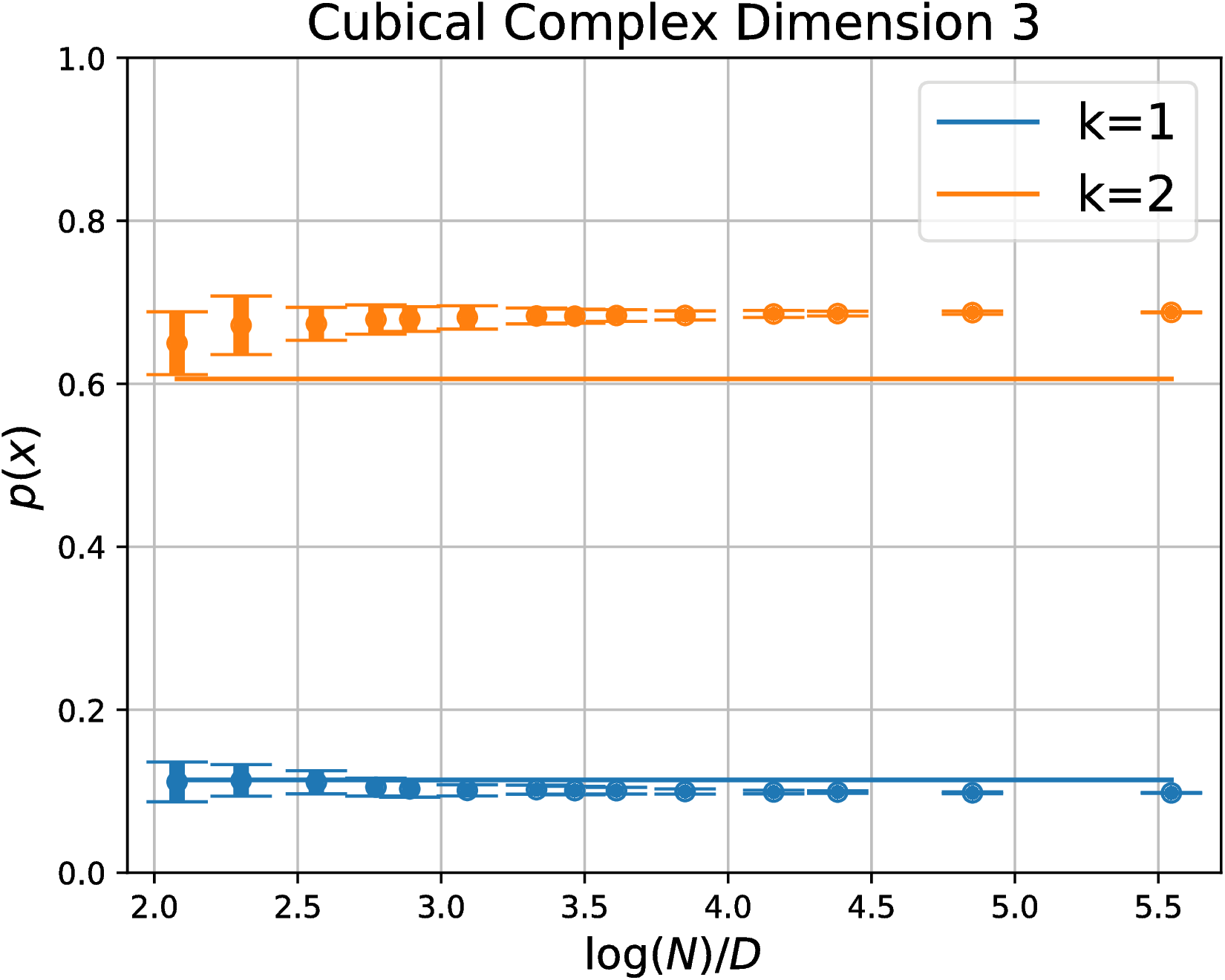}
         \caption{}
         \label{fig:uniformex3}
     \end{subfigure}
     \hfill
     \begin{subfigure}[b]{0.3\textwidth}
         \centering
				 \includegraphics[width=\textwidth]{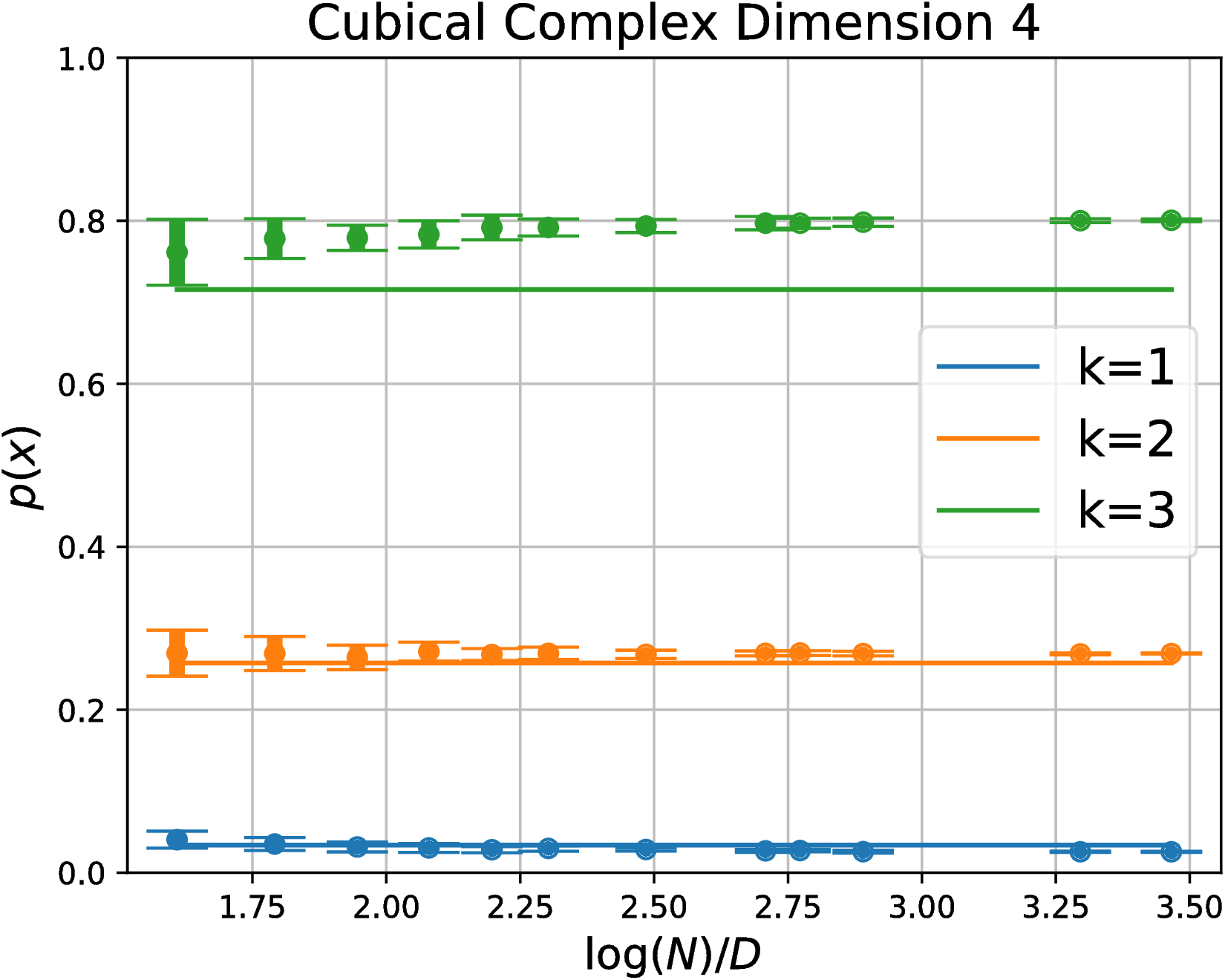}
         \caption{}
         \label{fig:uniformex4}
     \end{subfigure}

          \begin{subfigure}[b]{0.3\textwidth}
         \centering
         \includegraphics[width=\textwidth]{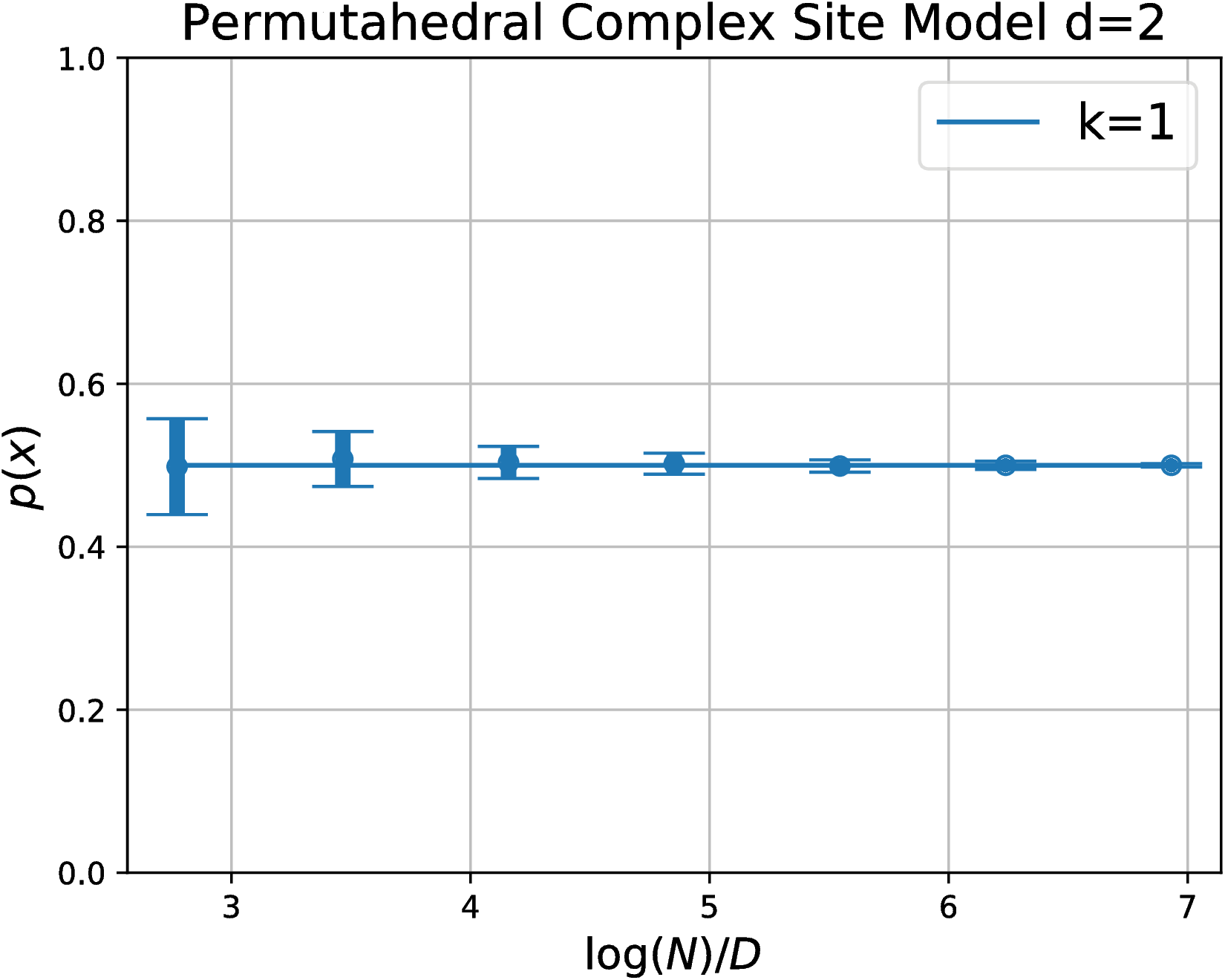}
         \caption{}
         \label{fig:permex2}
     \end{subfigure}
     \hfill
     \begin{subfigure}[b]{0.3\textwidth}
         \centering
				 \includegraphics[width=\textwidth]{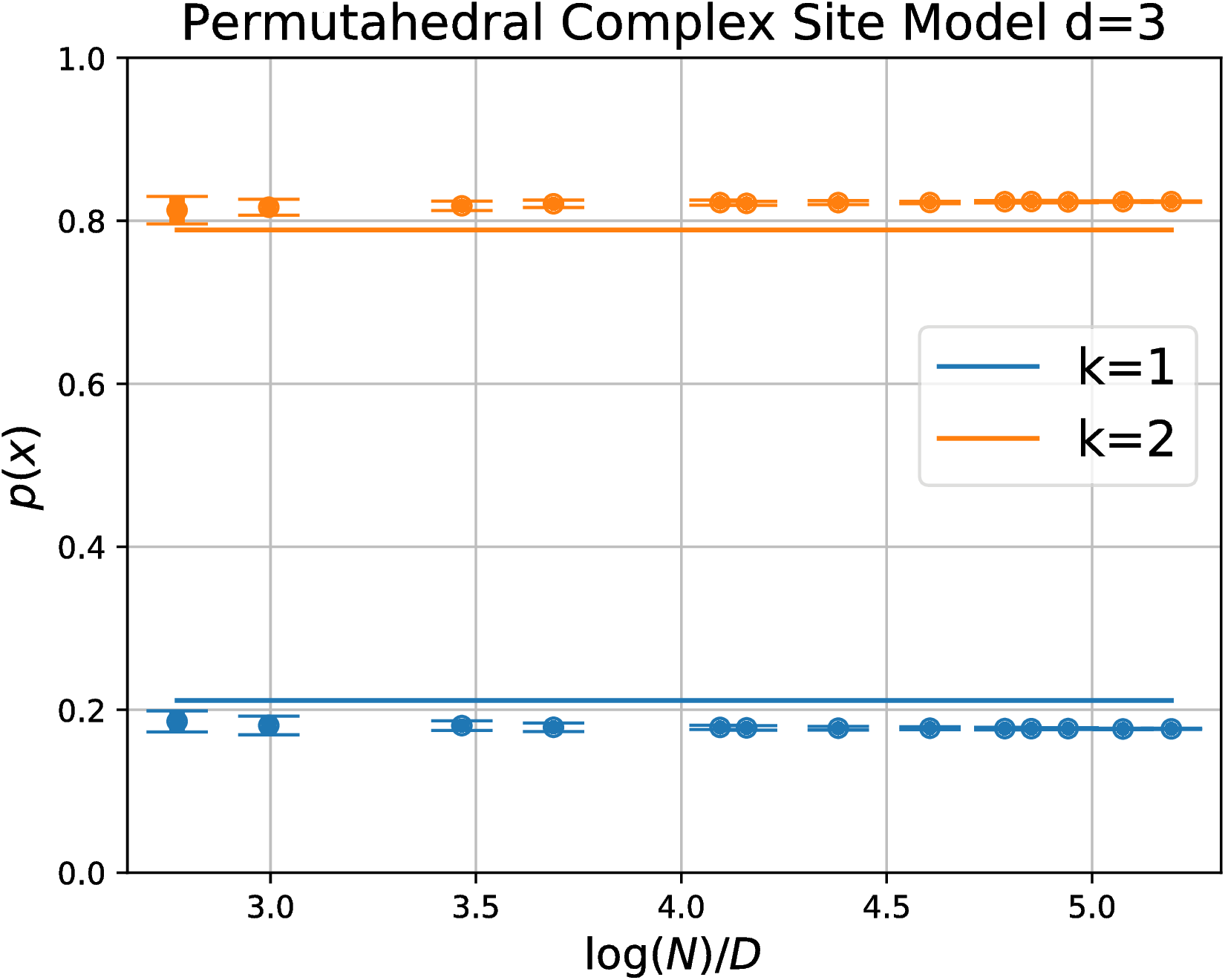}
         \caption{}
         \label{fig:permex3}
     \end{subfigure}
     \hfill
     \begin{subfigure}[b]{0.3\textwidth}
         \centering
				 \includegraphics[width=\textwidth]{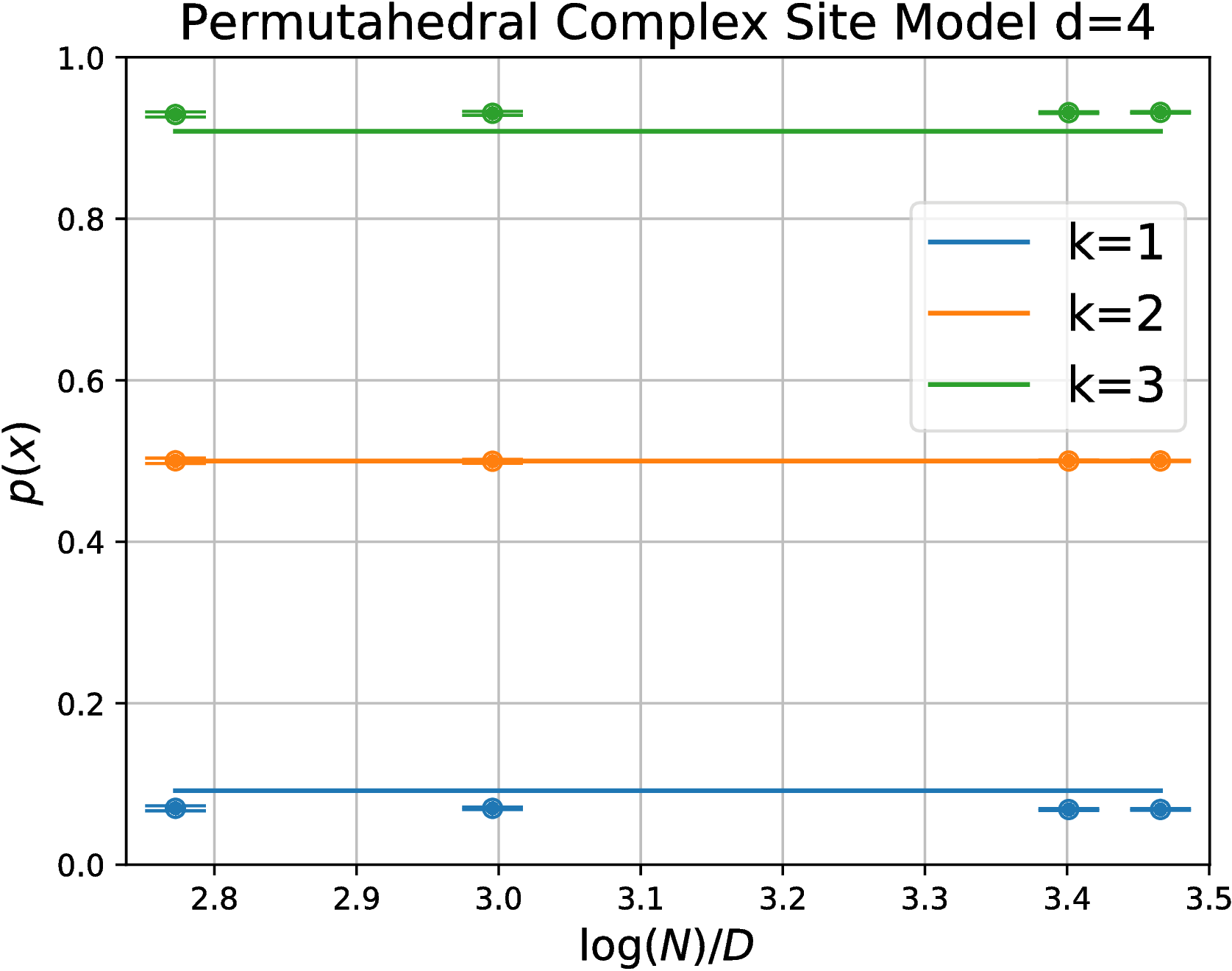}
         \caption{}
         \label{fig:permex4}
     \end{subfigure}

     \begin{subfigure}[b]{0.3\textwidth}
         \centering
         \includegraphics[width=\textwidth]{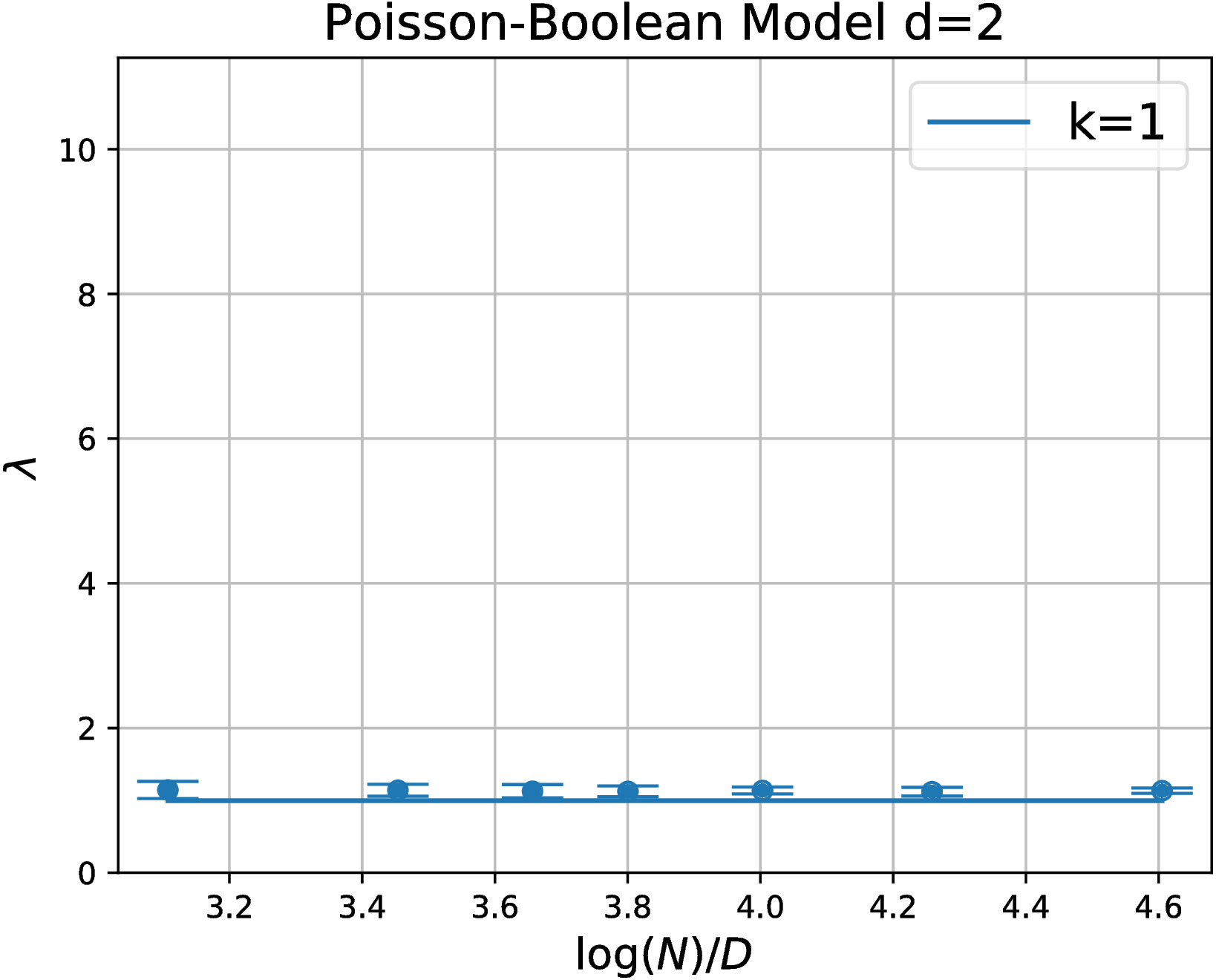}
         \caption{}
         \label{fig:boolex2}
     \end{subfigure}
     \hfill
     \begin{subfigure}[b]{0.3\textwidth}
         \centering
				 \includegraphics[width=\textwidth]{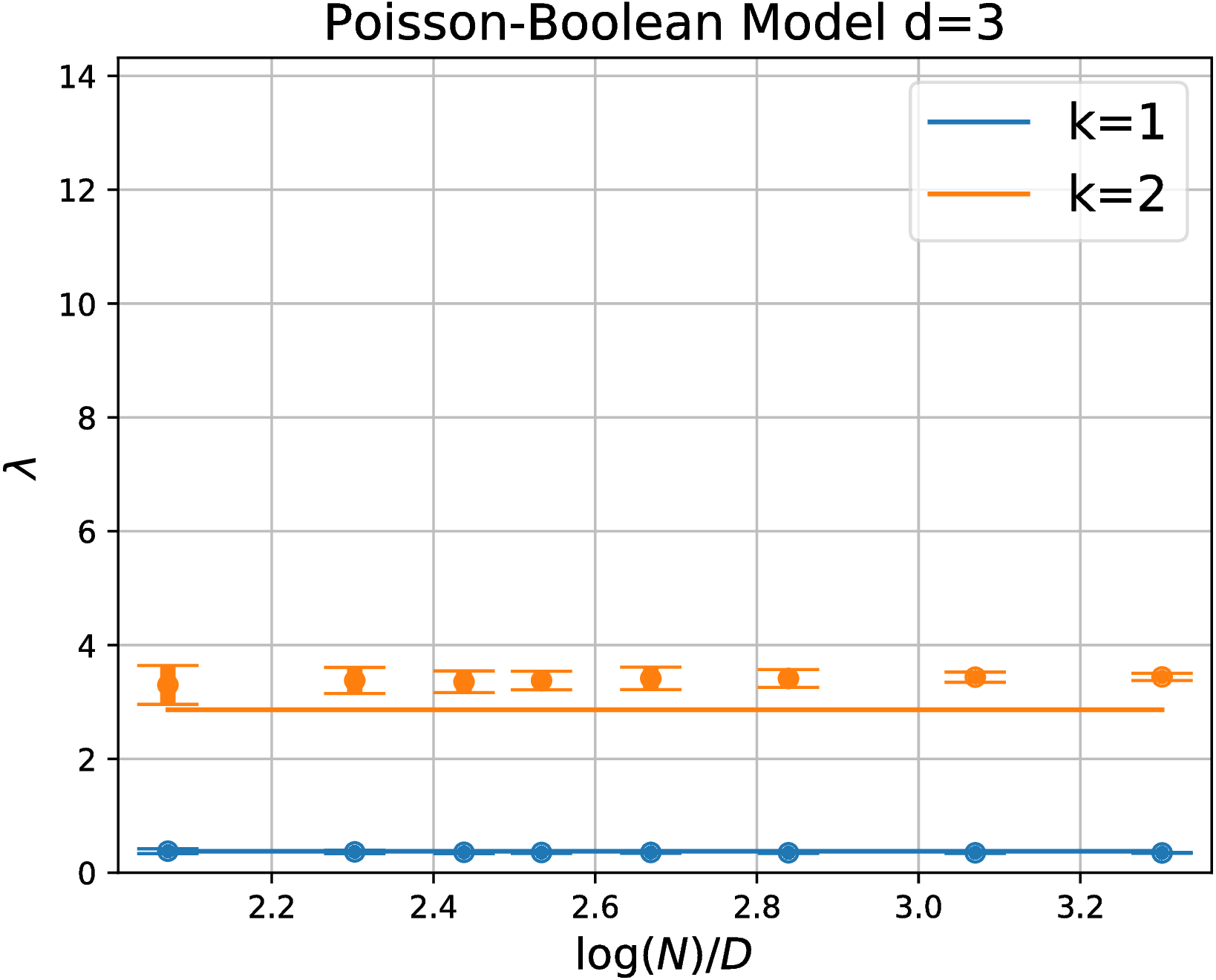}
         \caption{}
         \label{fig:boolex3}
     \end{subfigure}
     \hfill
     \begin{subfigure}[b]{0.3\textwidth}
         \centering
				 \includegraphics[width=\textwidth]{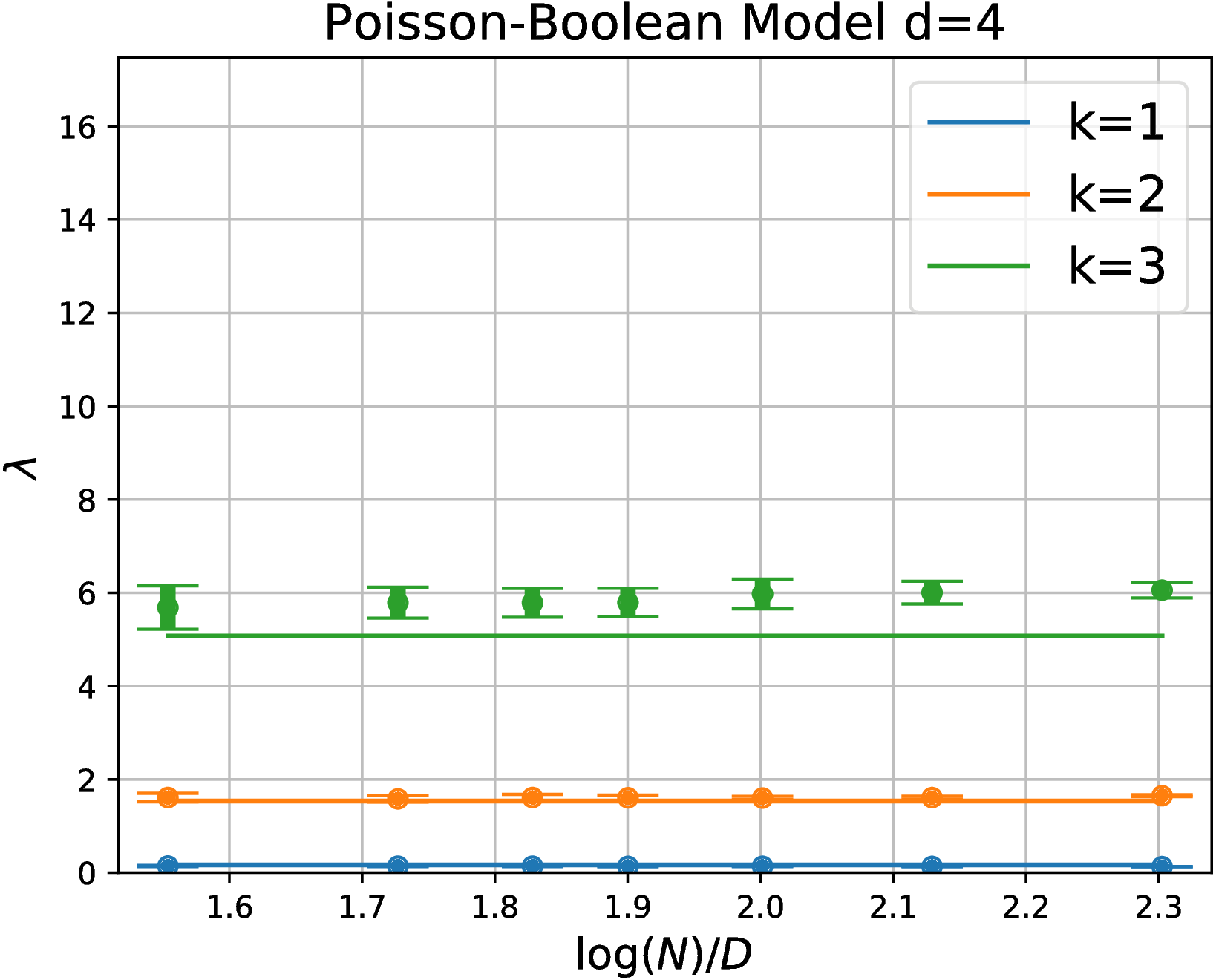}
         \caption{}
         \label{fig:boolex4}
     \end{subfigure}

     \begin{subfigure}[b]{0.3\textwidth}
         \centering
         \includegraphics[width=\textwidth]{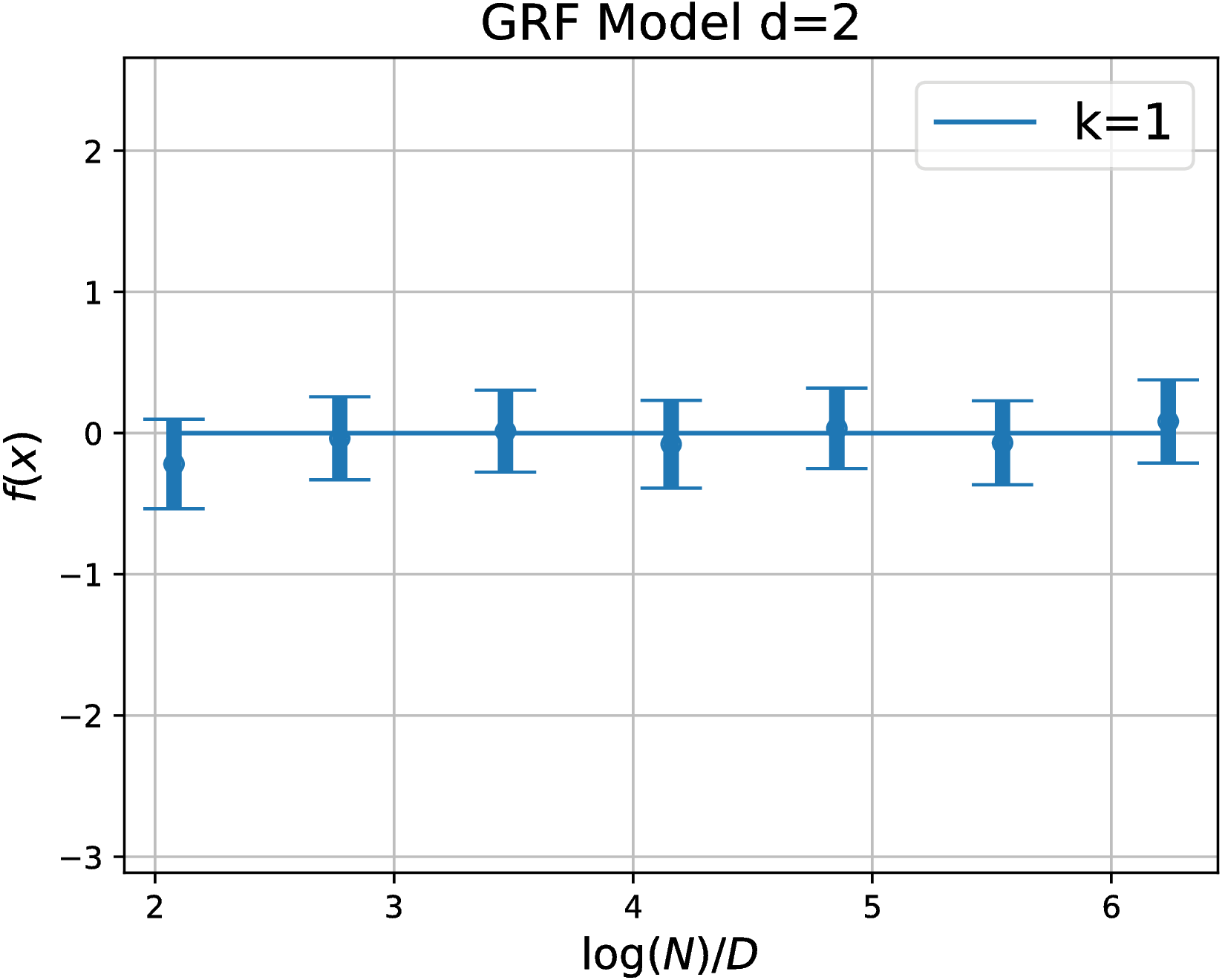}
         \caption{}
         \label{fig:gaussex2}
     \end{subfigure}
     \hfill
     \begin{subfigure}[b]{0.3\textwidth}
         \centering
				 \includegraphics[width=\textwidth]{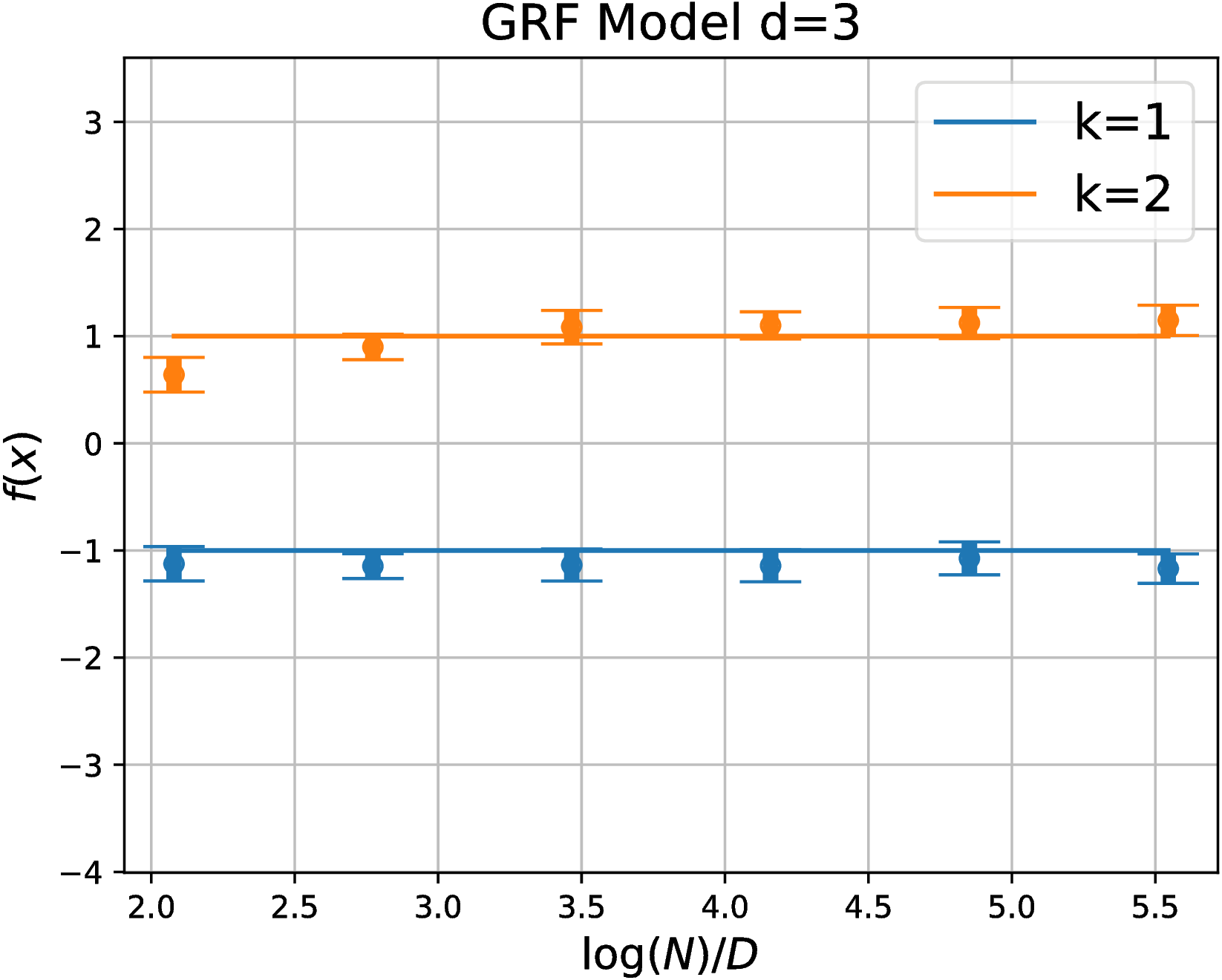}
         \caption{}
         \label{fig:gaussex3}
     \end{subfigure}
     \hfill
     \begin{subfigure}[b]{0.3\textwidth}
         \centering
				 \includegraphics[width=\textwidth]{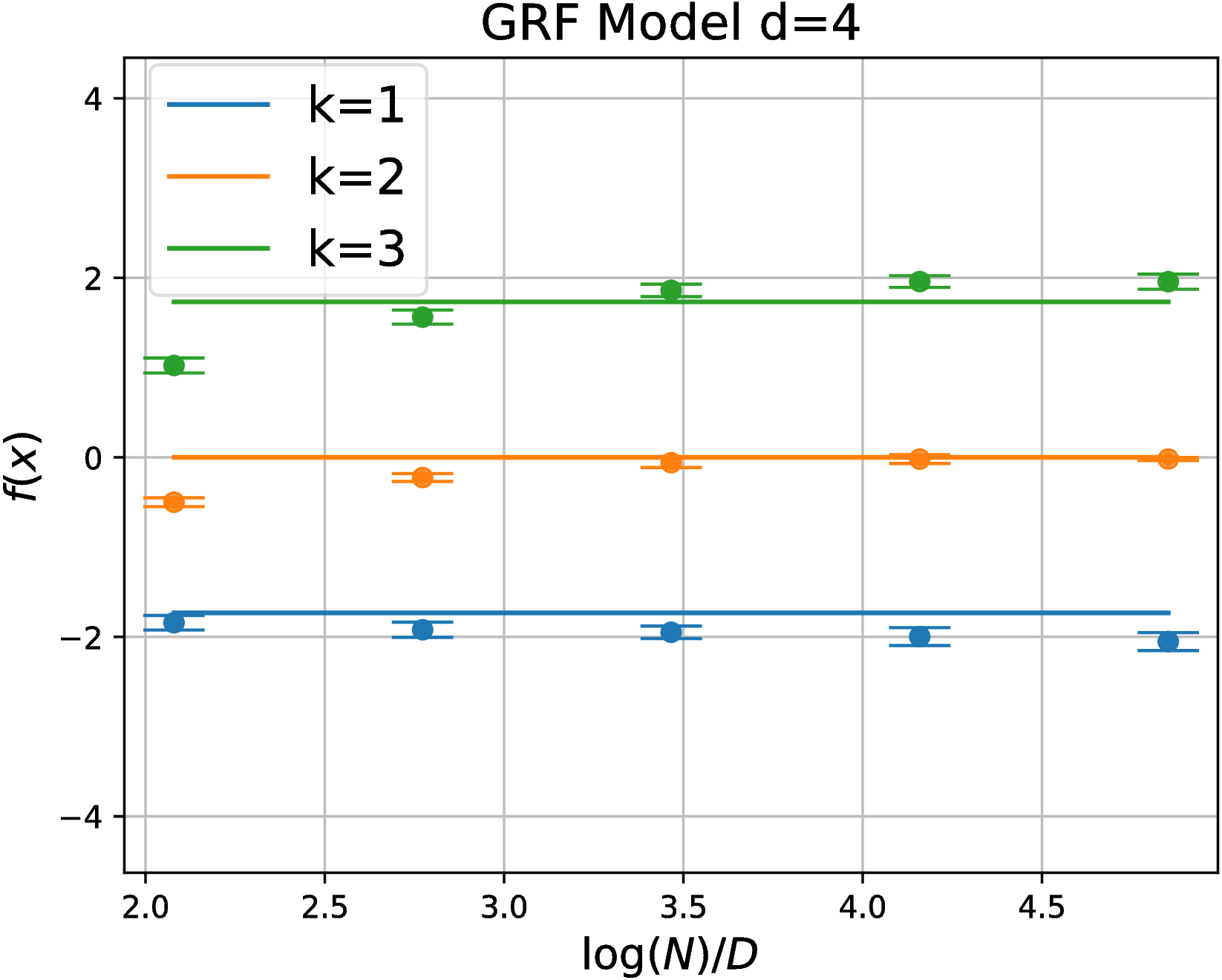}
         \caption{}
         \label{fig:gaussex4}
     \end{subfigure}
     \caption{\label{fig:stats} Statistics for the birth time of the giant $k$-cycles. For each model we repeated the simulations in order to estimate the mean and variance of the birth time of the first giant $k$-cycle. In each figure the x-axis is  $\log n$, and the y-axis represent the parameter value ($p,\lambda,$ or $\alpha$). The dots represent the mean value estimate, and the bars around them follows the standard deviation. The horizontal lines mark the corresponding zeros of the EC curve.
     (a)-(c) The random cubical complex. (d)-(f) The random permutahedral complex. (g)-(i) The Boolean model. (j)-(l) The Gaussian random field.}
\end{figure}
%%%%%

To give an alternative view on the error term $\Delta_k$ in Figure~\ref{fig:zero_vs_perc} we show a scatter plot of zeros of the EC curve and the individual appearance of the giant cycles for each of the four models. Note that each plot includes many different values of $n$, and the large spread in values is due to smaller values of $n$.
% Here we see that the overall spread is larger in higher dimensions, the small standard deviation in Figure~\ref{fig:stats} implies this is due to  samples at smaller grid sizes.

%%%%%
\begin{figure}[H]
	\centering
     \begin{subfigure}[t]{0.23\textwidth}
         \centering
         \includegraphics[width=\textwidth]{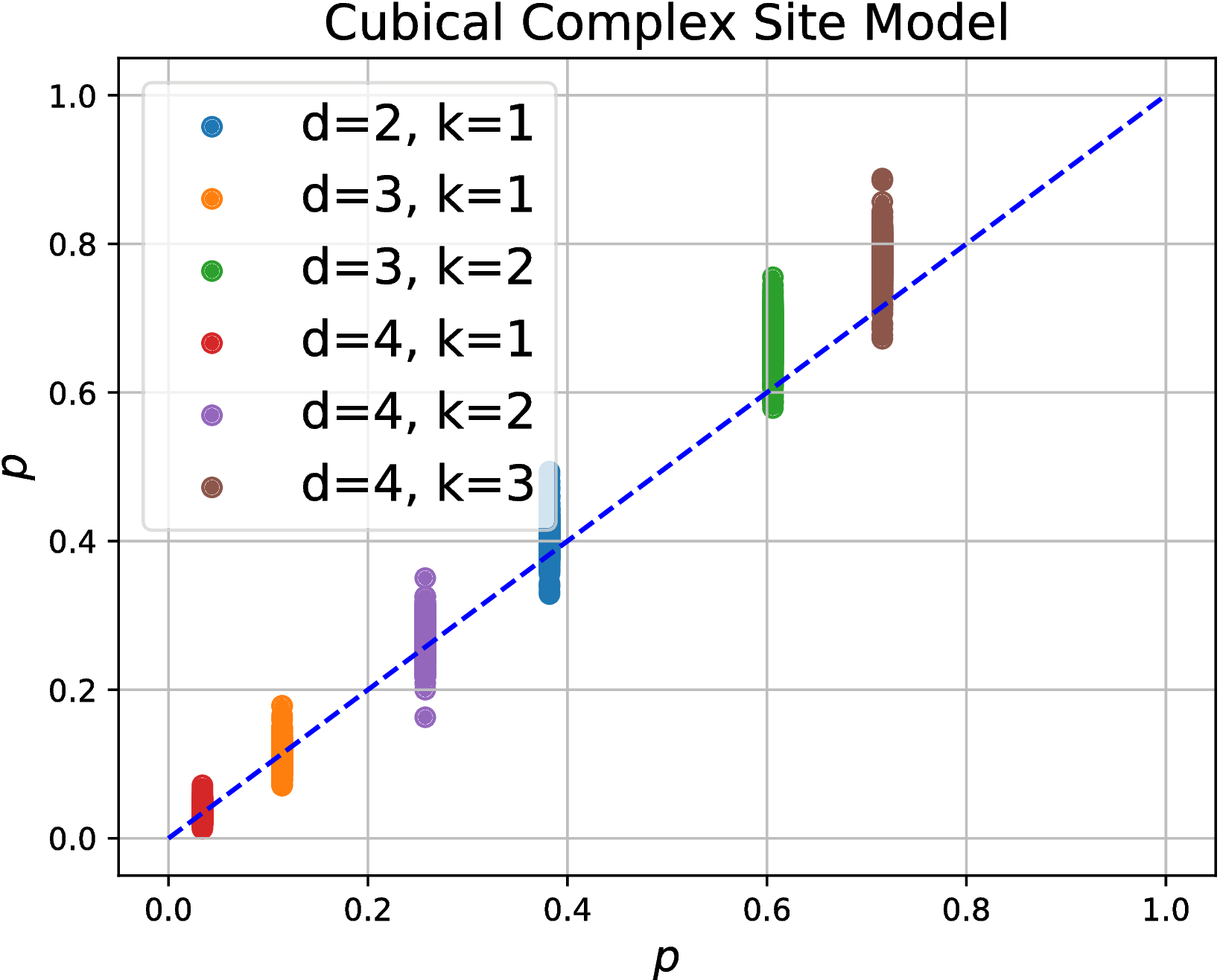}
         \caption{}
         \label{fig:puniform_all}
     \end{subfigure}
     \hfill
     \begin{subfigure}[t]{0.23\textwidth}
         \centering
				 \includegraphics[width=\textwidth]{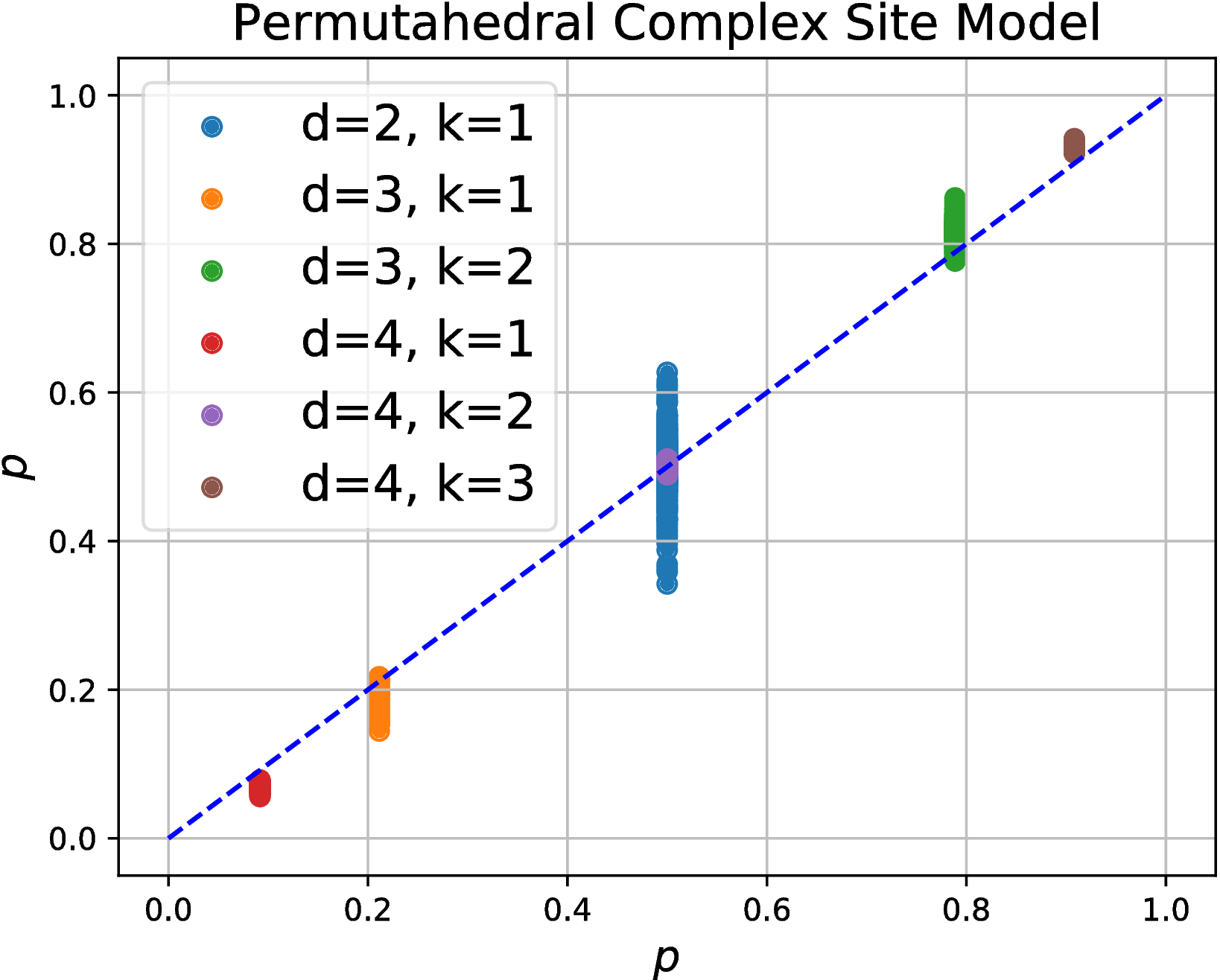}
         \caption{}
         \label{fig:pperm_all}
     \end{subfigure}
     \hfill
     \begin{subfigure}[t]{0.23\textwidth}
         \centering
				 \includegraphics[width=\textwidth]{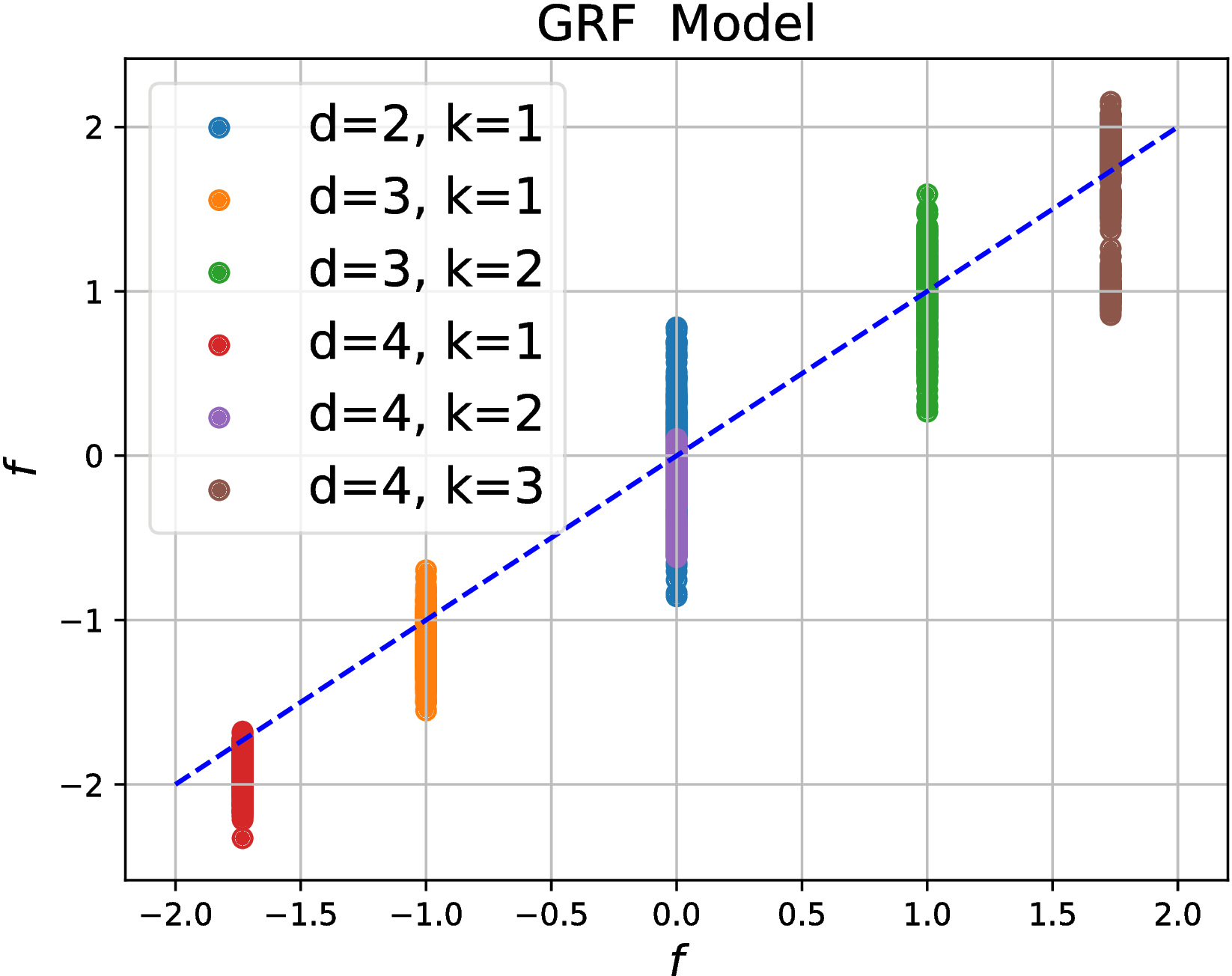}
         \caption{}
         \label{fig:pgaussian_all}
     \end{subfigure}
		 \hfill
     \begin{subfigure}[t]{0.23\textwidth}
         \centering
				 \includegraphics[width=\textwidth]{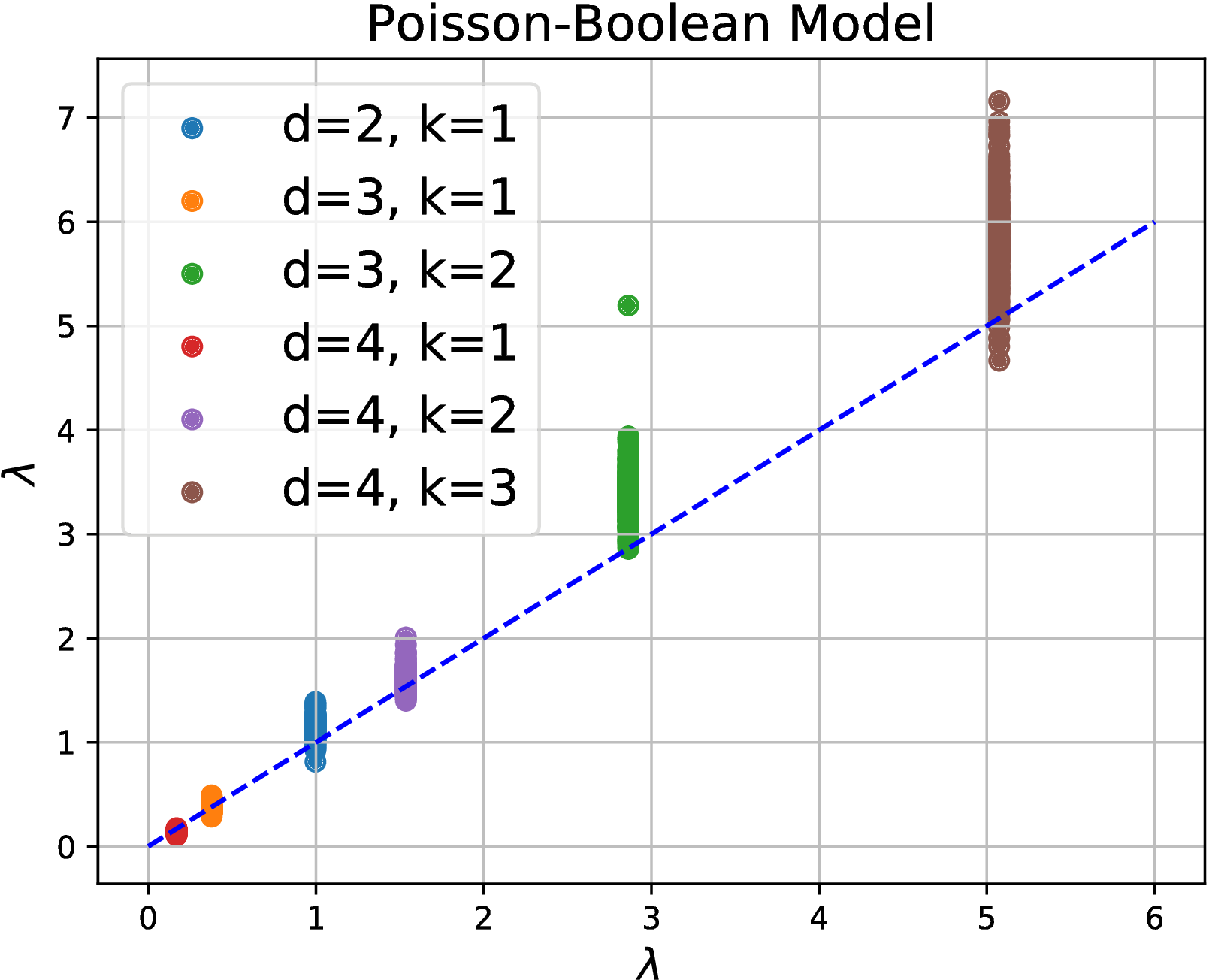}
         \caption{}
         \label{fig:pbool_all}
     \end{subfigure}
        \caption{Appearance of giant cycles vs.~zeros of the expected EC curve. The x-coordinate of each point is the corresponding zero of the expected EC curve, and the y-coordinate is the value when a giant $k$-cycle appears,
 (a) Cubical site model (b) Permutahedral site model (c) Gaussian random field (d) Poisson-Boolean.}
        \label{fig:zero_vs_perc}
\end{figure}
%%%%%

\begin{figure}[ht!]
	\centering
     \begin{subfigure}[b]{0.3\textwidth}
         \centering
				 \includegraphics[width=\textwidth]{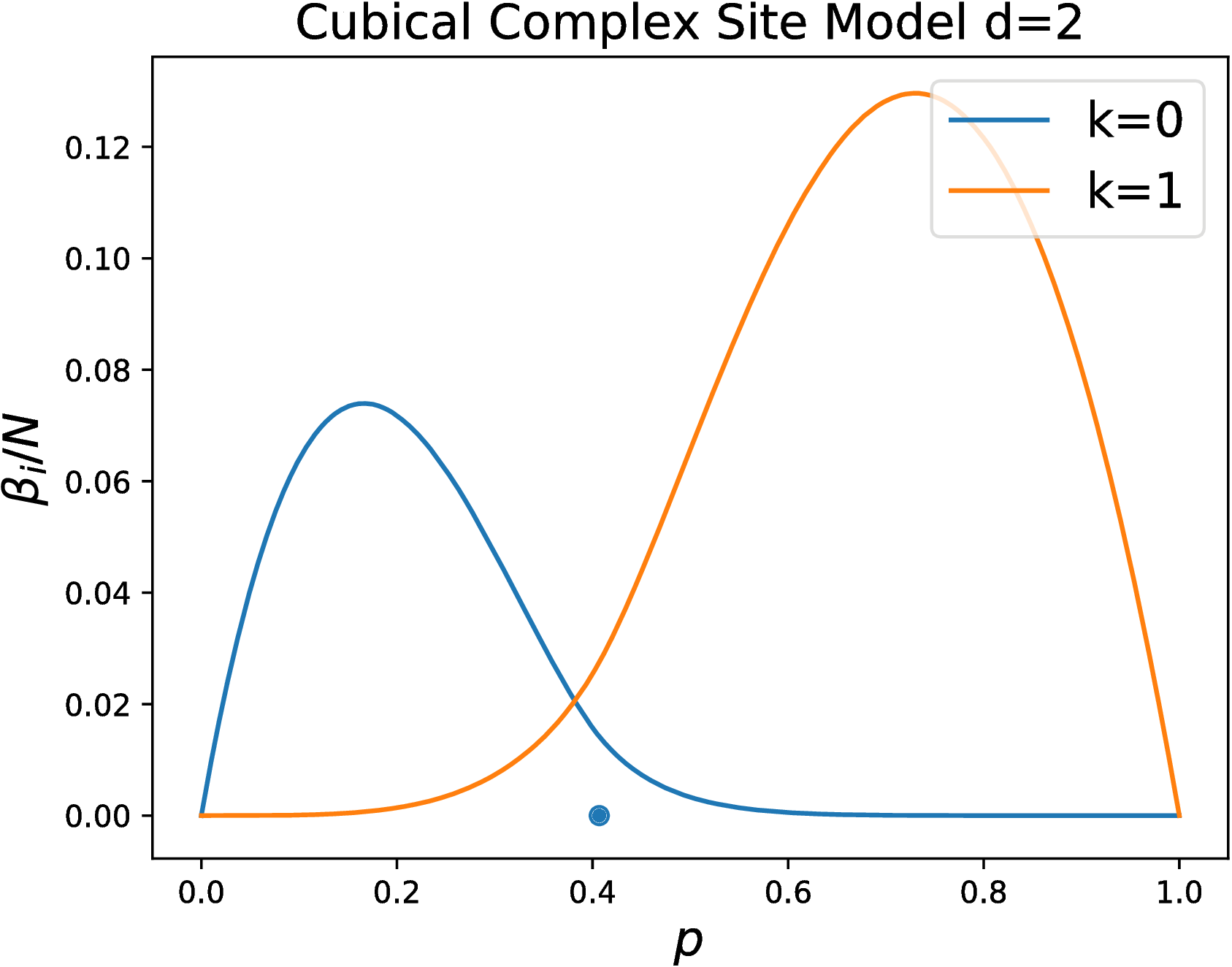}
				 \caption{}
         \label{fig:bettiuniform2}
     \end{subfigure}
     \hfill
     \begin{subfigure}[b]{0.3\textwidth}
         \centering
				 \includegraphics[width=\textwidth]{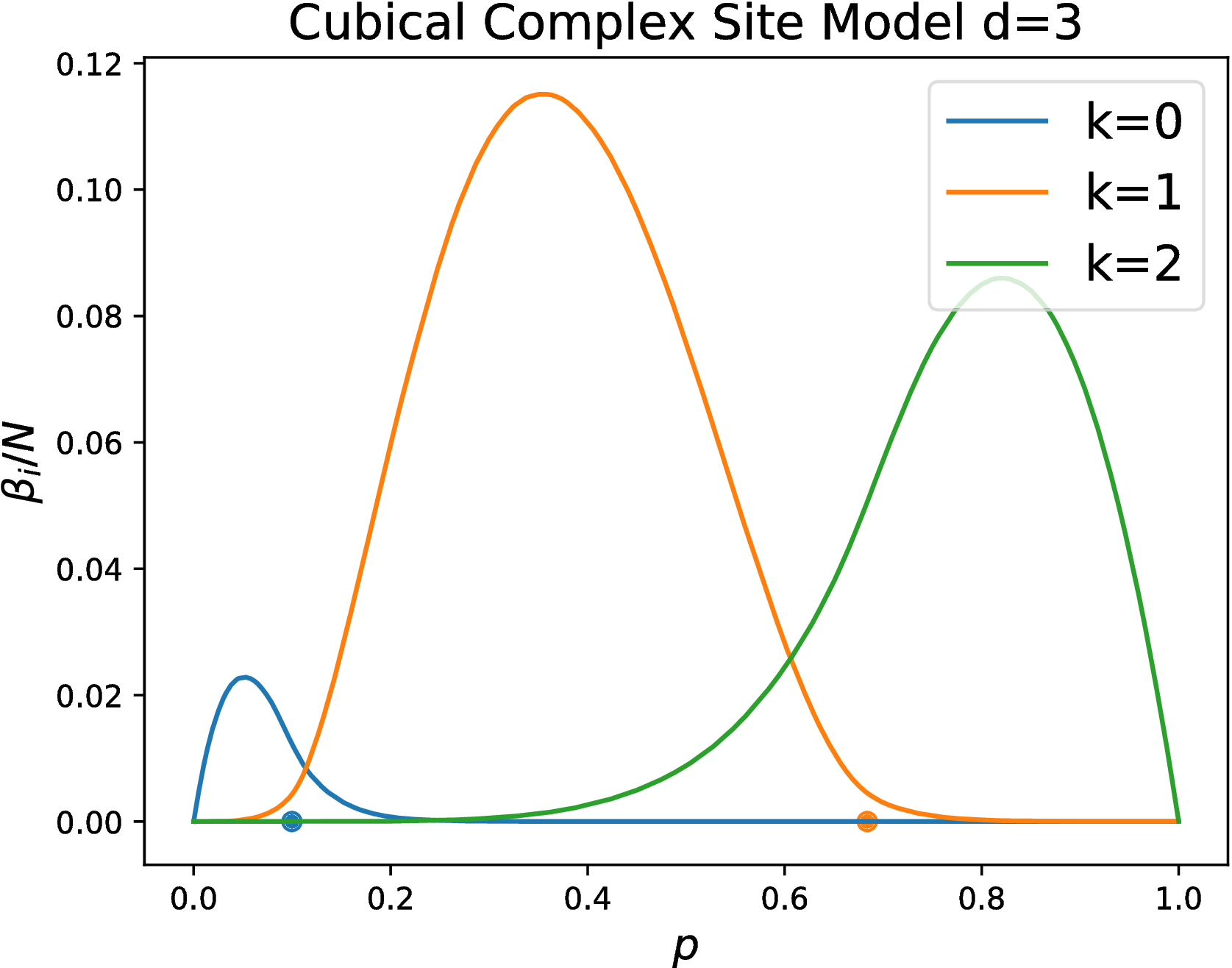}
				 \caption{}
         \label{fig:bettiuniform3}
     \end{subfigure}
     \hfill
     \begin{subfigure}[b]{0.3\textwidth}
         \centering

\includegraphics[width=\textwidth]{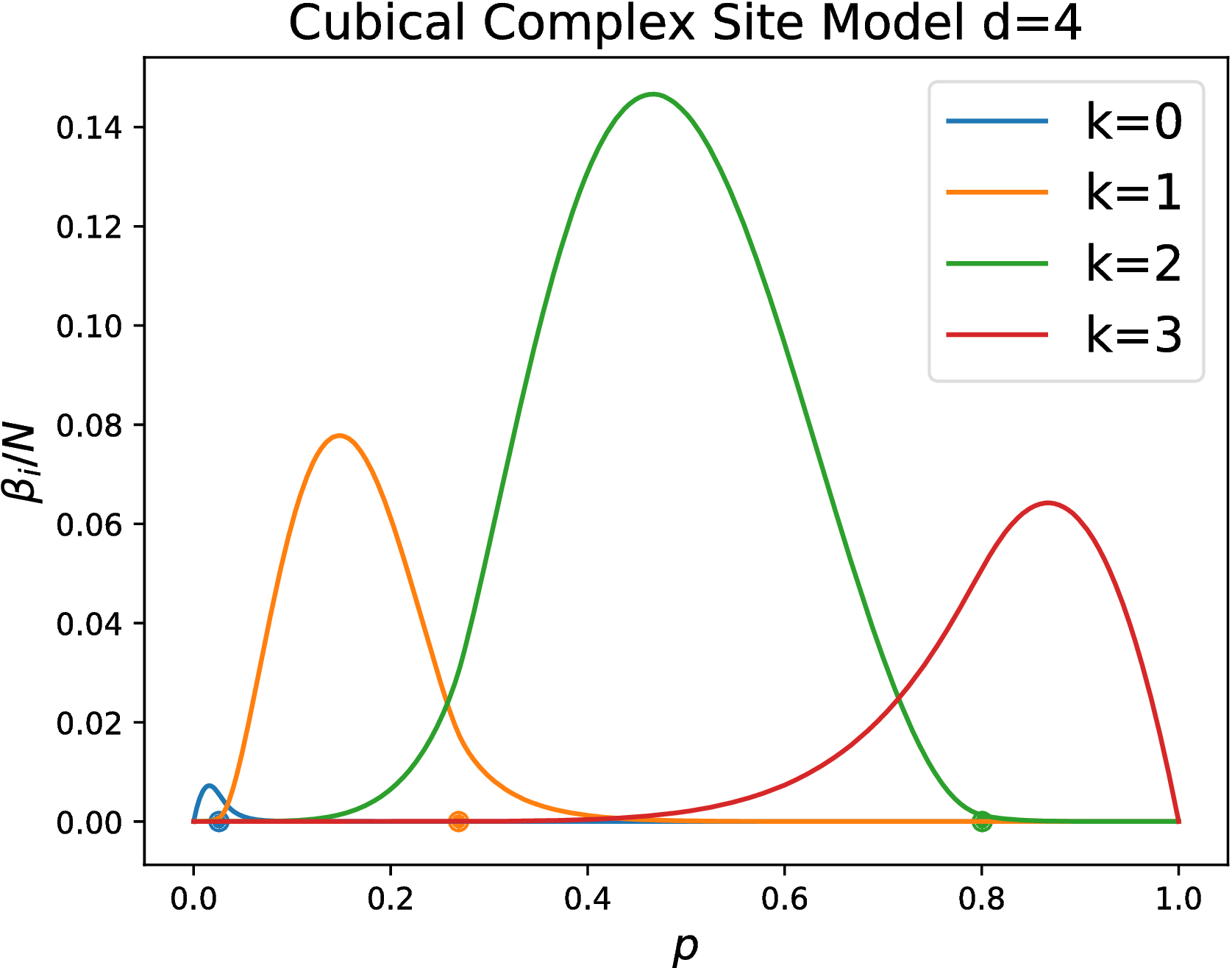}
				 \caption{}
         \label{fig:bettiuniform4}
     \end{subfigure}
     \hfill
     \begin{subfigure}[b]{0.3\textwidth}
         \centering
      				 \includegraphics[width=\textwidth]{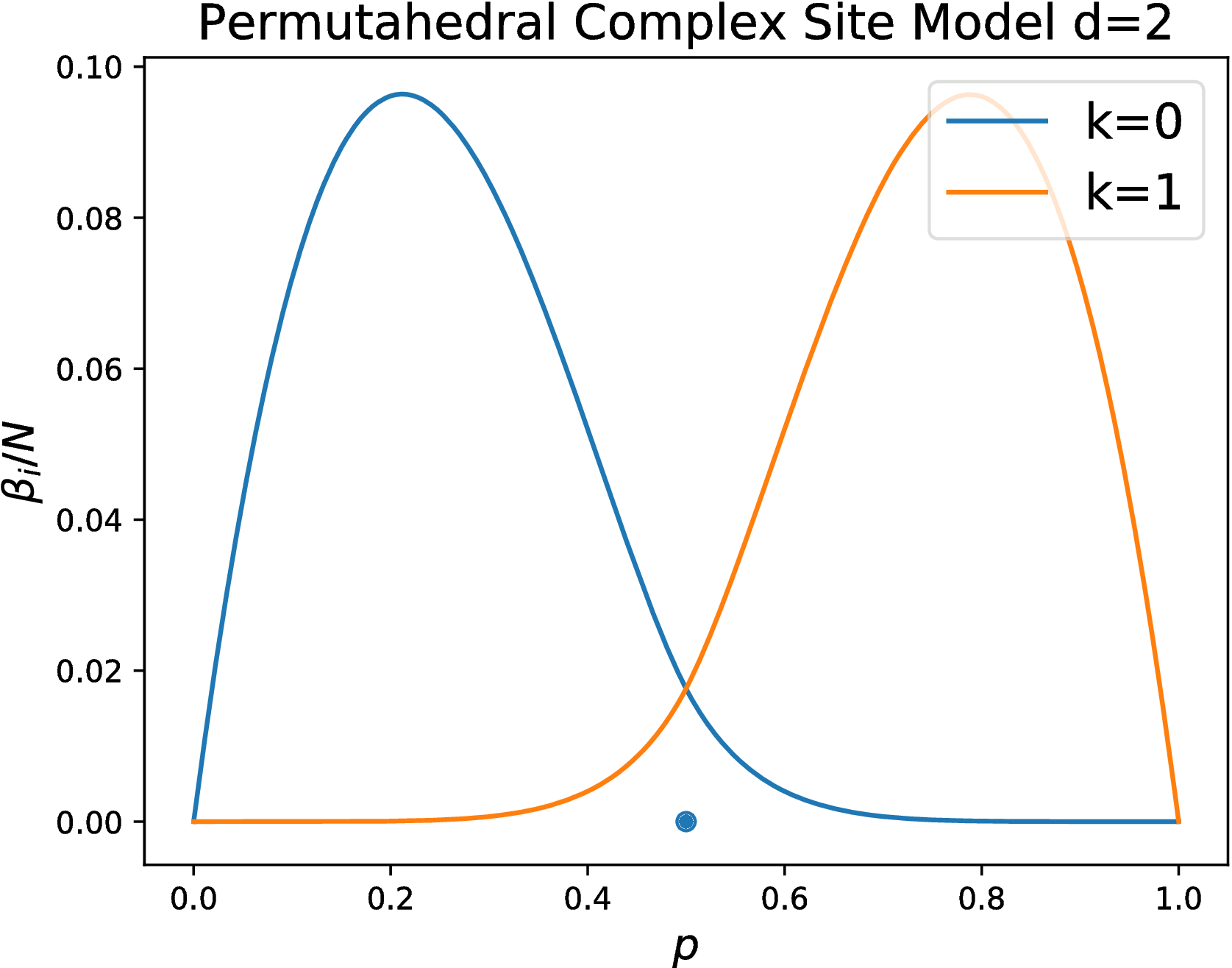}
				 \caption{}
         \label{fig:bettiperm2}
     \end{subfigure}
     \hfill
     \begin{subfigure}[b]{0.3\textwidth}
         \centering
				 \includegraphics[width=\textwidth]{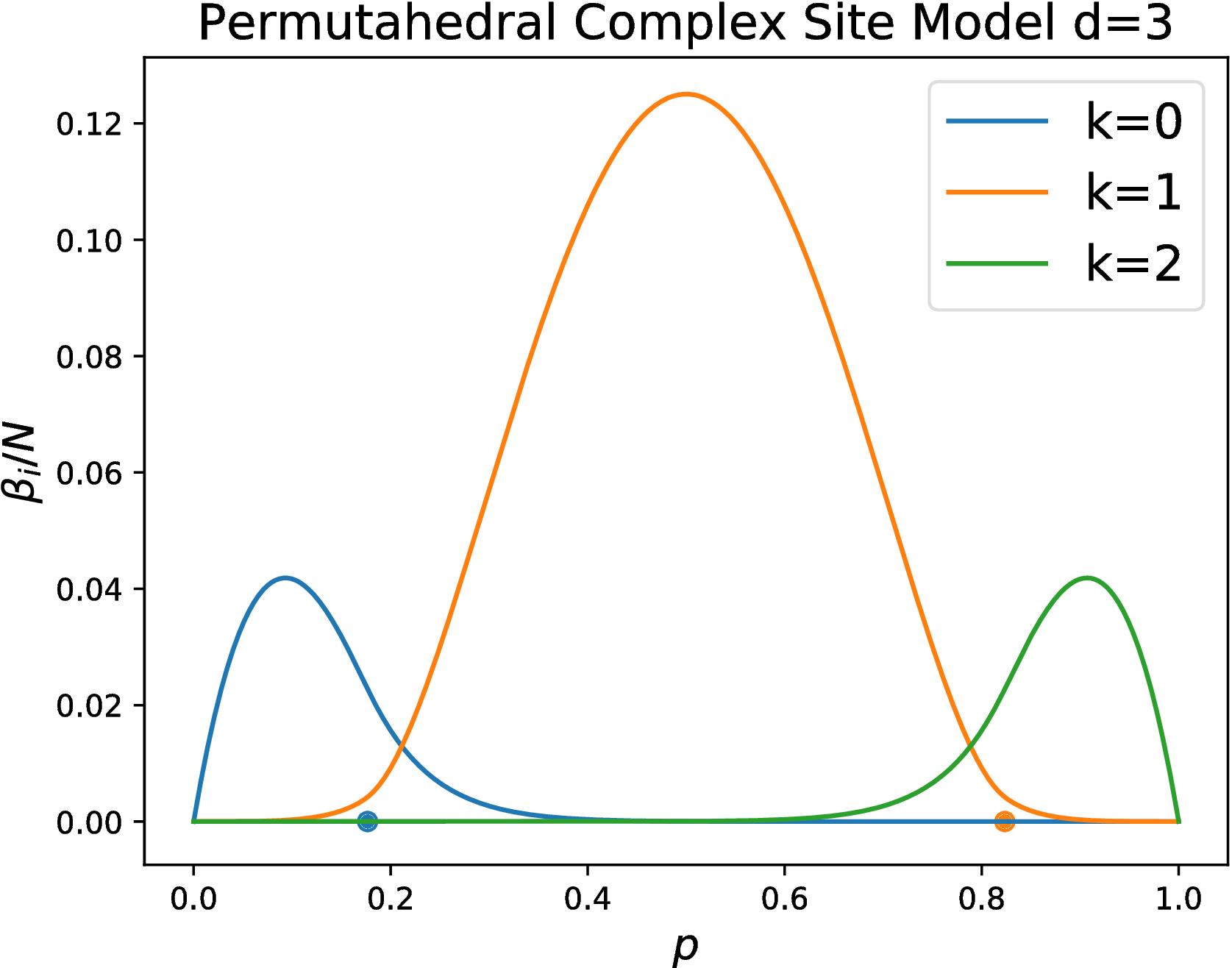}
         \caption{}
         \label{fig:bettiperm3}
     \end{subfigure}
     \hfill
     \begin{subfigure}[b]{0.3\textwidth}
         \centering
				 				 \includegraphics[width=\textwidth]{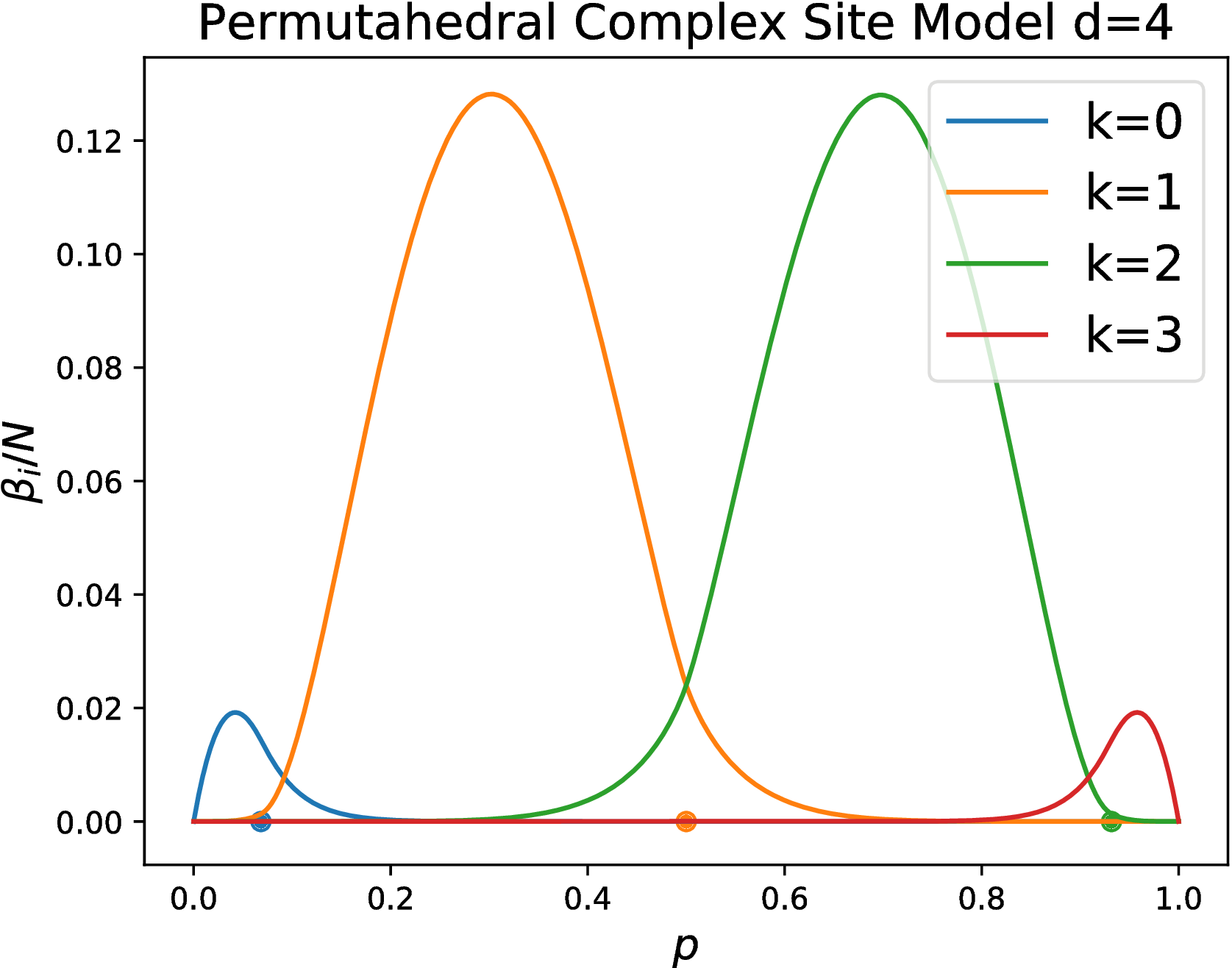}
 			   \caption{}
         \label{fig:bettiperm4}
     \end{subfigure}
     \begin{subfigure}[b]{0.3\textwidth}
         \centering
  				 \includegraphics[width=\textwidth]{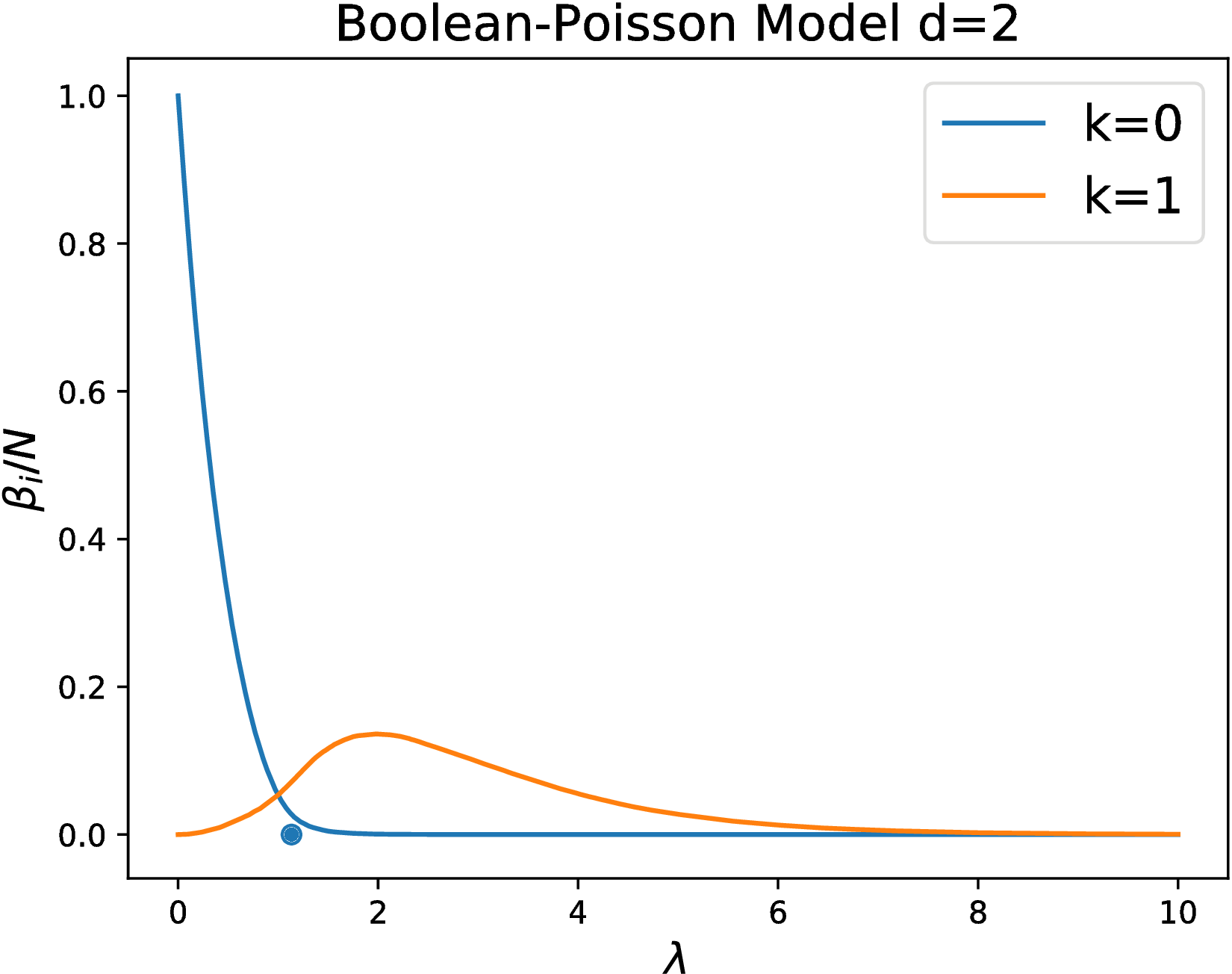}
				 \caption{}
         \label{fig:bettiboolean2}
     \end{subfigure}
     \hfill
     \begin{subfigure}[b]{0.3\textwidth}
         \centering
						 \includegraphics[width=\textwidth]{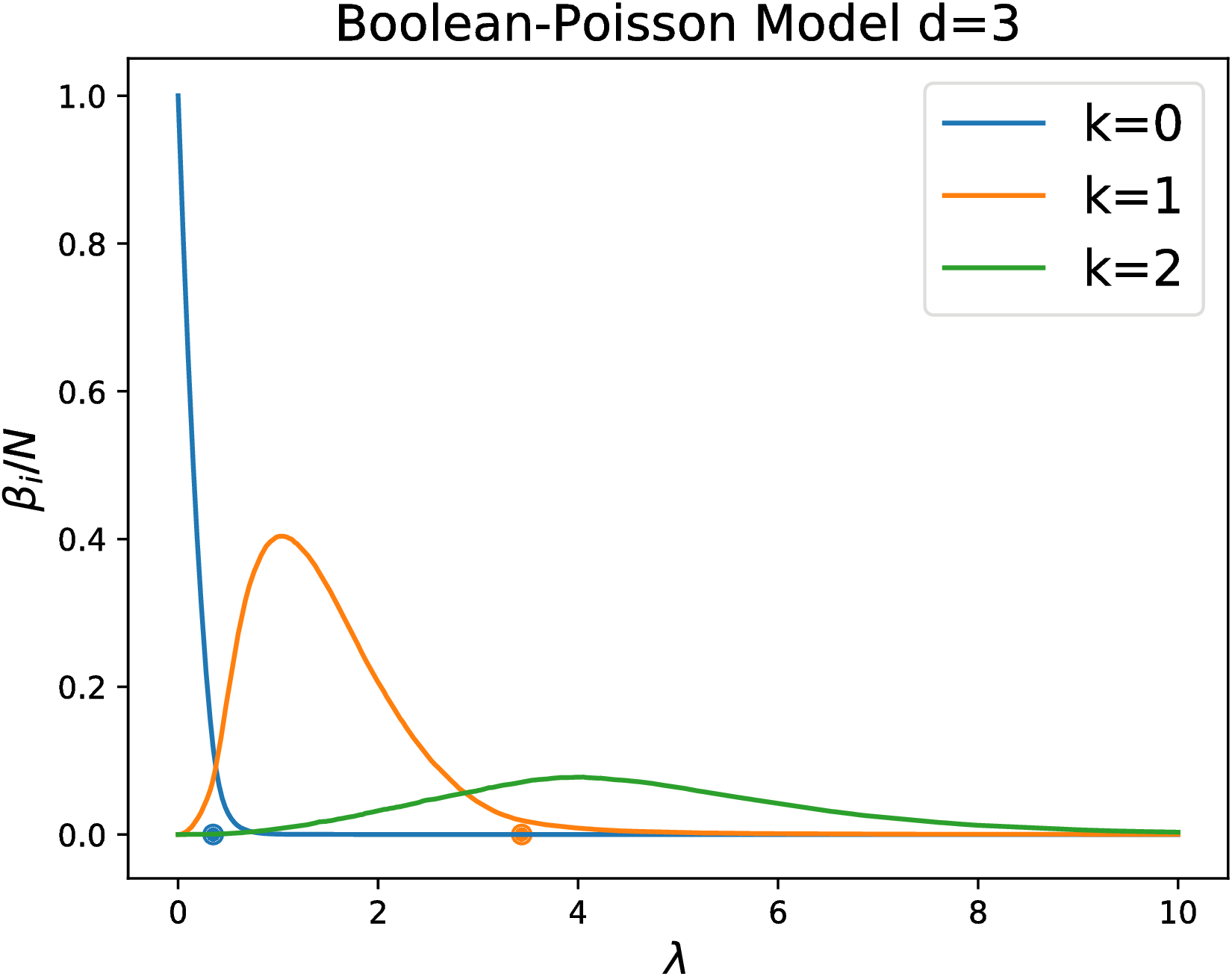}
				 \caption{}
         \label{fig:bettiboolean3}
     \end{subfigure}
     \hfill
     \begin{subfigure}[b]{0.3\textwidth}
         \centering
							 \includegraphics[width=\textwidth]{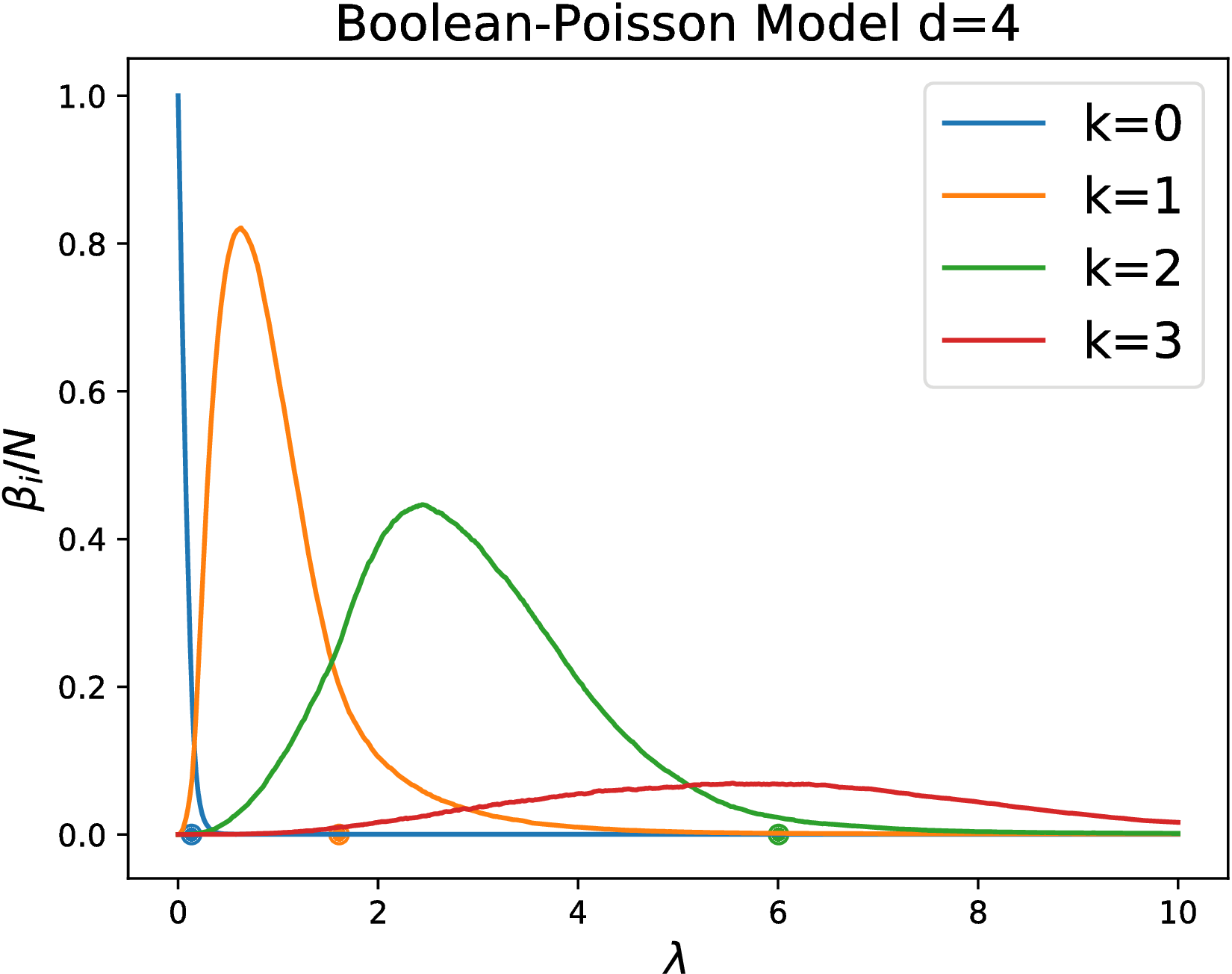}
         \caption{}
         \label{fig:bettiboolean4}
     \end{subfigure}
     \begin{subfigure}[b]{0.3\textwidth}
         \centering
				 \includegraphics[width=\textwidth]{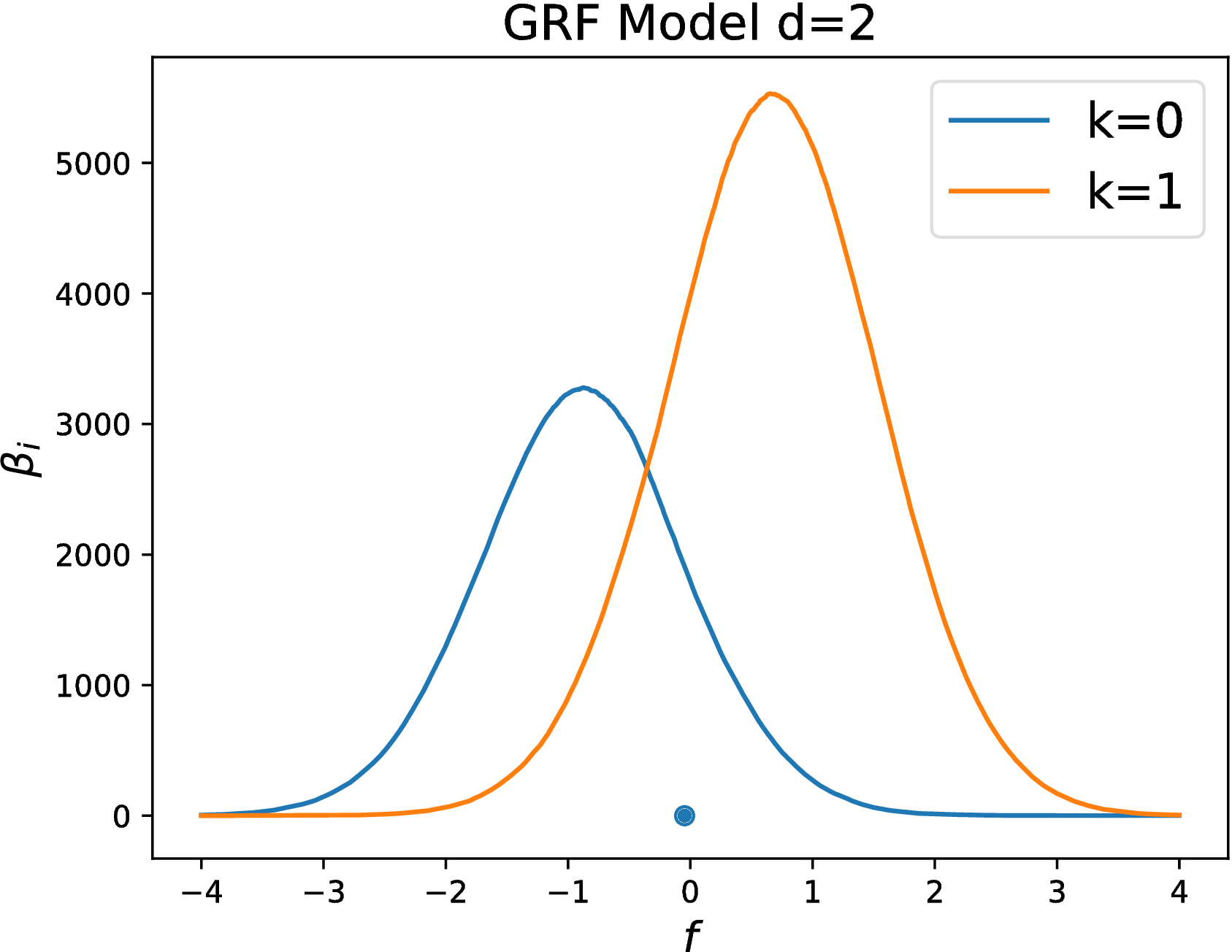}

				 \caption{}
         \label{fig:bettigaussian2}
     \end{subfigure}
     \hfill
     \begin{subfigure}[b]{0.3\textwidth}
         \centering
				 			 	 \includegraphics[width=\textwidth]{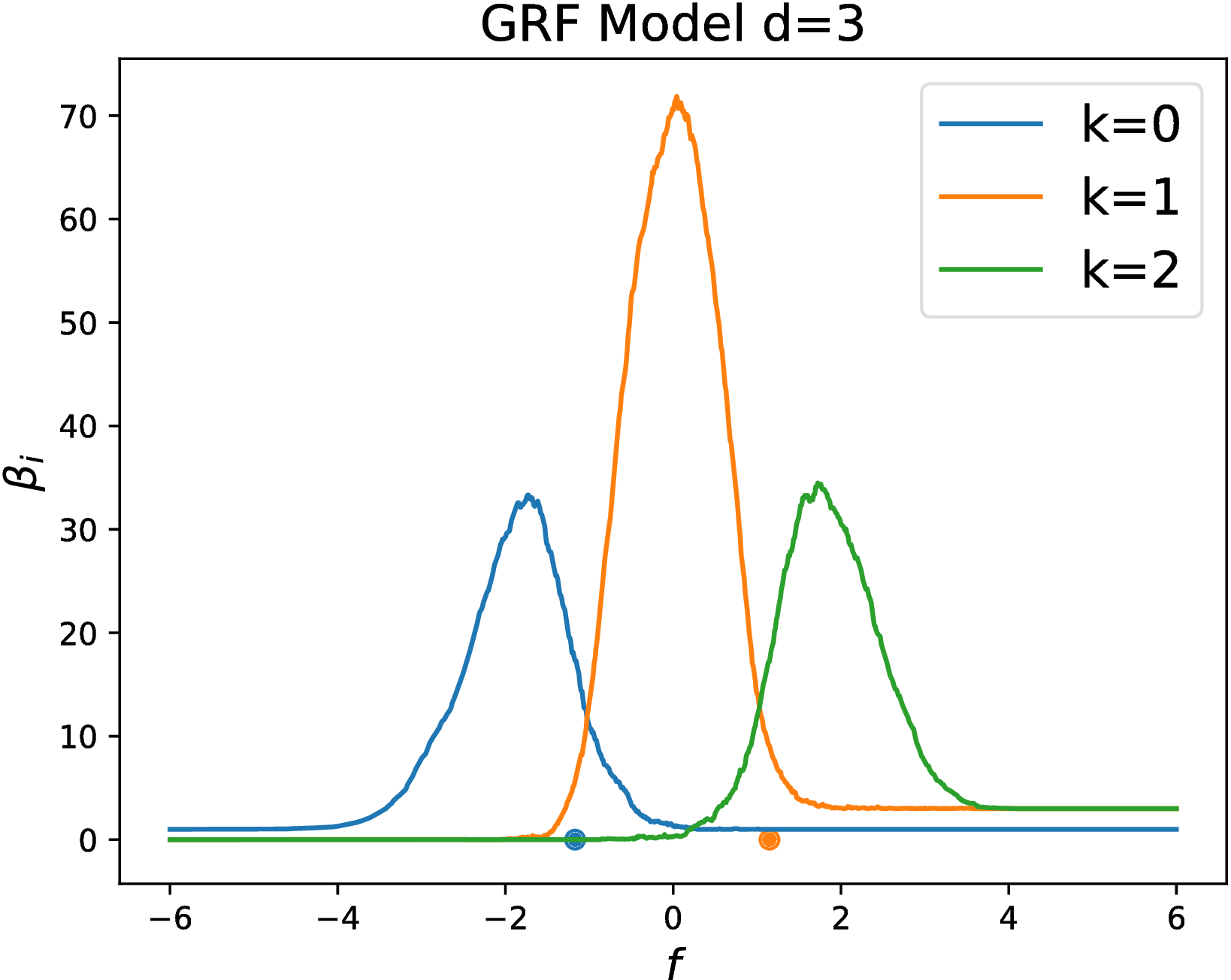}
				 \caption{}
         \label{fig:bettigaussian3}
     \end{subfigure}
     \hfill
     \begin{subfigure}[b]{0.3\textwidth}
         \centering

	 	 \includegraphics[width=\textwidth]{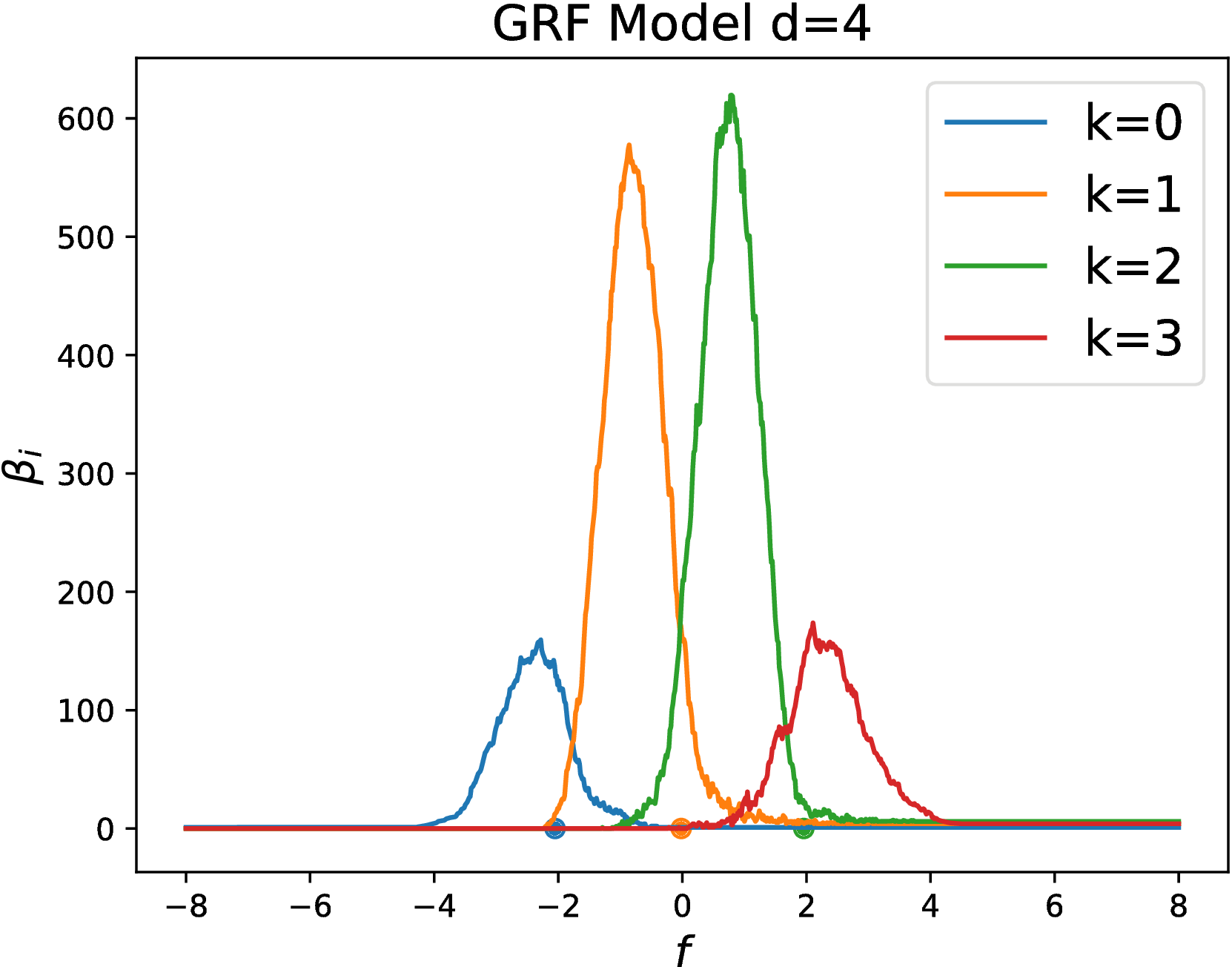}
				 \caption{}
         \label{fig:bettigaussian4}
     \end{subfigure}
     \caption{\label{fig:ec_betti_curves} The empirical Betti curves and the giant cycles.
     In each plot we draw the Betti curves (solid line), along with the birth time of the first giant $k$-cycle for $k=1,\ldots, d-1$. (a)-(c) The random cubical complex. (d)-(f) The random permutahedral complex. (g)-(i) The Boolean model. (j)-(l) The Gaussian random field. We simulated all the models on the $d$-dimensional torus, for $d=2,3,4$ (from left to right). }
\end{figure}
%%%%

%%%%
\begin{figure}
	\centering
     \begin{subfigure}[T]{0.23\textwidth}
         \centering
         \includegraphics[width=\textwidth]{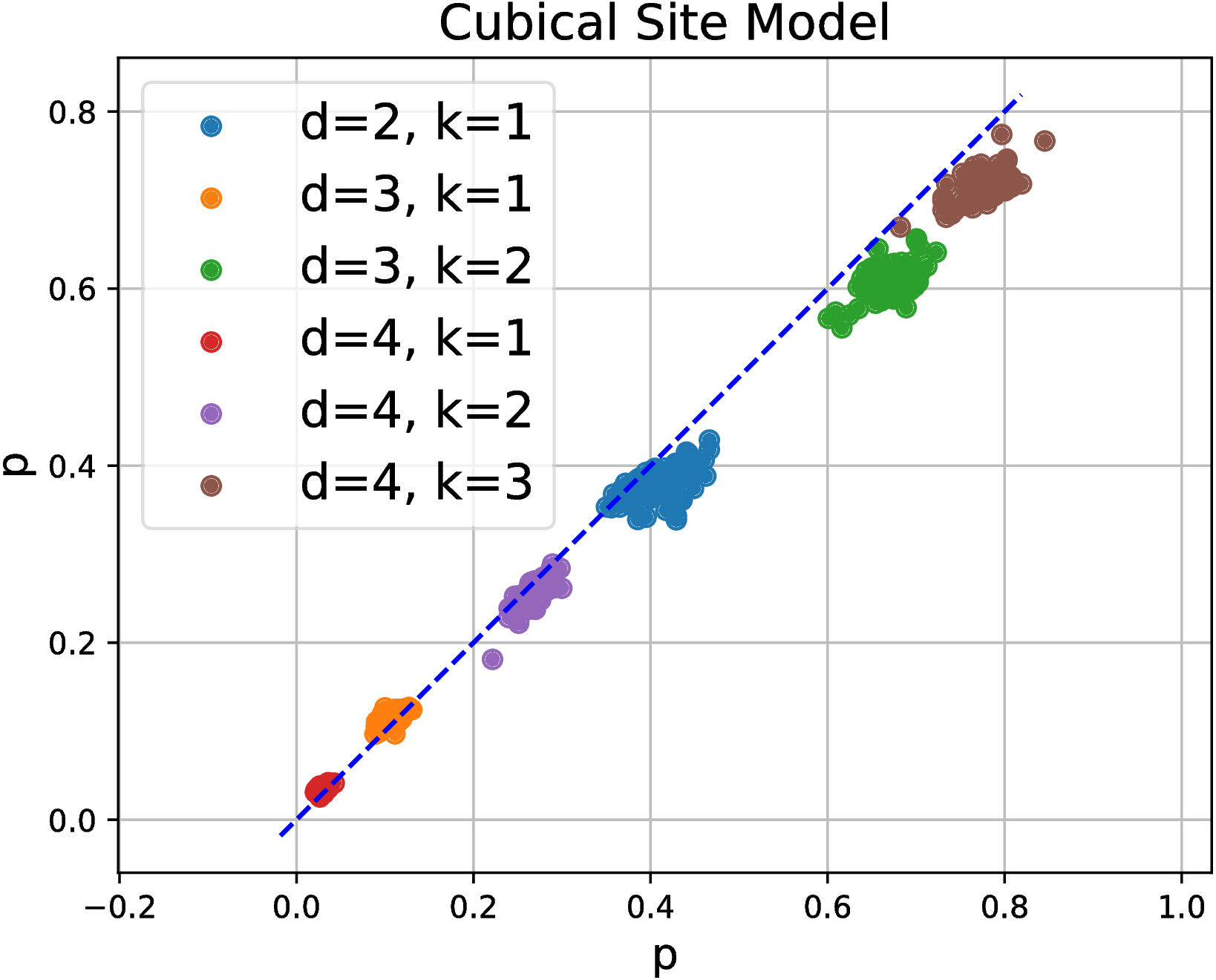}
         \caption{}
         \label{fig:uniform_betti_perc}
     \end{subfigure}\hfill%%%%
     \begin{subfigure}[T]{0.23\textwidth}
         \centering
				 \includegraphics[width=\textwidth]{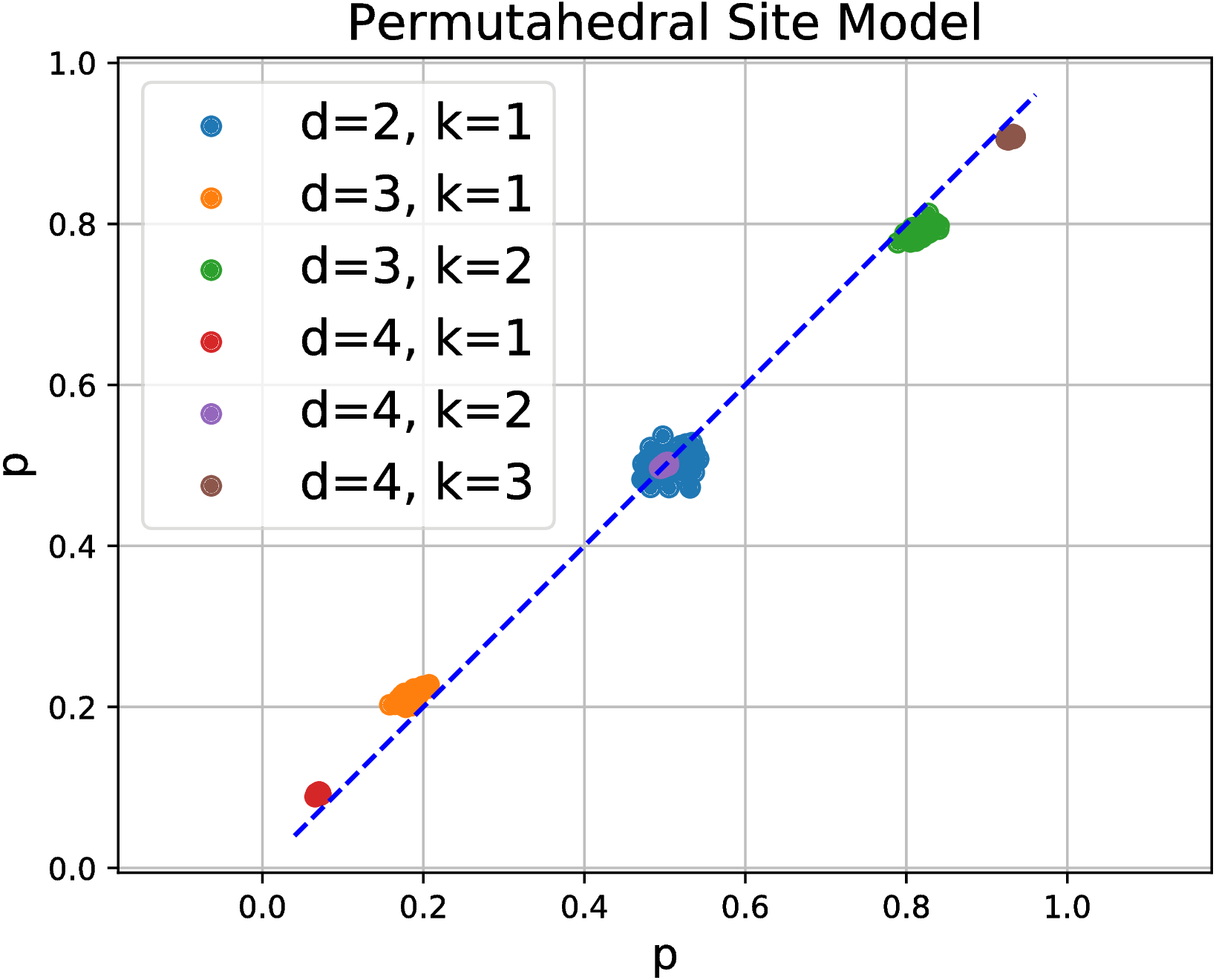}
         \caption{}
         \label{fig:perm_betti_perc}
     \end{subfigure}\hfill%%%%
     \begin{subfigure}[T]{0.23\textwidth}
				 \centering
				 \includegraphics[width=\textwidth]{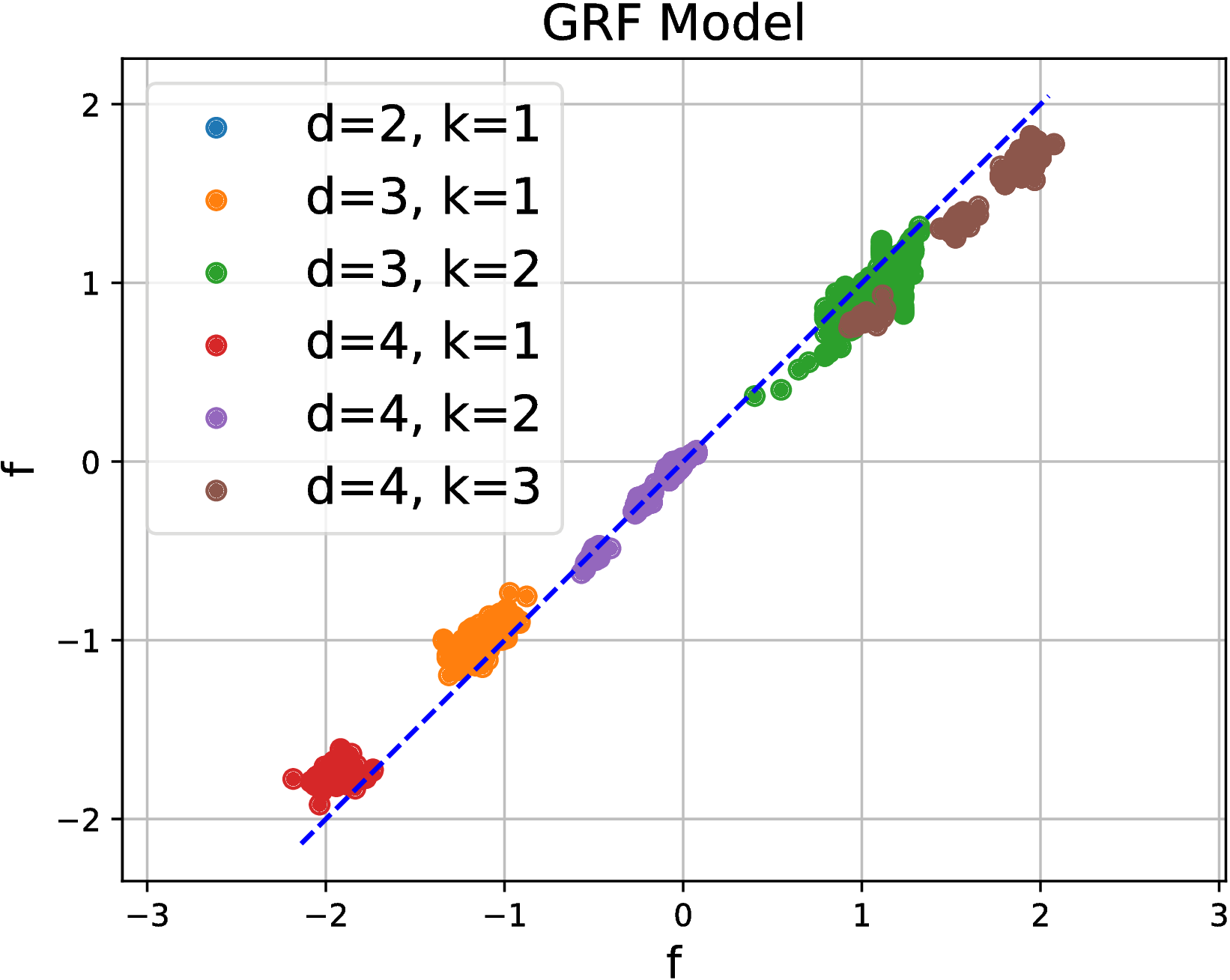}
         \caption{}
         \label{fig:gaussian_betti_perc}
     \end{subfigure}\hfill%%%%
     \begin{subfigure}[T]{0.23\textwidth}
         \centering
				 \includegraphics[width=\textwidth]{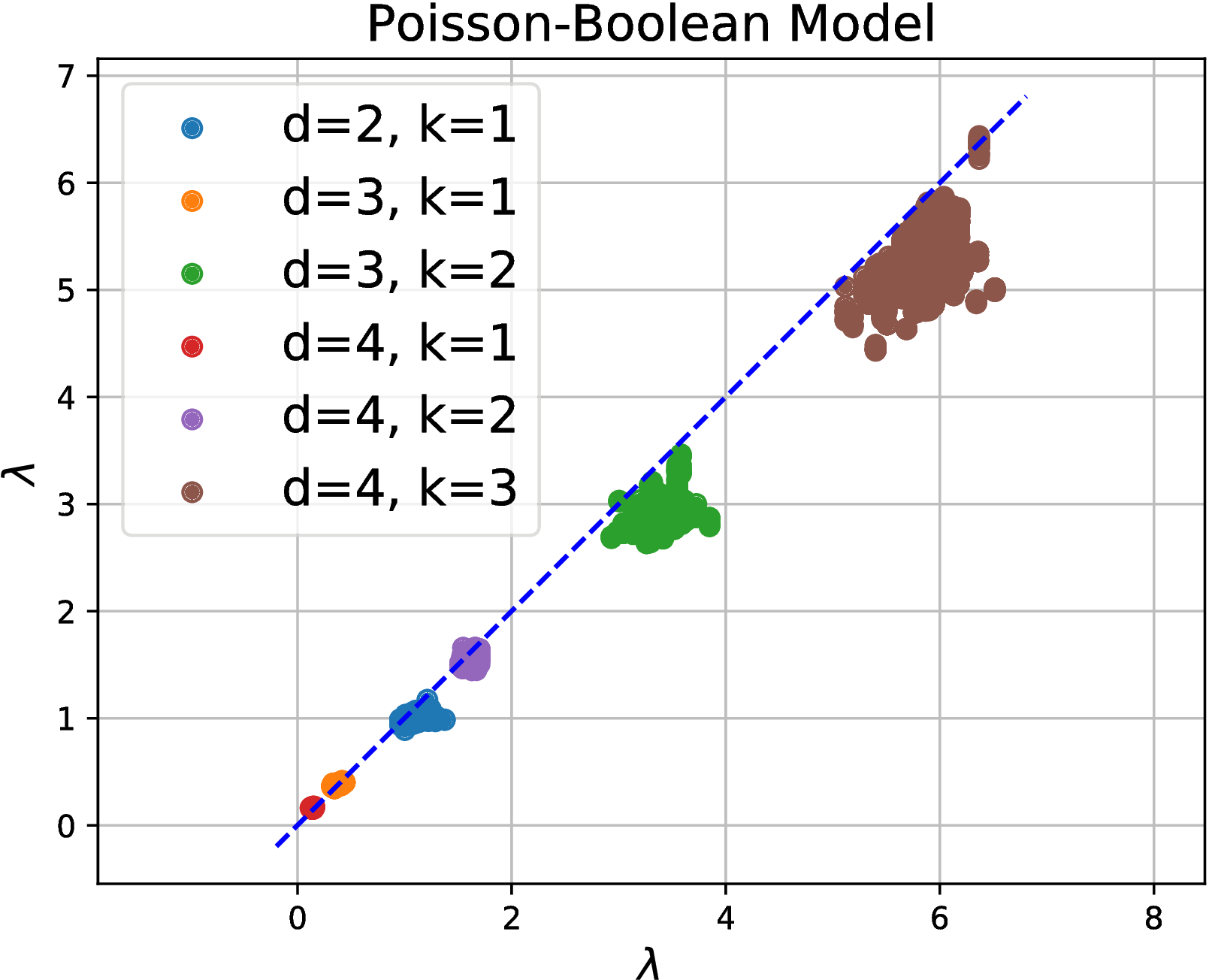}
         \caption{}
         \label{fig:bool_betti_perc}
     \end{subfigure}
        \caption{Appearance of giant cycles vs.~emerging of the Betti numbers.
       The x-coordinate of each point corresponds to the appearance time of the $k$-th giant cycle, while the y-axis is the time at which the equality $\beta_i = \beta_{i-1}$ occurs. The plots provide a strong experimental evidence that the appearance of higher dimensional cycles occurs before or this equality.
 (a) Cubical site model (b) Permutahedral site model (c) Gaussian random field (d) Poisson-Boolean.}
        \label{fig:betti_vs_perc}
\end{figure}

% \begin{figure}
% 	\centering
%      \begin{subfigure}[t]{0.2\textwidth}
%          \centering
%          \includegraphics[width=1.2\textwidth]{figures/uniform_all.png}
%          \caption{}
%          \label{fig:uniform_all}
%      \end{subfigure}
%      \hfill
%      \begin{subfigure}[t]{0.2\textwidth}
%          \centering
% 				 \includegraphics[width=1.2\textwidth]{figures/perm_all.png}
%          \caption{}
%          \label{fig:perm_all}
%      \end{subfigure}
%      \hfill
%      \begin{subfigure}[t]{0.2\textwidth}
%          \centering
% 				 \includegraphics[width=1.2\textwidth]{figures/gaussian_all.png}
%          \caption{}
%          \label{fig:gaussian_all}
%      \end{subfigure}
% 		 \hfill
%      \begin{subfigure}[t]{0.2\textwidth}
%          \centering
% 				 \includegraphics[width=1.2\textwidth]{figures/bool_all.png}
%          \caption{}
%          \label{fig:bool_all}
%      \end{subfigure}
%         \caption{The difference between the empirical expectation of the appearance of the giant cycles and the zeros of the Euler curve vs. the grid size for all dimensions and $k$-giant cycles. (a) Cubical site model (b) Permutahedral site model (c) Gaussian random field (d) Poisson-Boolean.}
%         \label{fig:errors}
% \end{figure}

Another quantity that is interesting to consider are the Betti curves (i.e. the evolution of $\beta_k$ over time).
It has been observed in the past that these curves exhibit a ``separation" phenomenon, where for each range of parameters a single Betti number dominates all the others (see Figure \ref{fig:ec_betti_curves}).
Therefore, if we consider $\beta_0,\ldots, \beta_{d-1}$, we can define
\[
	t_k^{\mathrm{betti}} := \inf \set{t : \beta_{k-1}(t) = \beta_k(t)}.
\]
Our simulations show that $t_k^{\mathrm{betti}}$ is tightly connected to $t_k^{\mathrm{perc}}$ and $t_k^{\mathrm{ec}}$.
Further, Figures \ref{fig:ec_betti_curves}-\ref{fig:betti_vs_perc}  suggest  the following conclusions:
\begin{enumerate}
\item The giant $k$-cycles appear after the peak in $\beta_{k-1}$ and before the peak in $\beta_k$.
\item The zeros of the EC curve are good approximations for the $t_k^{\mathrm{betti}}$.
\end{enumerate}
If the above could be shown, one potential path to understanding the difference $\Delta_k$ could be through investigating what percentage of $(k-1)$-cycles  must be filled before most $k$-simplicies create $k$-cycles (rather than destroying $(k-1)$-cycles), which is when we expect giant cycles to appear. This is related to the recent study of phase transitions in non-geometric models~\cite{linial_phase_2016}, and an object known as the ``giant shadow".

We conclude this section by noting that the  simulation results also support an open conjectures about the Betti curves being unimodal \cite{kahle_random_2011}. If the above holds over a large enough set of parameters, the shape of the Euler curve may be used to show unimodality of the Betti curves.

%%%%%
\section{Conclusion}
In this paper we discussed a new type of percolation phenomena we call ``homological percolation" where giant $k$-cycles appear in the homology of a random structure. Our results suggest a strong connection between the percolation thresholds $t_k^{\mathrm{perc}}$ and the zeros of the expected EC curve m $t_k^{\mathrm{ec}}$. We demonstrated this connection in four types of random percolation models across multiple dimensions.

The results in this paper are purely experimental and should serve as the basis for a deep theoretical study to prove the conjectures we made in this paper. Aside from the mathematical challenge of proving these conjectures, they can have significant implication in various fields. For example, for most models in percolation theory the exact thresholds are not known. Therefore, proving an explicit rigorous link  between the expected EC curve and the percolation thresholds, will allow us to approximate the thresholds in various models, and  maybe even find their exact values.

Another application is in the field of Topological Data Analysis (TDA). A significant effort in TDA is to identify significant topological features in data. If we consider the giant cycles to be significant (as they represent a feature of the true underlying shape), then our conjectures suggest that in order to locate these significant features, we can calculate the EC curve and search for cycles that appear around the corresponding zero of the EC. This heuristic should be further developed, once any of the conjectures is proved.

%write about intuition

\section*{Acknowledgement}
The authors are grateful to Yogeshwaran Dhandapani for suggesting the idea of checking for a connection between homological percolation and the Euler characteristic, and to Stephen Muirhead for discussions on percolation for Gaussian random field models.

\appendix
\section{Calculating the EC for site-percolation models}\label{sec:ec_calc}
Here we present the details of the calculation of the expected Euler characteristic curve.
Recall, that for the site percolation models the filtration parameter is $t=p$. To establish a formula for $\bar\chi(p)$, we use Equation~\eqref{eqn:ec_faces} and the linearity of expectation. Thus, we need to evaluate the expected number of $k$-faces in each of the site models.

Beginning with the cubical model, we observe that each $d$-dimensional cube has $2^{d-k}\binom{d}{k}$ $k$-faces on its boundary. In addition, since we consider $Q_n^d$ (the discretization of $\T^d$ into $n$ boxes),  each $k$-face is on the boundary of precisely $2^{d-k}$ $d$-dimensional boxes. Therefore, the total number of $k$-faces in $Q_n^d$ is exactly $n\binom{d}{k}$.

Now, for any $k$-face,  if it is included in $Q(n,p)$, then at least one of $d$-dimensional boxes that contains it must be open. Therefore, the probability of a $k$-faces to be in $Q(n,p)$ is $(1-(1-p)^{2^{d-k}})$.
Putting everything together, we have that
\[
	\mean{F_k(Q(n,p))} = n\binom{d}{k} (1-(1-p)^{2^{d-k}}),
\]
which then yields \eqref{eqn:ec_cubical},
\[
\bar\chi_Q(p) = \sum_{k=0}^d (-1)^k\mean{F_k(Q(n,p))} =  n \sum_{k=0}^d (-1)^k\binom{d}{k} (1-(1-p)^{2^{d-k}}).
\]

The calculation  for the permutahedral complex $P(n,p)$ is similar, where the only difference is the face counting. From~\cite{ziegler2012lectures}[p.18], each $k$-face in $P_n^d$ corresponds to a partition of the set $\{0,\ldots,d\}$ into $d+1-k$ nonempty parts. Therefore, the number of $k$-faces for each cell is given by a Stirling number of the second kind~\cite{graham1989concrete}[Section 6.1, p.258],
\begin{align*}
 F(P^d_1) = (d+1-k)!S(d+1,d+1-k) &= \sum\limits_{i=0}^{d+1-k}(-1)^i \binom{d+1-k}{i}(d+1-k-i)^{d+1} \\
	&= \sum\limits_{j=0}^{d+1-k}(-1)^{d+1-k-j} \binom{d+1-k}{j}j^{d+1}
\end{align*}
where the second equality follows from using the substitution $j=d+1-k-i$.

Now every $k$-face belongs to $(d+1-k)$ $d$-cells. This follows from the genericity of the corresponding Voronoi cells (see the proof of Lemma~\ref{lem:perm_equivalence} in the Appendix~\ref{sec:duality}). Hence, the total number of $k$-cells is
\[
F_k(P^d_n) = \frac{n}{d+1-k} \sum\limits_{j=0}^{d+1-k}(-1)^{d+1-k-j} \binom{d+1-k}{j}j^{d+1} = n\sum\limits_{j=0}^{d+1-k}(-1)^{d+1-k-j} \binom{d-k}{j}j^{d}
\]
Therefore,
\[ \mean{\chi_P(n,p)} = n\sum\limits_{k=0}^d (-1)^k\left(1- (1-p)^{d+1-k}\right) \sum\limits_{j=0}^{d+1-k}(-1)^{d+1-k-j} \binom{d-k}{j}j^{d+1}.
\]
Exchanging between $k$ and $(d-k)$ then yields \eqref{eqn:ec_perm}.

%%%%%
\section{The expected EC curve for the Gaussian random field}

As stated in Section \ref{sec:grf}, the expected EC is calculated via the Gaussian Kinematic Formula, developed in \cite{taylor_gaussian_2009}. Suppose that $M$ is a $d$-dimensional manifold, and let $f:M\to \R$ be a Gaussian random field with zero mean and unit variance (with some further smoothness conditions detailed in \cite{taylor_gaussian_2009}).
Let $D_u = [u,\infty)$, then $f^{-1}(D_u)$ is a super-level set of $f$. The GKF (Theorem 4.1 in \cite{taylor_gaussian_2009}) then states that
\[
\mean{\chi( f^{-1}(D_u))} = \sum_{j=0}^d (2\pi)^{-j/2} \cL_j(M) \cM_j(D_u),
\]
where $\cL_j(M)$ are geometric functionals of $M$ known as the \emph{Lipschitz-Killing curvatures}, and $M_j$ is slightly different object known is the \emph{Gaussian-Minkowski functional}. For the special case where $M=\T^d$ it can be shown that $\cL_j(\T^d) = 0$ for all $j < d$, and $\cL_d(\T^d) = 2/\omega_d$. In addition, in \cite{taylor_gaussian_2009} it is shown that $M_j(D_u) = (2\pi)^{-1/2}\cH_{j-1}(u)e^{-u^2/2}$, where $\cH_n(u)$ are the Hermite polynomial in \eqref{eqn:hermite}.
Thus, we have
\[
\mean{\chi( f^{-1}(D_u))} =
\frac{2}{\omega_d} (2\pi)^{-\frac{d+1}{2}} \cH_{d-1}(u) e^{-u^2/2}.
\]
Finally, recall that we defined $G(\alpha)$ as the sub-level sets of $f$. In addition, since $f$ is a zero-mean Gaussian field, we have that $f(x)$ and $\tilde f(x) := -f(x)$ have the same distribution. Since the sub-level sets of $f$ are the super-level sets of $\tilde f$, we have
\[
\bar\chi_G(\alpha) := \mean{\chi(G(\alpha)} = \meanx{\chi(\tilde f^{-1}(D_{-\alpha}))} = \meanx{\chi( f^{-1}(D_{-\alpha}))},
\]
and therefore,
\[
\bar\chi_G(\alpha) := \mean{\chi(G(\alpha)} = \frac{2}{\omega_d} (2\pi)^{-\frac{d+1}{2}} \cH_{d-1}(-\alpha) e^{-\alpha^2/2}.
\]
\section{Symmetry and Duality for the Permutahedral Complex}\label{sec:duality}
In this section we provide formal proofs for the statements on symmetry that are discussed in Section~\ref{sec:perm}. For the case of the site percolation on a hexagonal grid, the symmetry around
$p=1/2$ is well known. Here we extend it to arbitrary dimension, but note that the proofs assume some familiarity with algebraic topology.

The idea behind the proofs is to relate a subspace of a manifold (in this case, $d$-torus), with its complement. Informally, the topology of the manifold and a subspace determine the topology of the complement. The most well known example of this is Alexander duality, which relates the $k$-th homology of a subspace with the $d-k$ cohomology of the complement. In our setting, we consider the Betti numbers so there is no distinction between homology and cohomology.
\begin{figure}[h!]
\centering\includegraphics[width=0.5\textwidth]{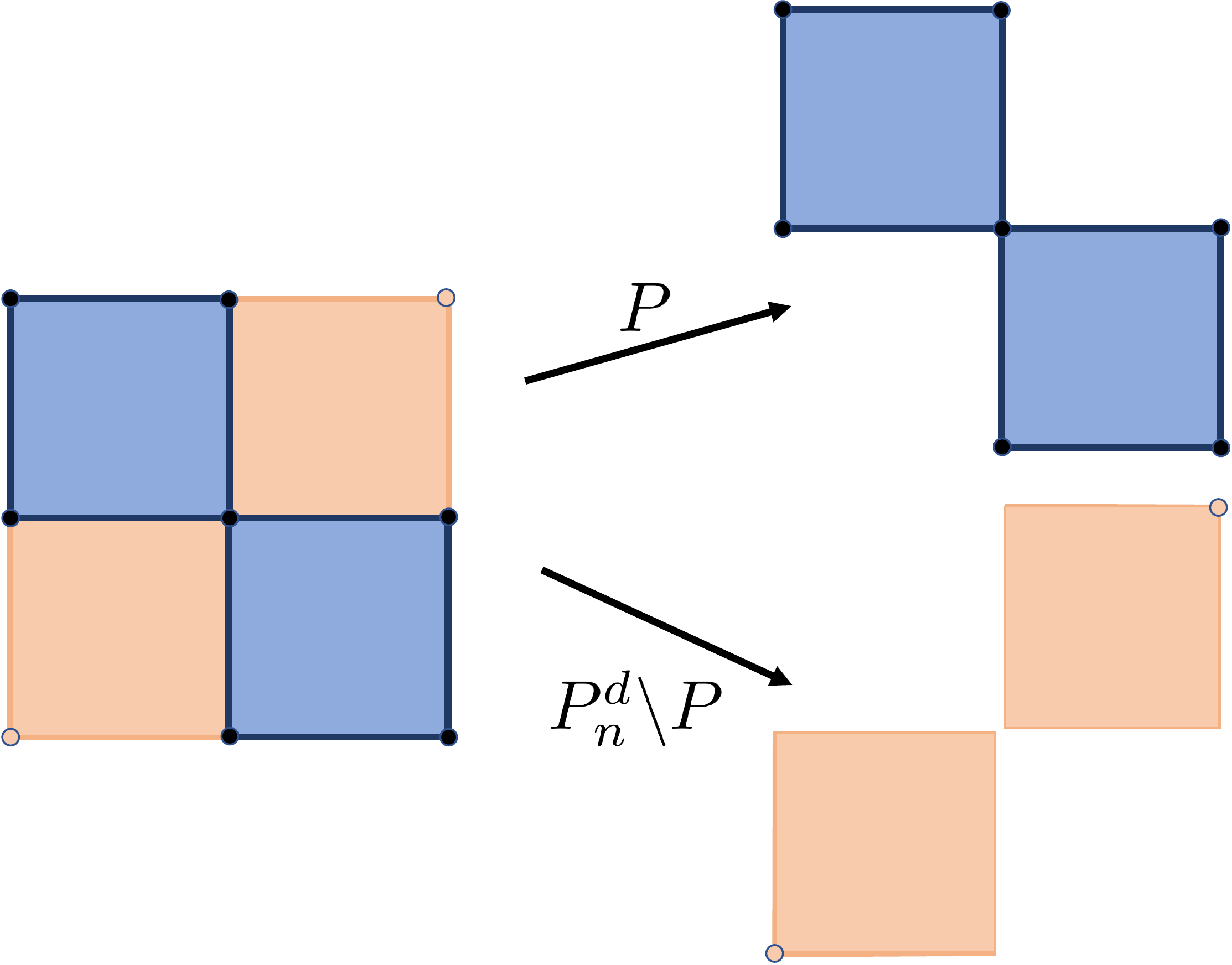}
\caption{\label{fig:complement} An example of a failure of symmetry for the cubical complex. The dark squares indicate open sites. The complement of the open sites is different (bottom right) than the union closed sites. In particular, the latter is connected (via the point in the middle), while the former is not.}
\end{figure}
Before getting to the duality, there is a technical obstacle to overcome. In the site models we consider, taking the complement of the open sites is not the same as considering the union of the closed sites, but rather it is equivalent to the \emph{closure} of the complement. This difference can change the topology as in the case of the cubical complex, as can be seen in Figure~\ref{fig:complement}. Hence, we first prove that the complement and the closure of the complement are equivalent.
As in Section~\ref{sec:perm}, let $P\subseteq P^d_n $ and $P^c = \mathrm{cl}(P^d_n \backslash P)$.
\begin{lem}\label{lem:perm_equivalence}
		For all $0\le k \le d$,
		 $$\Hg_k(P^c) \cong \Hg_k(P^d_n \backslash P). $$
\end{lem}
\begin{proof}
	To prove this lemma, we prove a stronger statement, namely that the $P^c$ and $P^d_n \backslash P$ are homotopy equivalent. First, consider the open cover induced by the sites in $P^c$, denoted by $\mathcal{U}$. That is, each element in the cover is an open neighborhood of each site. Since the sites are convex, it follows $\mathcal{U}$ is a good cover and hence $P^c$ is homotopy equivalent to the nerve of the cover, $\mathcal{N}\mathcal{U}$.

	As noted in Section~\ref{sec:perm}, each site is a permutahedron of order $d+1$.  The interior of each site corresponds to the top dimensional cell of the permutahedra which are the same for $P^d_n \backslash P$ and $P^c$. The two differ in that $P^d_n \backslash P$ is does not have lower dimensional faces (of the sites) which are adjacent to sites both in $P^c$ and in $P$.
	Taking the same open cover as above, but on $P^d_n \backslash P$, denoted by $\mathcal{U}'$.
	We show that this is a good cover \emph{and} that the nerves are the same, which implies the result.

	 First, we note that there is a one-to-one correspondence between the $(d-k)$-faces of permutahedron and $k$-simplices of the nerve ~\cite{choudhary2019polynomial}[Proposition 7]. That is, each intersection of $(k+1)$ cells corresponds to a $(d-k)$-face of the permutahedron. For any $(d-k)$-face $\tau$ in  $P^d_n \backslash P$,  it must be adjacent to $(k+1)$ sites in $P^d_n \backslash P$ and so cannot be adjacent to any sites in $P$. Note that  lower dimensional faces of $\tau$ (which are in the closure of $\tau$) may be missing from $P^d_n \backslash P$ and so it is not convex. It however remains star-shaped and hence the $(k+1)$ intersection of cover elements is contractible, implying that $\mathcal{U}'$ is a good cover.

	 The same argument also shows that any face which is in $P^c$ but not $P^d_n \backslash P$,
	  does not affect the nerve as the corresponding interesction remains non-empty.  Note that in the 2$D$ cubical complex a pairwise intersection may correspond to a vertex rather than an edge which breaks the argument above.

		Hence the $\mathcal{N}\mathcal{U} =\mathcal{N}\mathcal{U}'$ completing the proof.
\end{proof}

The above lemma allows us to use $P^c$ and $P^d_n \backslash P$ interchangeably. We can now prove Lemma~\ref{lem:dual_betti}.
\begin{lem}[\ref{lem:dual_betti}]For $0\le k \le d$,
\[
\cB_k(P) + \cB_{d-k}(P^\c) = \beta_k(\T^d).
\]
\end{lem}
\begin{proof}
	In this proof, we use $P^c$ in place of the complement of $P$ as they are equivalent by Lemma~\ref{lem:perm_equivalence}.
	There exists a commutative diagram where the rows are exact, due to the long exact sequence for relative (co)homology.
\begin{equation}\label{eq:diagram}
	\begin{tikzcd}
		\Hg_{k}(P) \arrow[r,"i_*"] & \Hg_{k}(P^d_n) \arrow[r] & \Hg_{k}(P^d_n, P) \arrow[r,"\delta_k"]  &		\Hg_{k-1}(P)\\
		\Hg^{d-k}(P^d_n,P^c) \arrow[u,"\cong"] \arrow[r] & \Hg^{d-k}(P^d_n)\arrow[u,"\cong"] \arrow[r,"j^*"] &\Hg^{d-k} (P^c) \arrow[r,"\delta^{d-k}"]\arrow[u,"\cong"] &\Hg^{d-k+1}(P^d_n,P^c)\arrow[u,"\cong"]
	\end{tikzcd}
\end{equation}
	The leftmost  isomorphism follows from the Lefschetz duality, the second from Poincare duality, and the third from the Five Lemma. Note that a detailed proof can be found in  \cite{hatcher_algebraic_2002}[Theorem 3.44]. We can decompose the full space as
\begin{equation}\label{eq:exact}
	\Hg_{k}(P^d_n) \cong \im i_* \oplus\coker i_*,
\end{equation}
	 and by exactness and a diagram chase, we have  that $\coker i_* \cong \im j^*$. We observe that
\begin{align*}
	\cB_k(P)  &= \dim(i_*),\\
	\cB_{d-k}(P^c) &= \dim(j^*),
\end{align*}
where the second equality follows from the equivalence for ranks of homology and cohomology over fields. Substituting into Equation \eqref{eq:exact}, we obtain the result,
$$ \beta_k(\T^d) = \beta_{k}(P^d_n) =  \cB_k(P) + \cB_{d-k}(P^c)$$
\end{proof}

We conclude by proving the symmetry of the Euler curve.
\begin{lem}
For the permutahedral complex we have the following symmetry,
\[\chi_P(p) = (-1)^d\chi_P(1-p).\]
\end{lem}
\begin{proof}
	Since $P^c(n,1-p)\sim P(n,p)$, it suffices to show that $\chi(P) = (-1)^d \chi(P^c)$ for some $P$.
Consider the diagram~\eqref{eq:diagram}. By exactness,
\begin{align*}
	\beta_k(P) = \dim(\im i_*(k)) + \dim(\im \delta_{k+1}),\\
	\beta_{d-k}(P^c) = \dim(\im j^*(d-k)) + \dim(\im \delta^{d-k}).
\end{align*}
Note that we have added the dimension to the notation of the corresponding morphisms, i.e. $ i_*(k):  \Hg_k(P) \rightarrow \Hg_k(P^d_n)$.
Furthermore, the diagram implies $\dim(\im \delta^{d-k}) = \dim(\im \delta_{k})$.
Computing the Euler characteristic, yields
\begin{align*}
	\chi(P) &= \sum\limits_{k=0}^d (-1)^k \beta_k(P) = \sum\limits_{k=0}^d (-1)^k \left(\dim(\im i_*(k) + \dim(\im \delta_{k+1})\right)\\
	&= \sum\limits_{k=0}^d (-1)^k \left(\binom{d}{d-k} -  \dim(\im j^*(d-k)) + \dim(\im \delta_{d-k-1})\right)\\
	&= \sum\limits_{k=0}^d (-1)^k \left(  \dim(\im \delta^{d-k-1})-\dim(\im j^*(d-k))\right)\\
	&= \sum\limits_{k=0}^d (-1)^k \left(  \dim(\im \delta^{d-k-1}) - \beta_{d-k}(P^c) + \dim(\im \delta^{d-k})\right)\\
		&= (-1)^d\sum\limits_{k=0}^d \beta_{k}(P^c) +  \sum\limits_{k=0}^d (-1)^k \left(  \dim(\im \delta^{d-k-1}) + \dim(\im \delta^{d-k})\right)\\
			&= (-1)^d\chi(P^c)  +  \dim(\im \delta^{d}) - \dim(\im \delta^{-1}) = (-1)^d\chi(P^c)
	 \end{align*}
where the last inequality, we use that $\dim(\im \delta^{d}) = \dim(\im \delta^{-1}) = 0 $.
\end{proof}

% &= \sum\limits_{k=0}^d (-1)^k \binom{d}{d-k} + \sum\limits_{k=0}^d (-1)^k \left( \dim(\im \delta^{d-k+1}) -  \dim(\im j^*(d-k)) \right)\\
% &=\sum\limits_{k=0}^d (-1)^k \left( \dim(\im \delta^{d-k+1}) -  \dim(\im j^*(d-k)) \right)\\
% &=\sum\limits_{k=0}^d (-1)^k \left( \beta_{d-k+1} - \dim(\im j^*(d-k+1))  -  \dim(\im j^*(d-k)) \right)\\
% &=\sum\limits_{\ell=0}^d (-1)^{d-\ell} \left( \beta_{\ell+1} - \dim(\im j^*(\ell+1))  -  \dim(\im j^*(\ell)) \right)\\
% &=(-1)^d\sum\limits_{\ell=0}^d (-1)^{\ell} \left( \beta_{\ell+1} - \dim(\im j^*(\ell+1))  -  \dim(\im j^*(\ell)) \right)\\
% &=(-1)^d\left(\sum\limits_{\ell=0}^d (-1)^{\ell} \beta_{\ell+1} - \sum\limits_{\ell=0}^d (-1)^{\ell} \left(\dim(\im j^*(\ell+1))  + \dim(\im j^*(\ell)) \right)\right)\\
% &=(-1)^d\chi(P^c) - \beta_0(P^c) - j^*(0)

\section{Simulation details}\label{sec:sim_details}
Here we present some implementational details of the simulations. The persistence diagrams and hence Betti curves and appearance of the giant cycles were computed with GUDHI\cite{maria2014gudhi}. For the cubical site model, a $n^d$-grid with periodic connectivity was used  for varying values of $n$. After generating a uniform random function taking values in $[0,1]$, with the values assigned to the top dimensional cells. Persistence was then computed directly on the resulting cubical complex.

We note that the Gaussian random field was also approximated on a cubical grid according to the method described in \cite{wood1994simulation}. The resulting GRF was always generated on the unit torus, with $\sigma^2 = 10^{-3}$.

The permutahedral complex was built by constructing a set of points in the $A^*_d$ grid  embedded in $\R^{d+1}$. The 1-skeleton was then built by choosing an appropriate radius, so that all the neighbors were connected. Note that this can be thought of as the 1-skeleton of the Delaunay complex of the pointset. The points, along wiht their adjacent edges, were then identified with the outgoing edges appropriately identified. This embeds the 1-skeleton in $\T^d$. The full complex was then computed using  clique completion to higher dimensions. Again a random function was assigned to the sites, which correspond to the vertices of the resulting complex.
Persistent homology of the \emph{lower-star filtration} on this complex was then computed.

Finally for the Boolean model, in dimensions 2 and 3, the $\alpha$-filtration~\cite{edelsbrunner2010alpha} on the torus was used , whereas for dimension 4, for each sample point set, the connectivity threshold was first computed and was then used to compute the threshold for the construction of the \v Cech filtration.

%===============================================================================

\bibliographystyle{plain}
\bibliography{zotero}

%===============================================================================
\end{document}

%% file: defs.tex
\usepackage{amsmath,amsthm,bbm,amssymb}
\usepackage{xcolor}
\DeclareMathOperator{\birth}{birth}
\DeclareMathOperator{\death}{death}

\DeclareMathOperator{\coker}{coker}

\DeclareMathOperator{\im}{Im}

\DeclareMathOperator{\PH}{PH}

\def\N{\mathbb{N}}

\def\R{\mathbb{R}}

\def\Z{\mathbb{Z}}

\def\T{\mathbb{T}}

\def\cH{\mathcal{H}}

\def\cL{\mathcal{L}}
\def\cM{\mathcal{M}}

\def\cP{\mathcal{P}}

\def\cX{\mathcal{X}}

\def\cL{\mathcal{L}}

\newcommand{\Hg}{\mathrm{H}}

\newcommand{\E}{\mathbb{E}} %expectation
\newcommand{\given}{\;|\;}
\newcommand{\mean}[1] {\E\left\{{#1}\right\}}
\newcommand{\meanx}[1] {\E\{{#1}\}}

\newcommand{\bX}{{\mathbf{X}}}

\newcommand{\cov}[2]{\mathrm{Cov}\param{{#1},{#2}}}

\newcommand{\set}[1]{\left\{#1\right\}}

\newcommand{\param}[1]{\left(#1\right)}

\newcommand{\floor}[1] {\left\lfloor{#1}\right\rfloor}

\newcommand{\prob}[1]{\mathbb{P}\left(#1\right)}

\newcommand{\cB}{{\cal{B}}}

\providecommand{\setthms}[1]{#1}
\setthms{
\newtheorem{lem}{Lemma}[section]

\newtheorem{con}[lem]{Conjecture}

\theoremstyle{definition}

\newtheorem*{rem}{Remark}
}

\newcommand{\iid}{\mathrm{i.i.d.}}

\newcommand{\pois}[1]{\mathrm{Poisson}\param{{#1}}}

\newcommand{\bs}{\backslash}

\numberwithin{equation}{section}

\def\bsplit#1\esplit{\begin{split} #1 \end{split} }
\def\splitb#1\splite{\begin{split} #1 \end{split} }
\def\beq#1\eeq{\begin{equation} #1 \end{equation}}
\def\eqb#1\eqe{\begin{equation} #1 \end{equation}}